\documentclass[aos]{imsart}

\RequirePackage{amsthm,amsmath,amsfonts,amssymb}
\RequirePackage[authoryear]{natbib}

\usepackage[utf8]{inputenc}
\usepackage{float}
\usepackage{bbm}
\usepackage{amssymb}
\usepackage{amsmath}
\usepackage{amsthm}
\usepackage{graphicx}
\usepackage{mathabx}
\usepackage{mathtools}
\usepackage{mathrsfs, enumerate}
\usepackage{appendix}
\usepackage[normalem]{ulem}
\usepackage[colorinlistoftodos,prependcaption]{todonotes}
\usepackage{tikz}
\usepackage{tikzscale}
\usetikzlibrary{shapes,arrows}
\usetikzlibrary{decorations.pathreplacing}

\usepackage{xr}
\makeatletter
\newcommand*{\addFileDependency}[1]{
  \typeout{(#1)}
  \@addtofilelist{#1}
  \IfFileExists{#1}{}{\typeout{No file #1.}}
}
\makeatother

\startlocaldefs

\graphicspath{{./figs/}}

\theoremstyle{plain}
\newtheorem{thm}{Theorem}[section]
\newtheorem{prop}{Proposition}[section]

\newtheorem{remark}{Remark}[section]
\newtheorem{example}{Example}[section]
\newtheorem{lemma}{Lemma}[section]
\newtheorem{defi}{Definition}[section]

\renewcommand{\citet}{\cite}

\DeclareMathAlphabet\EuRoman{U}{eur}{m}{n}
\SetMathAlphabet\EuRoman{bold}{U}{eur}{b}{n}

\newcommand{\cC}{\mathcal{C}}

\newcommand{\cF}{\mathcal{F}}
\newcommand{\cG}{\mathcal{G}}
\newcommand{\cH}{\mathcal{H}}
\newcommand{\cN}{\mathcal{N}}

\newcommand{\cR}{\mathcal{R}}
\newcommand{\cS}{\mathcal{S}}
\newcommand{\cA}{\mathcal{A}}

\newcommand{\tcR}{\widetilde{\mathcal{R}}}

\newcommand{\RR}{\mathbb{R}}
\newcommand{\EE}{\mathbb{E}}

\newcommand{\PP}{\mathbb{P}}
\newcommand{\one}{\mathbf{1}}
\newcommand{\ep}{\varepsilon}

\newcommand{\bS}{\mathbb{S}}

\newcommand{\hc}{\hat{c}}
\newcommand{\hw}{\widehat{w}}

\newcommand{\hE}{\widehat{E}}

\newcommand{\tX}{\tilde{\mat X}}
\newcommand{\tE}{\widetilde{E}}
\newcommand{\tOmega}{\widetilde{\Omega}}

\newcommand{\RSS}{\textnormal{RSS}}
\newcommand{\FDR}{\textnormal{FDR}}

\newcommand{\hFDP}{\widehat{\FDP}}
\newcommand{\FDP}{\textnormal{FDP}}
\newcommand{\TPP}{\textnormal{TPP}}
\newcommand{\TPR}{\textnormal{TPR}}
\newcommand{\DP}{\textnormal{DP}}
\newcommand{\kn}{{\textnormal{Kn}}}
\newcommand{\ckn}{{\textnormal{cKn}}}

\newcommand{\BH}{{\textnormal{BH}}}

\newcommand{\App}{Appendix}

\newcommand{\simiid}{\overset{\textnormal{i.i.d.}}{\sim}}
\newcommand{\eqd}{\overset{d}{=}}

\newcommand{\sgn}{{\textnormal{sgn}}}

\newcommand{\pth}[1]{\left( #1 \right)}
\newcommand{\br}[1]{\left[ #1 \right]}

\newcommand{\vct}[1]{\boldsymbol{#1}}
\newcommand{\mat}[1]{\boldsymbol{#1}}

\newcommand{\norm}[1]{\left\| #1 \right\|}
\newcommand{\set}[1]{\left \{  #1 \right \}}

\newcommand{\ceil}[1]{\left \lceil #1 \right \rceil}

\newcommand{\Supp}{\textnormal{Supp}}

\newcommand{\mim}{\wedge}
\newcommand{\mam}{\vee}
\newcommand{\setcomp}{\mathsf{c}}

\newcommand*{\tran}{{\mkern-1.5mu\mathsf{T}}}
\newcommand\independent{\protect\mathpalette{\protect\independenT}{\perp}}
\def\independenT#1#2{\mathrel{\rlap{$#1#2$}\mkern2mu{#1#2}}}

\endlocaldefs

\begin{document}

\begin{frontmatter}
\title{Improving knockoffs with conditional calibration}
\runtitle{Improving knockoffs with conditional calibration}

\begin{aug}
\author[A]{\fnms{Yixiang}~\snm{Luo}\ead[label=e1]{yixiangluo@berkeley.edu}},
\author[B]{\fnms{William}~\snm{Fithian}\ead[label=e2]{wfithian@berkeley.edu}}
\and
\author[C]{\fnms{Lihua}~\snm{Lei}\ead[label=e3]{lihualei@stanford.edu}}
\address[A]{Department of Mathematics, University of California, Berkeley, USA\printead[presep={,\ }]{e1}}
\address[B]{Department of Statistics, University of California, Berkeley, USA\printead[presep={,\ }]{e2}}
\address[C]{Graduate School of Business, Stanford University, USA\printead[presep={,\ }]{e3}}
\end{aug}

\begin{abstract}
The knockoff filter of \citet{barber15} is a flexible framework for multiple testing in supervised learning models, based on introducing synthetic predictor variables to control the false discovery rate (FDR). Using the conditional calibration framework of \citet{fithian2022conditional}, we introduce the {\em calibrated knockoff procedure}, a method that uniformly improves the power of any fixed-X or model-X knockoff procedure. We show theoretically and empirically that the improvement is especially notable in two contexts where knockoff methods can be nearly powerless: when the rejection set is small, and when the structure of the design matrix in fixed-X knockoffs prevents us from constructing good knockoff variables. In these contexts, calibrated knockoffs even outperform competing FDR-controlling methods like the (dependence-adjusted) Benjamini--Hochberg procedure in many scenarios.
\end{abstract}



\end{frontmatter}

\section{Introduction}

The Gaussian linear regression model is one of the most versatile and best-studied models in statistics, with myriad applications in experimental analysis, causal inference, and machine learning. In modern applications, there are commonly many explanatory variables, and we suspect that most of them have little to do with the response, i.e. that the true coefficient vector is (approximately) {\em sparse}. In such problems, multiple hypothesis testing methods are a natural tool for discovering a small number of variables with nonzero coefficients among numerous noise variables, while controlling some error measure such as the false discovery rate (FDR), introduced by \citet{benjamini1995controlling}.

At present, however, the multiple testing literature offers practitioners little clarity regarding {\em how} they ought to perform the inference. There are at least two well-known methods for multiple testing with FDR control: the {\em knockoff filter} of \citet{barber15, candes2018panning} and the Benjamini--Hochberg (BH) procedure of \citet{benjamini1995controlling} (recently modified by \citet{fithian2022conditional} to ensure provable FDR control in linear regression among other problems with dependent $p$-values). Knockoffs and BH use radically different approaches and can return very different rejection sets on the same data, and it is not uncommon for one method to dramatically outperform the other, depending on the problem context. For example, Example~\ref{ex:mcc-block} illustrates a simple problem setting where BH has much higher power at FDR level $\alpha = 0.05$, but the knockoff filter recovers and outperforms BH at level $\alpha = 0.2$. 
In particular, the knockoff filter suffers from a so-called {\em threshold phenomenon}, explained in Section~\ref{subsec:knockoff}, that makes it nearly powerless when the number of discernibly non-null variables is smaller than $1/\alpha$, making it a risky choice for an analyst who aims for more stringent FDR control. In problems with enough rejections to avoid this issue, however, the knockoff filter often excels, since it can use efficient estimation methods like the lasso \citep{tibshirani1996regression} to guide its prioritization of variables.

In this work, we propose a new method, the {\em calibrated Knockoff procedure} (cKnockoff), which uniformly improves the knockoff filter's power while achieving finite-sample FDR control. Our method augments the rejection set of any knockoff procedure using a ``fallback test'' for each variable that is not selected by knockoffs. To set the power of the fallback tests without violating FDR control, we use the conditional calibration framework proposed in \citet{fithian2022conditional}. With exactly the same assumptions and data, cKnockoff is more powerful than knockoffs in every problem instance, but the power gain is especially large in problems with a small number of non-null variables, resolving the threshold phenomenon while retaining the knockoff filter's advantages.


\subsection{Multiple testing in the Gaussian linear model} \label{sec:multi_test}

We consider the linear model relating an observed response vector $\vct y = (y_1, \ldots, y_n)^\tran$ to explanatory variables $\vct X_j = (X_{1j}, \ldots, X_{nj})^{T}$, for $j=1,\ldots,m$ via
\begin{equation}\label{eq:linear-model}
\vct y = \sum_{j=1}^{m} \vct X_j \beta_j + \vct \ep = \mat X \vct \beta + \vct \ep, \quad \vct \ep \sim \cN(\vct 0, \sigma^2 \vct I_n),
\end{equation}
where the design matrix $\mat X \in \RR^{n\times m}$ has $\vct X_j$ as its $j$th column. Both the coefficient vector $\vct \beta = (\beta_1, \ldots, \beta_m)^\tran$ and the error variance $\sigma^2$ are assumed to be unknown. The parameters $\vct \beta$ and $\sigma^2$ are identifiable provided that $n > m$ and $X$ has full column rank.

A central inference question in this model is whether a given variable $\vct X_j$ helps to explain the response, after adjusting for the other variables. Formally, we will study the problem of testing the hypothesis $H_j: \beta_j = 0$ for each variable $\vct X_j$ simultaneously, while controlling the FDR.

Let $\cH_0 = \{j: H_{j}\mbox{ is true}\}$ be the index set of true null hypotheses and $m_0=|\cH_0|$ be the number of nulls; we say $\vct X_j$ is a {\em null variable} if $j\in\cH_0$. For a multiple testing procedure with rejection set $\cR\subset \{1, \ldots, m\}$, the {\em false discovery proportion} (FDP) and FDR are defined respectively as 
\[ \FDP(\cR) = \frac{|\cR \cap \cH_0|}{|\cR| \vee 1}, \quad \FDR = \EE[\FDP]. \]
We write $R = |\cR|$ and $V=|\cR\cap\cH_0|$ for the number of rejections and false rejections respectively. Our goal is to control FDR at a pre-specified threshold $\alpha$ while achieving a power as high as possible. Throughout this paper we define power in terms of the {\em true positive rate} (TPR), defined as the expectation of the {\em true positive proportion} (TPP), the fraction of the $m_1 = m-m_0$ non-null hypotheses rejected:
\[
\TPP(\cR) = \frac{|\cR \cap \cH_0^\setcomp|}{m_1}, \quad \TPR = \EE[\TPP].
\]

\subsection{Fixed-X Knockoff methods}\label{subsec:why-knockoffs}

A traditional approach to multiple testing would start with the usual two-sided $t$-test statistics $|\hat\beta_j| / \hat\sigma$, which are calculated from the ordinary least squares (OLS) estimator and the unbiased estimator of the error variance
\[
\hat{\vct \beta} = (\mat X^\tran \mat X)^{-1} \mat X^\tran \vct y, \quad \text{ and } \hat\sigma^2 = \RSS/(n-m),
\]
where $\RSS = \|\vct y - \mat X\hat{\vct \beta}\|_2^2$ is the residual sum of squares.  Then an appropriate multiplicity correction is applied to their corresponding $p$-values. The celebrated {\em Benjamini--Hochberg procedure} (BH), the best-known FDR-controlling method, orders the $p$-values from smallest to largest $p_{(1)} \leq \cdots \leq p_{(m)}$, and rejects
\[
\cR^{\BH} = \set{j: p_j \leq \frac{\alpha R^{\BH}}{m}}, \quad \text{ where } \;\; R^{\BH} = \max \left\{r:\;p_{(r)} \leq \frac{\alpha r}{m}\right\}.
\]
While BH does not provably control FDR in this context due to the dependence between $p$-values, a corrected version called the {\em dependence-adjusted BH procedure} (dBH) does, while achieving nearly identical power \citep{fithian2022conditional}.

The knockoff filter, described below in Section~\ref{subsec:knockoff}, is a flexible class of methods that bypass the usual $p$-values and instead introduce a ``knockoff'' variable $\tX_j$ to serve as a negative control for each real predictor variable $\vct X_j$, and then apply a learning algorithm to rank the $2m$ variables according to some importance measure in the model. The knockoffs are constructed to ensure that, under $H_j$, $\vct X_j$ and $\tX_j$ are indistinguishable in an appropriate sense. To allow for the construction of the additional $m$ knockoff variables, knockoffs require $n \geq 2m$.

Knockoff methods enjoy substantially higher power than BH and dBH in some scenarios while struggling in others, with the relative performance depending on the problem dimensions, the structure of the design matrix, and the true $\beta$ vector, among other considerations, as we see next. 




\begin{example}[MCC-Block]\label{ex:mcc-block}
Consider a pharmaceutical company jointly analyzing $K$ independent experiments. Each experiment compares a treatment group for each of $G$ different treatments to a common control group to estimate a treatment effect for the $g$-th treatment in the $k$-th experimental condition. Assume that all observations are independent with $\cN(0,\sigma^2)$ errors, and all groups have a common sample size $r$. When $K=1$, the problem reduces to the classical multiple comparisons to control (MCC) problem \citep{dunnett1955multiple}.
\end{example}

By appropriately defining $\mat X$ and $\vct y$, this problem can be equivalently expressed as a linear model of the form \eqref{eq:linear-model} with coefficient vector $\vct\beta$ representing the $G \cdot K$ treatment effects arranged into a single vector. In this formulation, the OLS estimator $\hat{\vct \beta}$ represents the difference in sample means between the treatment groups and their corresponding control groups, and its covariance is block-diagonal, with off-diagonal correlation $0.5$ within each block. The corresponding $t$-tests for entries of $\vct \beta$ likewise coincide with the standard two-sample $t$-tests; see \App~\ref{app:mcc-details} for details.

If $\vct \beta$ is sparse, it may be possible to exploit this sparsity to improve on the OLS estimator by using the {\em lasso estimator} of \citet{tibshirani1996regression}, defined by
\begin{equation}\label{eq:lasso-def}
\hat{\vct \beta}^\lambda = \underset{\vct \beta\in\RR^m}{\text{argmin}} \;\;\frac 12 \norm{\vct y - \mat X \vct \beta}_2^2 + \lambda \cdot \norm{\vct \beta}_1.
\end{equation}

Just as the lasso may improve on OLS as an estimator, we may also find that the knockoffs method based on the lasso improves on the BH method based on the $t$-tests. Indeed, we do observe this improvement in Figure~\ref{fig:mcc-block}, but only for $\alpha = 0.2$. For smaller $\alpha$ values, a specific drawback of knockoffs --- ironically, that knockoff methods break down when the coefficient vector is {\em too} sparse --- prevents the method from realizing its potential. This drawback is resolved by our calibrated knockoff method, the main subject of this work.

\begin{figure}[!tb]
    \centering
    \includegraphics[width=\linewidth]{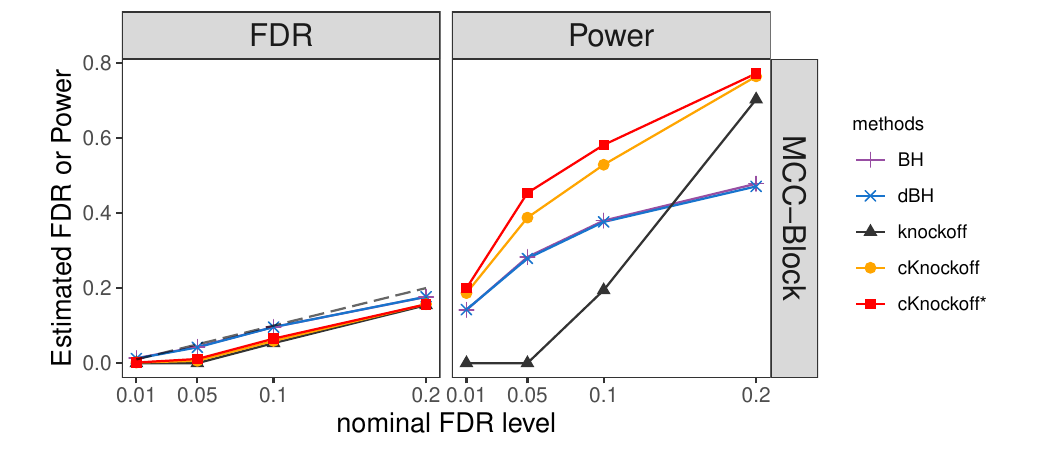}
    \caption{Performance of several multiple testing methods for the MCC-Block problem with $K=200$, $G=5$, $r=3$, $m_1 = 10$, and signal-to-noise ratio $\textnormal{SNR} = 0.05$
    for varying FDR significance levels, $\alpha = 0.01, 0.05, 0.1,$ and $0.2$. The knockoff method using lasso-based LCD feature statistics outperforms BH by a wide margin when $\alpha = 0.2$ by successfully exploiting sparsity, but it fails for smaller values of $\alpha$ due to the threshold phenomenon. The cKnockoff method described in this paper outperforms both BH and knockoffs.}
    \label{fig:mcc-block}
\end{figure}

\subsection{Multiple testing in high-dimensional regression}

Whereas the fixed-X knockoffs framework requires $n \geq 2m$, many modern regression problems of interest are high-dimensional. An important body of literature in the last decade has developed methods for marginal inference on individual parameters in such high-dimensional models, typically under additional assumptions such as sparsity of $\vct \beta$ or knowledge of the distribution of $\mat X$ \citep{zhang2014confidence, van2014asymptotically, javanmard2014confidence, ning2017general, shah2020hardness, shah2023double}. 

\citet{candes2018panning} extended the knockoffs paradigm to allow for multiple testing in high-dimensional settings. This approach requires knowledge of the distribution of the covariate $X$, but makes no model assumptions about the conditional distribution of the outcome $y | X$. These so-called {\em model-X knockoffs} allow for testing the conditional independence,
$H_j: y \independent X_j \mid X_{-j}$,
with guaranteed finite-sample FDR control, assuming the pairs of response and covariates are independent and identically distributed (i.i.d.).
As we will see, our cKnockoff method can also be applied to uniformly improve the power of model-X knockoffs as well.



\subsection{Outline and contributions} \label{sec:outline}

In this work, we propose the {\em calibrated Knockoff procedure} (cKnockoff), a method that controls finite-sample FDR under the same assumption as knockoffs. Our method augments the rejection set of any knockoff procedure, uniformly improving its power by means of a {\em fallback test} that allows for selection of variables not selected by knockoffs. 

Our fallback test takes a very simple form. For any reasonable test statistic $T_j$ for the significance of $X_j$, our framework will compute a data-adaptive threshold $\hc_j$, and reject $H_j$ for any $j$ which has $T_j \geq \hc_j$, in addition to the knockoff rejection set. That is,
\[
\cR^{\ckn} \,=\, \cR^{\kn} \,\cup\, \{j:\; T_j \geq \hc_j\},
\]
where $\cR^{\kn}$ and $\cR^{\ckn}$ are respectively the rejection sets for the baseline and calibrated knockoff methods.

The calculation of $\hc_j$ involves a lot of mathematical details though. As a high-level description, we figure out the gap between the FDR ``budget'' $\alpha$ and the actual FDR of knockoffs. Then the budget gap is distributed to each variable via $\hc_j$ through a conditional FDR calibration method of \citet{fithian2022conditional}, which we review in more detail in Section~\ref{subsec:conditional-calibration}. Define the statistic
\begin{equation}\label{eq:Sj-cases}
S_j = \begin{cases}
(\mat X_{-j}^\tran \vct y,\; \|\vct y\|_2^2) & \text{ for fixed-X cKnockoffs}\\
(\mat X_{-j}, \;\vct y) & \text{ for model-X cKnockoffs}
\end{cases},
\end{equation}
which is sufficient for the submodel described by the corresponding null hypothesis $H_j$ under each modeling framework. Then under $H_j$, the distribution of $\vct y$ given $S_j$
is known so that, for any fallback test threshold $c_j$, we can calculate the resulting {\em conditional FDR contribution} of $H_j$, defined as
\[
\FDR_j(\cR^{\ckn} \mid S_j) 
\;\coloneqq\; \EE_{H_j}\br{ \frac{\one \set{j \in \cR^{\ckn}}}{\left| \cR^{\ckn} \right| \mam 1} \;\Big|\; S_j }
\,\leq\, \EE_{H_j}\br{ \frac{\one \set{j \in \cR^\kn} \mam \one \set{T_j \geq c_j}}{\left| \cR^{\kn} \cup \{j\}\right|} \;\Big|\; S_j }.
\]
We will choose the threshold $\hc_j(S_j)$ to set the last expression equal to the FDR budget that we distribute to variable $j$.

Because the cKnockoff rejection set always includes the knockoff rejection set and sometimes exceeds it, the method is uniformly more powerful than knockoffs. We find in simulations that the power gain is especially large when the true $\vct \beta$ vector is so sparse that the number of strong signal variables does not exceed $1/\alpha$.

The only downside of cKnockoff is the additional computation it requires. In the fixed-X settings, to reduce this burden, we only carry out the fallback test on hypotheses that appear promising, and we use a conservative algorithm to speed the fallback test calculation. We prove that these speedup techniques do not inflate the FDR, and we find numerically that the computation time of our implementation is a small multiple of the knockoff computation time, which is further improved when parallel computing is available.

Section~\ref{sec:review} reviews the basics of knockoffs and conditional calibration, and Section~\ref{sec:our-method} defines our method in full detail. Section~\ref{sec:implement} gives more detail about how we implement the cKnockoff efficiently.  Sections~\ref{sec:simu}--\ref{sec:hiv} illustrate our method's performance on selected simulation scenarios, in addition to the HIV data from the original knockoff paper \citep{barber15}, and Section~\ref{sec:discussion} concludes.

\section{Review: knockoffs and conditional calibration}\label{sec:review}

\subsection{Knockoffs: a flexible framework}\label{subsec:knockoff}

Within both the fixed-X and model-X knockoff frameworks, there are a variety of ways to generate the knockoff matrix $\tX = (\tX_1, \ldots, \tX_m) \in \RR^{n \times m}$ \citep{barber15, candes2018panning, spector2022powerful}, but every knockoff method yields intermediate outputs called {\em feature statistics} $W_1, \ldots, W_m \in \RR$. The absolute value $|W_j|$ quantifies the overall importance of the pair $\{\vct X_j, \tX_j\}$, while the sign $\sgn(W_j)$ indicates whether the original variable $\vct X_j$ is considered more important than its knockoff $\tX_j$.

The two most popular feature statistics, proposed by \citet{barber15} and \citet{candes2018panning} respectively, are both based on the lasso estimator for the augmented model:
\[
\hat{\vct \beta}^\lambda = \underset{\vct \beta \in \RR^{2m}}{\text{argmin}} \;\;\frac 12 \norm{\vct y - \mat X_+ \vct \beta}_2^2 + \lambda \cdot \norm{\vct \beta}_1,
\]
where $\mat X_+ = (\mat X,\tX) \in \RR^{n \times 2m}$ is the augmented design matrix.

The {\em lasso signed-max} (LSM) statistics are defined by entry points on the regularization path:
\begin{equation}\label{eq:LSM}
W_j^{\text{LSM}} = (\lambda_j^* \mam \lambda_{j+m}^*) \cdot \sgn(\lambda_j^* - \lambda_{j+m}^*), \quad \text{ for } \;\;\lambda_j^* = \sup \set{\lambda:\; \hat\beta_j^\lambda \neq 0},
\end{equation}
and the {\em lasso coefficient-difference} (LCD) statistics are defined by the estimator for fixed $\lambda$:
\begin{equation}\label{eq:LCD}
W_j^{\text{LCD}} = |\hat\beta_j^\lambda| - |\hat\beta_{j+m}^\lambda|.
\end{equation}
If $\beta_j^\lambda = \beta_{j+m}^\lambda = 0$, then $W_j^{\text{LCD}} = 0$. For the simulations in this paper, we use a minor modification of $\vct W^{\text{LCD}}$, {\em LCD with tiebreaker} (LCD-T), that breaks ties using the variables' correlations with the lasso residuals $\vct r^{\lambda} = \vct y - \mat X_+ \hat{\vct \beta}^{\lambda}$.
See \App~\ref{app:LCD-T} for details.

Knockoff methods' FDR control guarantee arises from a crucial stochastic property of the feature statistics: conditional on their absolute values $|\vct W|=(|W_1|,\ldots,|W_m|)$, the signs for null variables are independent Rademacher random variables \citep{barber15, candes2018panning}:
\begin{equation}\label{eq:iid-rademacher}
\sgn(W_j) \,\mid\, |W_j|,\, \vct W_{-j} \;\stackrel{H_j}{\sim}\; \textnormal{Unif}\{-1,+1\},
\end{equation}
where $\vct W_{-j}$ encodes all entries other than $W_j$. To avoid trivialities, we assume all $W_j \neq 0$ and $W_i \neq W_j$ for any $i \neq j$.

Once the feature statistics are calculated, knockoff rejects $H_j$ if $W_j \geq \hw$, for an adaptive rejection threshold $\hw \geq 0$ that is based on a running estimator of $\FDP$:
\begin{equation}\label{eq:fdphat-kn}
\hw = \min\left\{w \geq 0:\; \hFDP(w) \leq \alpha\right\}, \quad \text{for } \;\hFDP(w) = \frac{1 + |\{j:\; W_j \leq -w\}|}{|\{j:\; W_j \geq w\}|},
\end{equation}
where $\hw=\infty$ (no rejections) if $\hFDP(w) > \alpha$ for all $w$. Let $\cR^{\kn} = \{W_j \geq \hw\}$ denote the knockoff rejection set. This rejection rule controls FDR at level $\alpha$ whenever the feature statistics satisfy \eqref{eq:iid-rademacher}, as Section~\ref{subsec:finding-budgets} discusses in detail. \App~\ref{app:knockoff} gives further details about the construction of knockoffs, including a proof of \eqref{eq:iid-rademacher} for fixed-X knockoffs.



\subsection{Two limitations of knockoffs}\label{subsec:limitations}

Despite its deft exploitation of sparsity, the knockoff method has two major limitations that can inhibit its performance in certain settings. One limitation, reflected in Figure~\ref{fig:mcc-block}, is the so-called {\em threshold phenomenon}: because the denominator of $\hFDP_{w}$ is the size of the candidate rejection set, we cannot make any rejections at all unless we have $R \geq 1/\alpha$. For example, if $\alpha = 0.1$, we must make at least $10$ rejections or none at all, even if several $p$-values lie well below the Bonferroni threshold $\alpha/m$. Even when the number of potential rejections is above the $1/\alpha$ threshold, the FDP estimator can be highly variable and upwardly biased, adversely affecting the method's power and stability. Some recent proposals ameliorate this limitation \citep{gimenez2019improving, emery2019controlling, nguyen2020aggregation, sarkar2021adjusting}. However, they often come at the cost of substantial power loss compared to standard knockoffs. By contrast, our calibrated knockoff method is always more powerful than the baseline knockoffs.


A second issue is the {\em whiteout phenomenon} in fixed-X knockoffs discussed by \citet{li2021whiteout}, 
who prove that in large MCC problems (Example~\ref{ex:mcc-block} with $K=1$), or more general problems in which the design matrix has an unfavorable eigenstructure, even the Bonferroni method is dramatically more powerful than the best possible knockoff method. 
As we will see, calibrated knockoffs partially address this issue, delivering high power under some circumstances, but giving limited performance gains in other settings.


\subsection{Conditional calibration and dBH}\label{subsec:conditional-calibration}

\citet{fithian2022conditional} introduced a novel technique called {\em conditional calibration} for proving and achieving FDR control under dependence. They begin by decomposing the FDR into the contributions from each null hypothesis:
\begin{equation}\label{eq:fdr-decomp}
\FDR(\cR) \,=\, \sum_{j \in \cH_0} \FDR_j(\cR), \quad \text{ where } \;\;
\FDR_j(\cR) \,=\, \EE_{H_j}\left[\frac{\one\{j \in \cR\}}{R \mam 1}\right],
\end{equation}
and propose controlling each FDR contribution at level $\alpha/m$, so that $\FDR \leq \alpha m_0 / m \leq \alpha$.

Just as we decomposed the FDR into the contributions from each null hypothesis, we can likewise decompose the FDP as
\begin{equation}
    \FDP \,=\, \sum_{j \in \cH_0} \DP_j, \quad \text{ where }\;\; \DP_j(\cR) \,=\, \frac{\one\{j \in \cR\}}{R \mam 1} \,=\, \frac{\one\{j \in \cR\}}{\left| \cR\cup \{j\}\right|}.
\end{equation}
We will call $\DP_j$ the {\em realized discovery proportion} for $H_j$; then $\FDR_j = \EE_{H_j} [\DP_j]$. Note that $\DP_j$ only contributes to $\FDP$ if $H_j$ is true, but it is a well-defined statistic whether $H_j$ is true or false.

To control $\FDR_j$ at level $\alpha/m$, \citet{fithian2022conditional} first condition on a sufficient statistic $S_j$ for the submodel described by $H_j$. 
By the sufficiency of $S_j$, the conditional expectation $\EE_{H_j}[\DP_j(\cR) \mid S_j]$ can be calculated for any rejection rule $\cR$; as a result, a rejection rule with a tuning parameter can also be {\em calibrated} to control the conditional expectation at $\alpha/m$.



\section{Our method: calibrated knockoffs}\label{sec:our-method}

\subsection{Conditional calibration for knockoffs} \label{sec:cknockoff}


Our cKnockoff method supplements a given knockoff procedure's rejection set with additional rejections from a fallback test. Formally,
\begin{equation} \label{eq:ckn}
    \cR^{\ckn} \,=\, \cR^{\kn} \,\cup\, \{j:\; T_j \geq \hc_j\},
\end{equation}
where $T_j$ is any test statistic and $\hc_j$ is a data-adaptive threshold that is calibrated to achieve FDR control. We describe $T_j$ and $\hc_j$ in detail next.

\emph{Ingredient 1: The fallback test statistics.} The test statistic $T_j$ can be chosen arbitrarily by the analyst. To simplify the presentation, we assume by convention that $T_j$ is non-negative and continuously distributed, and the test rejects for large values of $T_j$. Our own implementation uses the correlation of $\mat X_j$ with the residuals from an $L_1$-penalized regression model (see Section \ref{sec:ckn_details}).

\emph{Ingredient 2: The fallback test threshold.} The rejection threshold $\hc_j$ is calibrated to control the FDR contribution of variable $j$. If we set the threshold at $c_j$, then
\[
\DP_j(\cR^\ckn) \,=\, \frac{\one \set{j \in \cR^\kn} \mam \one \set{T_j \geq c_j}}{\left| \cR^\ckn \cup \{j\}\right|}
\,\leq\, \frac{\one \set{j \in \cR^\kn} \mam \one \set{T_j \geq c_j}}{\left| \cR^\kn \cup \{j\}\right|}.
\]
We set $c_j$ to control the conditional expectation $\EE_{H_j}[\DP_j \mid S_j]$ at a data-dependent ``budget'':
\begin{equation}\label{eq:cond_bound}
    \EE_{H_j} \br{ \frac{\one \set{j \in \cR^\kn} \mam \one \set{T_j \geq c_j}}{\left| \cR^\kn \cup \{j\}\right|} \;\Big|\; S_j } \,\leq\, \textrm{Budget}_j(S_j),
\end{equation}
where $S_j$ is defined in \eqref{eq:Sj-cases}. Thus we define $\hc_j(S_j)$ as the minimal value $c_j \in [0,\infty]$ that satisfies \eqref{eq:cond_bound}. The budget for variable $j$ is defined in turn as the conditional expectation $\EE_{H_j}\left[b_j \mid S_j\right]$, where $b_j$ is defined in \eqref{eq:bj_def_ea}. As explained in Section \ref{subsec:finding-budgets}, $b_j$ is constructed to satisfy two conditions:
\begin{equation}\label{eq:budget-conditions}
(\text{i})\;\;\DP_j(\cR^{\kn}) \leq b_j \text{ almost surely, and } (\text{ii})\;\;\sum_{j \in \cH_0}  \EE\left[\,b_j\,\right] \leq \alpha.
\end{equation}
The first condition ensures that $c_j=\infty$ is a feasible solution for the inequality \eqref{eq:cond_bound}, and the second ensures our method's FDR control (Theorem \ref{thm:fdr}). 




Because $S_j$ is a sufficient statistic for the submodel described by $H_j$, the conditional data distribution given $S_j$ is fully known under $H_j$ (see \App~\ref{app:cond_dist}) and the conditional expectations in \eqref{eq:cond_bound} are computable for all $c_j$.

With these ingredients in place, we are ready to define our procedure:

\begin{defi}[cKnockoff] \label{def:ckn_fixed}
    Suppose the knockoff procedure employs the feature statistics $W_1, \ldots, W_m$ and rejects $\cR^\kn$ at FDR level $\alpha$.
    
    Let the budget
    \begin{equation} \label{eq:bj_def_ea}
        b_j = \alpha\, \frac{\one\{j \in \cC(w^*)\}}{1+|\cA(w^*)|},
    \end{equation}
    where $\cC(w) = \{j:\; W_j \geq w\}$, $\cA(w) = \{j:\; W_j \leq -w\}$, and
    \[
    w^* = \min\left\{w \geq 0: \; \frac{1+|\cA(w)|}{|\cC(w)|} \leq \alpha\right\} \mim \min\left\{w:\; |\cC(w)| < 1/\alpha\right\}.
    \]
    The cKnockoff procedure rejection set is defined as
    \[ \cR^{\ckn} \,=\, \cR^{\kn} \,\cup\, \{j:\; T_j \geq \hc_j\}, \]
    where $T_j$ is any non-negative test statistic and $\hc_j$ is the minimal value $c_j \in [0,\infty]$ satisfying \eqref{eq:cond_bound}.
\end{defi}

The variable $b_j$ is crafted to allow our method to exploit the difference between the FDR budget $\alpha$ and the true FDR of a given knockoff method. We emphasize that, while our method is more powerful than knockoffs by construction, it requires no additional modeling assumptions, and controls FDR at the same advertised level, as we see next.

\begin{thm} \label{thm:fdr}
    cKnockoff controls FDR at level $\alpha$. That is, $\FDR(\cR^{\ckn}) \leq \alpha$.
\end{thm}

\begin{proof}
  By construction, $\cR^\ckn \supseteq \cR^\kn$, so 
  \begin{align*}
  \EE_{H_j}[\DP_j(\cR^\ckn) \mid S_j]
  &\,=\, \EE_{H_j} \br{ \frac{\one \set{j \in \cR^\kn} \mam \one \set{T_j \geq \hc_j}}{\left| \cR^\ckn \cup \{j\}\right|} \;\Big|\; S_j } \\[7pt]
  &\,\leq\, \EE_{H_j} \br{ \frac{\one \set{j \in \cR^\kn} \mam \one \set{T_j \geq \hc_j}}{\left| \cR^\kn \cup \{j\}\right|} \;\Big|\; S_j } 
  \,\leq\, \EE_{H_j}[b_j \mid S_j].
  \end{align*}
  For the last inequality, note that $\hat c_j$ is degenerate given $S_j$. Marginalizing over $S_j$ and applying condition (ii) in \eqref{eq:budget-conditions}, we have
  \[
  \FDR(\cR^\ckn) \,=\, \sum_{j\in\cH_0} \EE [\,\DP_j(\cR^\ckn)\,] \,\leq\, \sum_{j\in\cH_0} \EE\left[\,b_j\,\right] \,\leq\, \alpha.
  \]
\end{proof}



\begin{remark}
The same calibration scheme can be applied to any baseline FDR-controlling method $\cR$ as long as we can find budgets $b_1,\ldots,b_m$ satisfying \eqref{eq:budget-conditions}. While we have assumed $\DP_j(\cR) \leq b_j$, it is enough to have $\EE_{H_j}[\DP_j(\cR) - b_j\mid S_j] \leq 0$ almost surely.
\end{remark}

\subsection{Important properties}\label{subsec:ckn-properties}

We can avoid directly calculating $\hc_j$ by subtracting the right-hand side of \eqref{eq:cond_bound} to obtain an inequality for the {\em excess FDR contribution} of variable $j$:
\begin{equation}\label{eq:Ej}
E_j(c\,; S_j) \;\coloneqq\; \EE_{H_j} \br{\frac{\one \set{j \in \cR^\kn} \mam \one \set{T_j \geq c}}{\left| \cR^\kn \cup \{j\}\right|} \,-\, b_j \;\Big|\; S_j} \,\leq\, 0.
\end{equation}
The function $E_j(\cdot\,; S_j)$ is continuous and non-increasing in $c$, with $E_j(\hat c_j; S_j) = 0$ by definition. As a result, we have $T_j \geq \hat{c}_j$ if and only if $E_j(T_j; S_j)\leq 0$. \footnote{To avoid a possible source of confusion, we emphasize here that $E_j(T_j; S_j)$ means evaluating the function $E_j(\cdot\,; S_j):\;\RR \to \RR$ at the realized value $c=T_j$--- {\em not} substituting the expression $T_j$ for $c$ inside the integrand in \eqref{eq:Ej} and calculating the conditional expectation.} We thus obtain an equivalent, but more computationally useful, definition of calibrated knockoffs as
\[
\cR^{\ckn} \,=\, \cR^{\kn} \cup \{j:\; E_j(T_j; S_j) \leq 0\}.
\]

Another important property of our method for both implementation and theory is that we can filter the fallback test rejections arbitrarily without threatining FDR control, as stated formally in Theorem~\ref{thm:sandwich}. This is nontrivial because, in general, the FDR can increase after filtering the rejection set; see e.g. \citet{katsevich2021filtering}.

\begin{thm}[Sandwich] \label{thm:sandwich}
    For any rejection rule $\cR$ with $\cR^{\kn} \subseteq \cR \subseteq \cR^{\ckn}$ almost surely, we have $\FDR(\cR) \leq \alpha$.
\end{thm}

\begin{proof}
Recall
$
\cR^{\ckn} \,=\, \cR^{\kn} \cup \{j:\; T_j \geq \hc_j\}.
$
Hence $\cR^{\kn} \subseteq \cR \subseteq \cR^{\ckn}$ implies 
\[
\EE_{H_j}[\DP_j(\cR) \mid S_j]
\,\leq\, \EE_{H_j} \br{ \frac{\one \set{j \in \cR^\kn} \mam \one \set{T_j \geq \hc_j}}{\left| \cR^\kn \cup \{j\}\right|} \;\Big|\; S_j }\\[5pt]
\,\leq\, \EE_{H_j}[b_j \mid S_j],
\]
so that $\FDR(\cR) \leq \sum_{j\in \cH_0} \EE[b_j] \leq \alpha$.
\end{proof}

The Sandwich property allows us to reduce computational effort by carrying out the fallback tests only on variables that seem promising (e.g. because they have small marginal $p$-values) but are not selected by knockoffs. 
Further details on the fast and reliable implementation of cKnockoffs can be found in Section~\ref{sec:implement}, \App~\ref{app:implement}, \ref{app:filtering}, and \ref{app:R_star}.

\subsection{Test  statistic in cKnockoff}\label{sec:ckn_details}

Our implementation of cKnockoff uses the test statistic
\begin{equation}\label{eq:Tj}
T_j = \left|\vct X_j^\tran \left(\vct y - \hat{\vct y}^{(j)}\right)\right|,
\end{equation}
where $\hat{\vct y}^{(j)}$ is a vector of fitted values from $L_1$-penalized regression of $\vct y$ on $\mat X_{-j}$ with regularization parameter $\lambda^{(j)}$.
For fixed-X cKnockoff with the Gaussian linear model, setting $\lambda^{(j)}=0$ results in $\vct y-\hat{\vct y}^{(j)}$ being the OLS residuals under $H_j$. Holding $S_j$ fixed, $T_j$ is an increasing function of the OLS $t$-statistic's absolute value. In this sense, \eqref{eq:Tj} generalizes the usual two-tailed $t$-statistic. We opt for $\lambda^{(j)}>0$ because, in sparse settings, lasso-fitted values will likely yield a more accurate adjustment for other predictor variables' effects. For model-X cKnockoff, the same idea extend to other regression models.

It is computationally convenient for $\hat{\vct y}^{(j)}$ to be a function of $S_j$ only. In the model-X setting, we select $\lambda^{(j)}$ via cross-validation of $L_1$-penalized regression of $\vct y$ on $\mat X_{-j}$, and in the fixed-X setting we set $\lambda^{(j)} = 2\hat{\sigma}^{(j)}$, \footnote{Setting $\lambda = 2\sigma$ blocks most (95\%) null variables under orthogonal designs, where the explanatory variables
are standardized to have a unit norm.} where $\hat{\sigma}^{(j)}$ is the residual variance from regressing $\vct y$ on $\mat X_{-j}$.


\subsection{Finding budgets}\label{subsec:finding-budgets}


In this section, we explain why the budget $b_j$ satisfies the two conditions in \eqref{eq:budget-conditions}. 
Define the knockoffs {\em candidate rejection set} for threshold $w$ as $\cC(w) = \{j:\; W_j \geq w\}$, and let $\cA(w) = \{j:\; W_j \leq -w\}$. Then the FDP estimate may be expressed as
\[
\hFDP(w) \,=\, \frac{1+|\cA(w)|}{|\cC(w)|}.
\]
Let $w_1 < \cdots < w_m$ denote the order statistics of $|W_1|,\cdots,|W_m|$. It suffices to restrict our attention to these order statistics because they are the only values of $w$ where $\cC(w)$ or $\hFDP(w)$ change. Then we can equivalently write the knockoff rejecting threshold
\[
\hw = w_\tau, \quad \text{ where } \;\; \tau = \min\left\{t \in \{1,\ldots,m+1\}: \; \hFDP(w_t) \leq \alpha\right\},
\]
where we set $w_{m+1}=\infty$ and $\hFDP(\infty) = 0$ to cover the case where no rejections are made. Thus we can view knockoffs as a stepwise algorithm with discrete ``time'' index $t$. It calculates $\hFDP(w_t)$ for each $t = 1, 2, \ldots$, and stops and rejects $\cC(w_t)$ the first time $\hFDP(w_t) \leq \alpha$.

The FDR control proof for knockoffs is based on an optional stopping argument. Define
\[
    M_t \;\coloneqq\; \frac{|\cC(w_t) \cap \cH_0|}{1+|\cA(w_t) \cap \cH_0|}
    \;\geq\; \frac{|\cC(w_t) \cap \cH_0|}{1+|\cA(w_t)|}
    \,=\, \frac{\FDP(\cC(w_t))}{\hFDP(w_t)},
\]
where we set the last expression to be zero by convention if $\cC(w_t) = \emptyset$. \citet{barber15} show $M_t$ is a super-martingale with respect to the discrete-time filtration given by
\[
\cF_t = \sigma\Big(\,|\vct W|, \;\left(W_j: j \in \cH_0^\setcomp \text{ or } |W_j| < w_t\right), \; |\cC(w_t)|\,\Big), \quad \text{ for } t = 1,\ldots,m+1,
\]
and they also show that $\EE [M_1] \leq 1$. We include proofs of both facts in \App~\ref{app:kn-proofs} for completeness. Because $\tau$ is a stopping time with respect to the same filtration, we have a chain of inequalities
\begin{equation} \label{eq:kn_chain}
    \FDR(\cR^\kn)
    \,=\, \EE[\FDP(\cC(w_\tau))]
    \,\leq\, \alpha\, \EE\left[ \frac{\FDP(\cC(w_\tau))}{\hFDP(w_\tau)} \right]
    \,\leq\, \alpha\, \EE M_\tau
    \,\leq\, \alpha\,\EE M_1
    \,\leq\, \alpha.
\end{equation}
To find large budgets whose sum is controlled at $\alpha$ in expectation, the intermediate expressions in \eqref{eq:kn_chain} are natural places to look. While we cannot calculate $\alpha M_\tau$ or $\alpha M_1$ without knowing $\cH_0$, we can decompose the next largest expression to obtain the budgets
\[
b_j^0 \,=\, \alpha\, \frac{\DP_j(\cC(w_\tau))}{\hFDP(w_\tau)} \,=\, \alpha\, \frac{\one\{j \in \cC(w_\tau)\}}{1+|\cA(w_\tau)|},
\quad \text{ since } \;\; \sum_{j\in\cH_0} b_j^0 \,=\, \alpha\,\frac{\FDP(\cC(w_\tau))}{\hFDP(w_\tau)}.
\]
These budgets satisfy (i) in \eqref{eq:budget-conditions} because $\hFDP(w_\tau) \leq \alpha$ almost surely, and (ii) in \eqref{eq:budget-conditions} by the inequalities in \eqref{eq:kn_chain}.

The budgets $b_j^0$ do yield a small improvement over baseline knockoffs by taking up the slack in the first inequality of \eqref{eq:kn_chain}, but they do not resolve the main failure modes we discussed in Section~\ref{subsec:limitations}, where knockoffs usually makes no rejections. In that case, most of the slack is in the second-to-last inequality, since $\EE[M_{\tau}] \approx 0$ while $\EE[M_1]$ may be close to $1$.

We can find better budgets by considering the algorithm's behavior in realizations where no rejections are made. As soon as $|\cC(w_t)|$ falls below $1/\alpha$, it becomes a foregone conclusion that $\tau = m+1$ and $M_\tau = 0$, even while the current value of $M_t$ may still be fairly large. In such ``hopeless'' cases, we should stop the algorithm early and harvest as much of $M_t$ as we can. We thus replace $\tau$ with a stopping time $\tau_1$ that halts early in these cases:
\begin{equation}\label{eq:bj_def}
b_j \,=\, \alpha\, \frac{\DP_j(\cC(w_{\tau_1}))}{\hFDP(w_{\tau_1})} \,=\, \alpha\, \frac{\one\{j \in \cC(w_{\tau_1})\}}{1+|\cA(w_{\tau_1})|},
\quad \text{ for } \;\;  \tau_1 = \tau \mim \min\left\{t:\; |\cC(w_t)| < 1/\alpha\right\}.
\end{equation}
We have $b_j \geq b_j^0 \geq \DP_j(\cR^{\kn})$ almost surely since $\tau_1 = \tau$ unless $\cR^{\kn} = \emptyset$. We also have
\[ \sum_{j\in\cH_0} b_j = \alpha \frac{\FDP(\cC(w_{\tau_1}))}{\hFDP(w_{\tau_1})} \leq \alpha M_{\tau_1}, \]
whose expectation is below $\alpha$ by optional stopping. As a result, we have the following lemma:
\begin{lemma} \label{lem:bj}
    The budgets defined in~\eqref{eq:bj_def} satisfy
    \[ (\text{i})\;\;\DP_j(\cR^{\kn}) \leq b_j \text{ almost surely,} \quad \text{and } (\text{ii})\;\;\sum_{j \in \cH_0}  \EE\left[\,b_j\,\right] \leq \alpha. \]
\end{lemma}
Note the definition here is equivalent to \eqref{eq:bj_def_ea} since $w_{\tau_1} = w^*$.

To illustrate the improvement of $b_j$ over $b_j^0$, consider the problem setting in Figure \ref{fig:mcc-block}. Averaging over $100$ simulations with $\alpha = 0.05$, we estimate $\sum_{j \in \cH_0} \EE[b_j] \approx 0.99\alpha$, while $\sum_{j \in \cH_0} \EE[b^0_j] \approx 0$. While the increase is not always so dramatic, $b_j$ always yields a strict power improvement over fixed-X knockoffs, as we see next.

\begin{thm}[strict improvement over fixed-X knockoffs] \label{thm:strict_better}
    Assume $\cH_0^\setcomp \neq \emptyset$. Let $T_j$ be \eqref{eq:Tj} and the budget be defined as in \eqref{eq:bj_def} and the nominal FDR level $\alpha \in (0, 0.5]$. Then, in the fixed-X setting, we have
    \[ \TPR(\cR^\ckn) > \TPR(\cR^\kn). \]
\end{thm}

\begin{remark}
    The same result holds in the model-X setting, provided that the statistics $T_1,\ldots,T_m$ are continuous and supported on $[0,\infty)$.
\end{remark}

In short, the theorem follows from the fact that our fallback test always makes each hypothesis strictly more likely to be rejected than the knockoffs, due to the construction of $b_j$ in \eqref{eq:bj_def}. We defer the detailed proof to \App~\ref{app:strict_better}. It is worth noticing that although theoretically, the null hypotheses are also more likely to be rejected, the realized FDR is almost the same as knockoffs in our simulation studies in Section~\ref{sec:simu}, even when the power gain is significant. This is because most hypotheses rejected by the fallback test are non-null.

\subsection{Refined cKnockoff procedure}\label{subsec:refine}

Here we briefly discuss how to extend the analysis above to further improve cKnockoffs with additional computational effort, if desired. The proof of Theorem \ref{thm:fdr} indicates that $\cR^\kn$ in the denominator in \eqref{eq:cond_bound} can be replaced by any $\cR^*$ satisfying $\cR^\kn \subseteq \cR^* \subseteq \cR^\ckn$ to obtain an even more powerful procedure. In particular, we could use $\cR^* = \cR^\ckn$ and apply the calibration scheme recursively; this would be an example of {\em recursive refinement} as proposed in \citet{fithian2022conditional}. However, the computational cost of recursive refinement may be prohibitive since $\cR^{\ckn}$, which is already a computationally intensive method, becomes part of the integrand. 

A computationally feasible alternative is for $\cR^*$ to augment $\cR^{\kn}$ only with a set of very promising variables whose inclusion in $\cR^{\ckn}$ can be quickly verified. Informally, we use
\[ \cR^* \,=\, \cR^\kn \cup \left\{j:\; E_j(T_j; S_j) \leq 0, \text{ and } p_j \text{ is tiny}\right\} \;\subseteq\; \cR^{\ckn}, \]
where $p_j$ is the $p$-value from the standard two-sided $t$-test, a computationally cheap substitute for $T_j$. We defer our exact formulation of $\cR^*$ and the implementation in fixed-X cKnockoff to \App~\ref{app:R_star}; we do not implement $\cR^*$ in the model-X setting due to computational complexity.

Using $\cR^*$ leads to the {\em refined calibrated knockoff} (cKnockoff$^*$) procedure rejecting
\begin{equation} \label{eq:ckn_star}
\cR^{\ckn^*} \,=\, \cR^\kn \cup \left\{j:\; T_j \geq \hc_j^*\right\} \,=\, 
\cR^\kn \cup \left\{j:\; E_j^*(T_j; S_j) \leq 0\right\},
\end{equation}
where
\begin{equation}\label{eq:Ej_star}
E_j^*(c\,; S_j) \,=\,\EE_{H_j} \br{\frac{\one \set{j \in \cR^\kn} \mam \one \set{T_j \geq c}}{\left| \cR^* \cup \{j\}\right|} - b_j \;\Big|\; S_j} \,\leq\, E_j(c\,; S_j),
\end{equation}
and $\hc_j^* = \min \set{c:\; E_j^*(c\,; S_j) \leq 0} \leq \hc_j$.

cKnockoff$^*$ controls FDR and is uniformly more powerful than cKnockoff, as we show in \App~\ref{app:R_star}. However, as a price of handling its additional computational complexity, we will lose the theoretical upper bound of the numerical error in our implementation of cKnockoff$^*$, although simulation studies show the calculation is precise and reliable.

\section{Implementation}
\label{sec:implement}


One straightforward, if computationally costly, implementation of our method would use simple Monte-Carlo sampling to calculate $E_j$ to high precision for every variable not rejected by knockoffs. In practice, we can implement the method much more efficiently using two strategies: first, we filter the non-rejected hypotheses, only carrying out the fallback test on the more promising ones; and second, in place of $E_j$ we instead evaluate a computationally cheaper upper bound $\tE_j \geq E_j$ and reject $H_j$ when $\tE_j \leq 0$.

In fixed-X settings, we leverage the Gaussian linear model and its simple sufficient statistics to implement these strategies. Figure \ref{fig:algorithm} illustrates the resulting procedure, which we explain in steps next. Since Theorem~\ref{thm:sandwich} guarantees FDR control for any method $\cR$ with $\cR^{\kn} \subseteq \cR \subseteq \cR^{\ckn}$, these speedups preserve theoretical FDR control, and we find in practice that they have little impact on the power. 
In model-X settings, implementation of the strategy depends on specific model assumptions, and the random nature of the knockoff-generating process further complicates the implementation. We use straightforward Monte-Carlo sampling for the model-X cKnockoff as an illustration.

\begin{figure}[!tb]
    \centering
    \includegraphics[width=\linewidth]{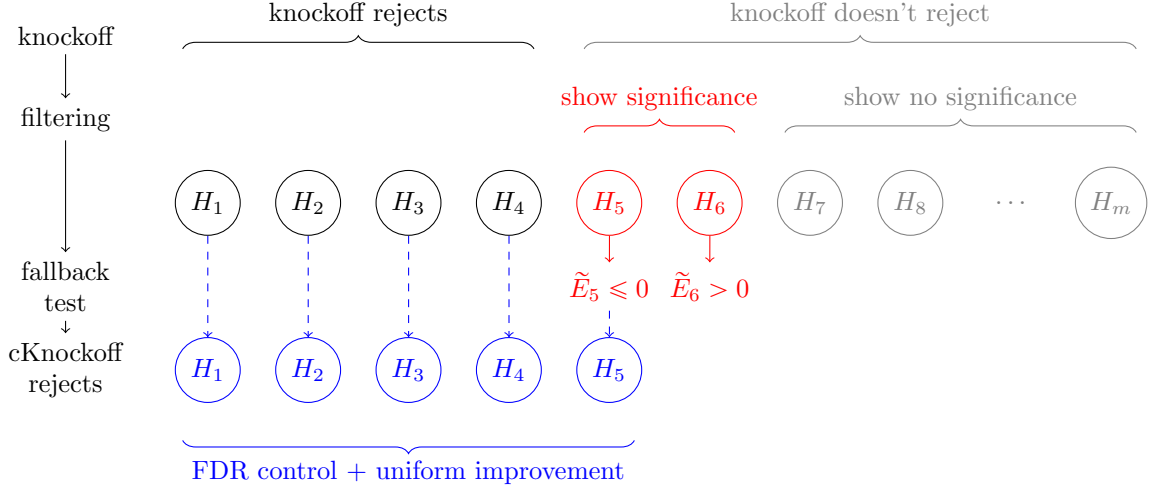}
    \caption{Steps in implementing cKnockoffs.}
    \label{fig:algorithm}
\end{figure}

\emph{Step 1. filtering.} 
Our filtering process aims to save computational resources by carrying out fallback tests only on those hypotheses with a relatively good chance of being rejected by our fallback test, which compares a (modified) $t$-statistic to a threshold whose value is closely related to the FDR budget $b_j$. In practice, the largest budgets are typically allocated to non-rejected variables with large $|W_j|$ values, so a variable is most promising when it {\em either} has a very large $t$-statistic {\em or} has a large $|W_j|$ value. Our implementation, detailed in \App~\ref{app:filtering}, selects variables with either of these properties.

\emph{Step 2. fallback test: finding upper-bound for $E_j$.}
A key obstacle in efficiently determine the sign of conditional expectation $E_j$ is that the integrand may be zero over most of the support of the conditional distribution. To speed up Monte Carlo sampling, we express the integrated, denoted as $f_j$, as a difference of two non-negative functions $f_j^{(T)} - f_j^{(b)}$, where the support of $f_j^{(T)}$ is small and easy to calculate, and $f_j^{(b)}$ can be truncated to a small and easy-to-calculate approximation of its own support. Focusing our Monte Carlo samples on the union of these two supports gives a more computationally efficient method. Since the truncated function $\tilde{f}_j^{(b)}$ is uniformly smaller than $f_j^{(b)}$, we have
\[
\tilde{f}_j = f_j^{(T)} - \tilde{f}_j^{(b)} \geq f_j,
\]
and its integral $\tE_j$ is correspondingly larger than the original $E_j$. We implement the fallback test conservatively by rejecting when $\tE_j \leq 0$; see \App~\ref{app:implement} for further details.

\emph{Step 3. fallback test: safe early stopping Monte-Carlo.}
We use Monte-Carlo methods to compute $\tE_j$ and treat checking $\tE_j \leq 0$ as an online testing problem. 
Using \citet{waudby2020estimating}'s method, we construct shrinking confidence intervals for $\tE_j$, stopping when the interval excludes 0. 



We reiterate here that neither the filtering nor the conservative fallback test on $\tE_j$ threatens our theoretical FDR control guarantee. Though Monte-Carlo error could possibly inflate FDR, this inflation is bounded and we found it negligible in our simulation studies. We discuss the impact of Monte-Carlo error further in \App~\ref{app:error_control}.

\section{Numerical studies in fixed-X settings}
\label{sec:simu}

In this section, we provide selected experiments that compare fixed-X cKnockoff with competing procedures. Extensions of these simulations under other settings can be found in \App~\ref{app:simu}.

\subsection{FDR and TPR}
\label{sec:simu_main}

We show simulations on the following design matrices $\mat X \in \RR^{n \times m}$ with $m = 1000$ and $n = 3000$.
\begin{enumerate}
    \item \textbf{IID normal}: $X_{ij} \simiid \cN(0, 1)$.
    \item \textbf{MCC}: the setting in Example~\ref{ex:mcc-block} with $G = 1000$ and $K = 1$.
    \item \textbf{MCC-Block}: the setting in Example~\ref{ex:mcc-block} with $G = 5$ and $K = 200$.
\end{enumerate}

The response vector is generated as:
\[ \vct y = \mat X \vct \beta + \vct \ep, \quad \vct \ep \sim \cN(\vct 0, \mat I_{n}),\]
where $\beta_j = \beta^*$ for all $j \in \cH_0^\setcomp$. The signal strength $\beta^*$ is calibrated to ensure the BH procedure $\cR^{\BH}$, with nominal FDR level $\alpha = 0.2$, achieves a power of $\TPR = 0.5$ under the given design matrix, yielding moderate signal strength. The alternative hypotheses set $\cH_0^\setcomp$ is a random subset of $[m]$ with cardinality $m_1 = 10$, uniformly distributed among all such subsets.

Our experiments compare the following procedures:
\begin{enumerate}
    \item \textbf{BH}: The Benjamini-Hochberg procedure \citep{benjamini1995controlling}. Lacks provable FDR control for our design matrix settings.
    \item \textbf{dBH}: The dependence-adjusted Benjamini-Hochberg procedure introduced by \cite{fithian2022conditional}. We use $\gamma = 0.9$ without recursive refinement. It is similar to BH but with provable FDR control in this context.
    \item \textbf{knockoff}:Fixed-$X$ knockoff method \citep{barber15}.
    \item \textbf{BonBH}: Adaptive Bonferroni-BH method \citep{sarkar2021adjusting}.
    \item \textbf{cKnockoff}: Our method as defined in Section~\ref{sec:cknockoff}.
    \item \textbf{cKnockoff*}: Our refined method using $\cR^*$ as introduced in Section~\ref{subsec:refine}.
\end{enumerate}
For knockoff, cKnockoff, and cKnockoff*, we construct the knockoff matrix using the default semidefinite programming procedure and employ the LCD-T feature statistics.

For each trial, we generate $\mat X$ (or its realization if random), $\vct \beta$, and $\vct y$, then apply all aforementioned procedures. We estimate the FDR and TPR for each procedure by averaging results over 400 independent trials. Figure \ref{fig:main_expr_10} presents these results.

\begin{figure}[!tb]
    \centering
    \includegraphics[width = 0.95\linewidth]{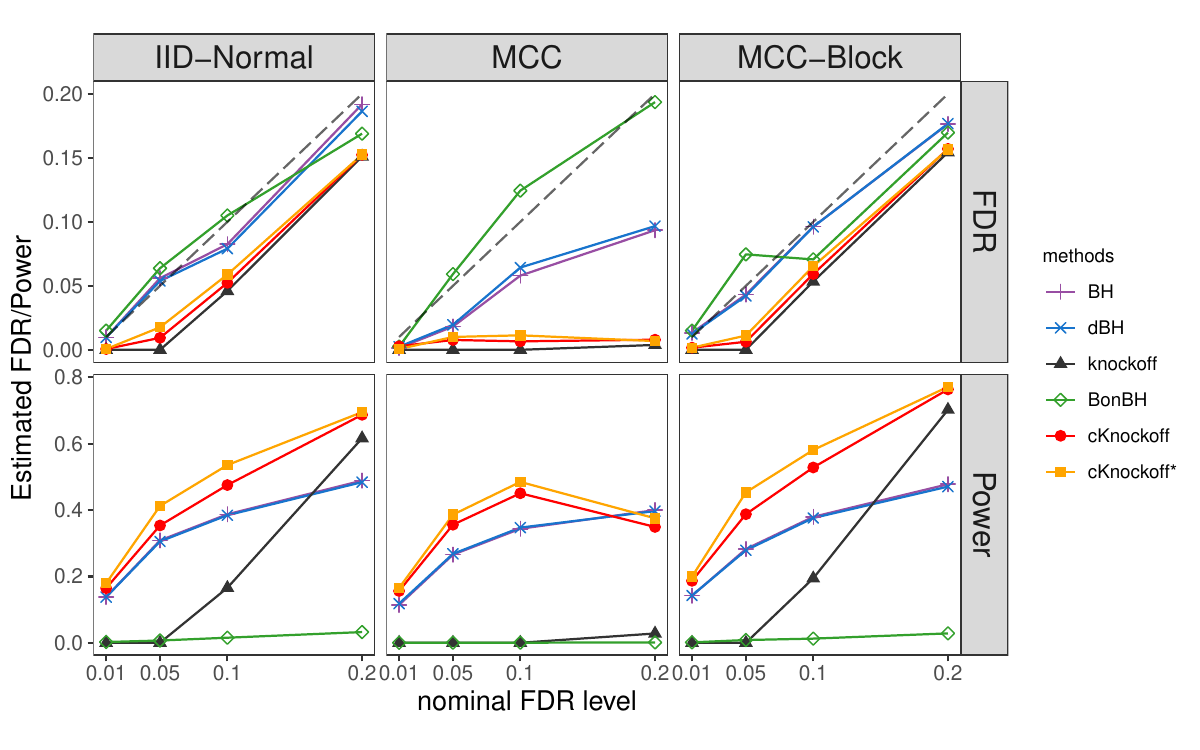}
    \caption{Estimated FDR and TPR under different design matrix settings, with the additional BonBH method of \citet{sarkar2021adjusting} included for comparison. cKnockoff/cKnockoff$^*$ outperforms the other procedures in general.}
    \label{fig:main_expr_10}
\end{figure}

Our observations confirm that cKnockoff controls FDR and outperforms knockoff, aligning with our theoretical predictions. Notably, when knockoff encounters the threshold phenomenon (at small $\alpha$) or the whiteout phenomenon (in MCC problems), cKnockoff and cKnockoff* achieve comparable or even superior correct rejections to BH/dBH on average. Moreover, in scenarios where knockoff performs well, cKnockoff/cKnockoff* demonstrates even better performance.

The readers might be puzzled by the non-monotone power curve for cKnockoff in the MCC case. This is mainly driven by the shrinking advantage of cKnockoff over knockoffs as $\alpha$ increases. The power gain of cKnockoff over knockoffs is mostly given by the realizations for which $\tau_1 < \tau$ (i.e., $\cR^\kn = \emptyset$) and hence $b_j$ is substantially larger than $b_j^0$. This event happens less likely with a larger $\alpha$. Furthermore, when the signal-to-noise ratio is large to the extent that the ordering of knockoff statistics is relatively stable, only the top $O(1/\alpha)$ variables could gain an extra budget, thus limiting the power boost. This heuristic analysis also suggests that cKnockoff/cKnockoff$^*$ only alleviates the whiteout phenomenon to a limited extent because knockoffs, the baseline procedure that cKnockoff wraps around, suffers even when $m_1 \gg 1/\alpha$.
See \App~\ref{app:less_sparse} for a numerical study. 

Multiple knockoffs are not included in this comparison due to their requirement for an aspect ratio larger than 3. However, we present a comparison with multiple knockoffs under different problem settings in \App~\ref{app:multi_kn}. To summarize those results: when $m_1 < 1/\alpha$, multiple knockoffs mitigate the threshold phenomenon but are outperformed by cKnockoff/cKnockoff*. When $m_1 \gg 1/\alpha$, multiple knockoffs prove less powerful than even vanilla knockoffs. Consequently, despite requiring stronger conditions on sample size, multiple knockoffs do not demonstrate competitiveness with our method.

Figure~\ref{fig:power_gain} illustrates the TPR difference between cKnockoff and knockoff at a nominal FDR level of $\alpha = 0.05$ across various design matrices, sparsity levels, and signal strengths. The simulation settings mirror those in Section \ref{sec:simu_main}, with variations in the number of non-null variables $m_1$ and signal strength $\beta^*$. Our method's superiority is most pronounced in sparse problems with strong signal strength, though significant power gains persist even in less favorable conditions. Importantly, Theorem \ref{thm:strict_better} guarantees that this power gain is always positive, regardless of problem sparsity, design matrix characteristics, or the significance of non-null variables.
\begin{figure}[!tb]
    \centering
    \includegraphics[width = 0.95\linewidth]{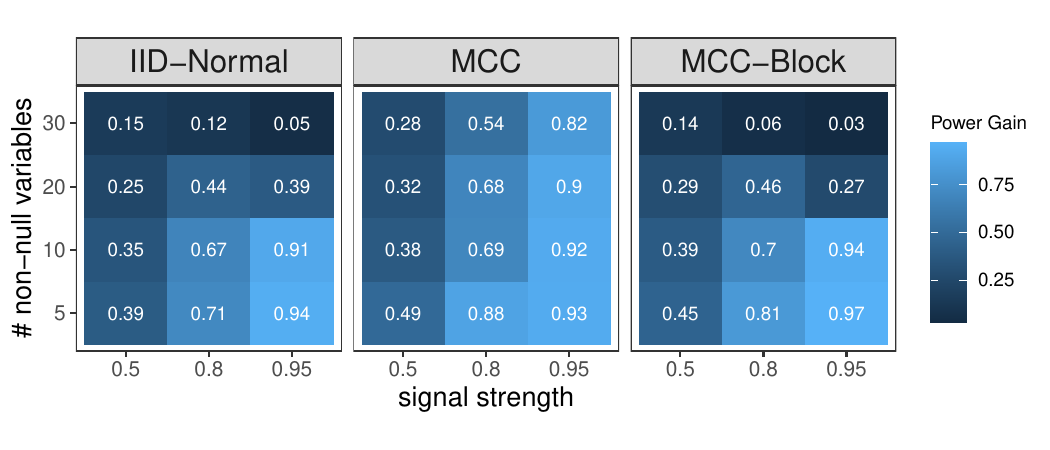}
    \caption{The TPR of cKnockoff minus the TPR of knockoffs in various design matrices, sparsity, and signal strength settings, with $m = 1000$ variables. The three columns each represent a type of design matrix $\mat X$ as described in Section \ref{sec:simu_main}. The $y$ axis is the number of non-null variables, with a larger value representing a less sparse problem. The $x$ axis represents the TPR of the BH procedure if we run it with nominal FDR level $\alpha = 0.2$. A larger value indicates the non-null variables are  more significant.}
    \label{fig:power_gain}
\end{figure}

\subsection{Robustness to non-Gaussian noise}
\label{sec:simu_robust}

To assess our methods' robustness to non-Gaussian outcomes, we examine scenarios with heavy-tailed distributions. We generate $\vct y = \mat X \vct \beta + \vct \ep$, where $\ep_i \simiid t_k$ follows a $t$-distribution with $k$ degrees of freedom. As $k$ decreases, the distribution becomes more heavy-tailed. We set $\alpha = 0.2$, maintaining other simulation settings from Section \ref{sec:simu_main}.
Figure \ref{fig:robust_tNoise} illustrates the FDR of various procedures as the noise distribution deviates further from Gaussian. Notably, our methods demonstrate robustness across all considered cases, consistently maintaining FDR below $\alpha$.

\begin{figure}[!tb]
    \centering
    \includegraphics[width = 0.95\linewidth]{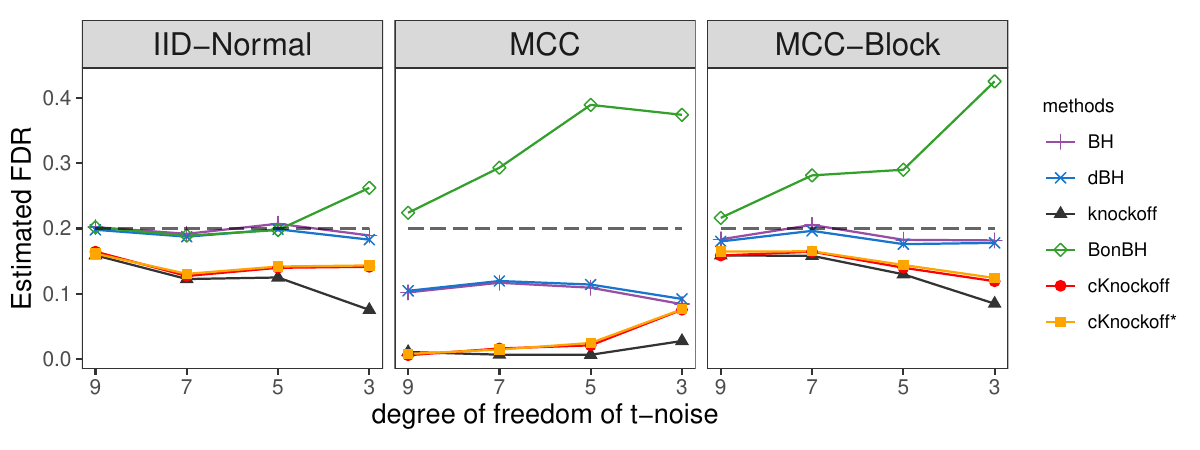}
    \caption{Estimated FDR when the noise is generated from a heavy-tailed $t$-distribution.}
    \label{fig:robust_tNoise}
\end{figure}

\subsection{Scalability}
\label{sec:simu_scale}

The computational complexity of our methods is instance-specific, rendering worst-case analysis uninformative. However, we offer an approximate instance-specific complexity analysis. For each variable examined post-filtering (described in \App~\ref{app:R_star}), the computation is roughly equivalent to running a bounded number of knockoffs with a given knockoff matrix (details in Appendices \ref{app:implement}-\ref{app:R_star}). Let $A$ denote the number of variables post-filtering, $C_{\kn, \mathrm{f}}$ the complexity of knockoffs with a given matrix, and $C_{\kn, \mathrm{m}}$ the complexity of generating a knockoff matrix. Our methods' complexity is thus:
\[O(A \cdot C_{\kn, \mathrm{f}}) + C_{\kn, \mathrm{m}},\]
In contrast, knockoffs' complexity is $O(C_{\kn, \mathrm{f}}) + C_{\kn, \mathrm{m}}$.
Notably, when $A = O(1)$, our method's complexity matches that of knockoffs within a constant factor. Often, $A$ is small because our method only examines a handful of promising variables not rejected by knockoffs: when knockoffs make little or no rejection, we expect $A = O(1/\alpha)$; with powerful knockoffs, typically few promising variables remain unrejected. Furthermore, we can force $A$ to be small by exploiting a more stringent filtering step. 

\begin{figure}[tb]
    \centering
    \includegraphics[width = 0.8\linewidth]{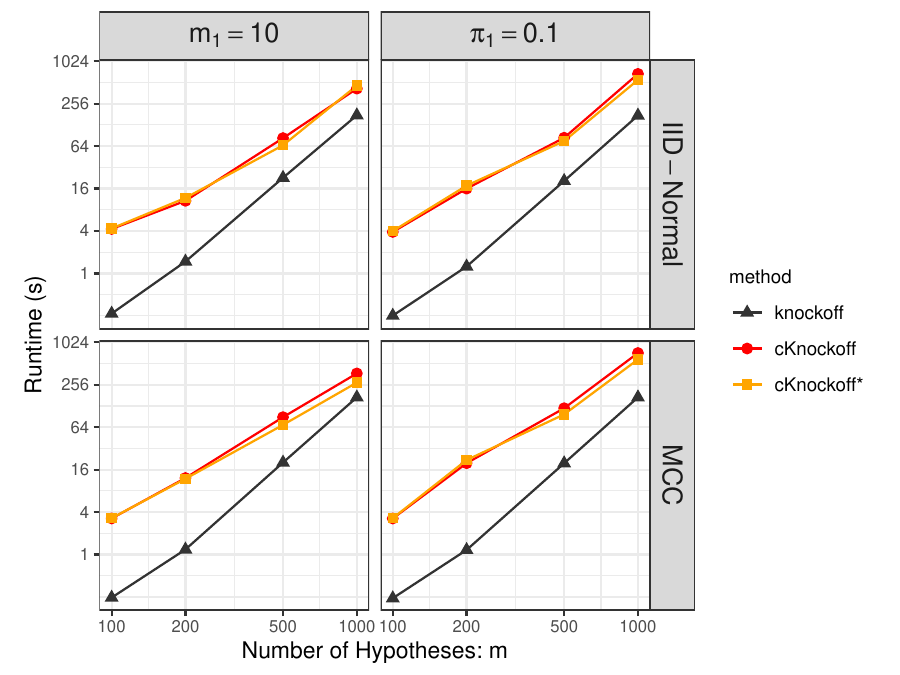}
    \caption{Averaged running time of knockoff, cKnockoff, and cKnockoff$^*$ as the problem size increases. The left panel has a fixed number of true alternatives $=10$ as $m$ increases; the right panel has a fixed proportion of true alternatives $=0.1$. All three methods show a roughly $O(m^3)$ time complexity in the figure. }
    \label{fig:simu-scale_m_expr}
\end{figure}

Figure \ref{fig:simu-scale_m_expr} illustrates the average running times of knockoffs and cKnockoff on a single-core 3.6GHz CPU. The experimental setup mirrors Section~\ref{sec:simu_main} for signal strength, knockoff matrix construction, and feature statistics generation. We set $\alpha = 0.05$ and $n = 3m$, with $m$ ranging from 100 to 1000. The left panel maintains $10$ true alternatives as $m$ increases, while the right panel keeps a constant proportion of true alternatives ($\pi_1 \coloneqq m_1/m = 0.1$).
Consistent with our heuristic complexity analysis, cKnockoff/cKnockoff* computation times are a small multiple of knockoffs across all settings, even on a single core. With multi-core systems, we can further optimize performance by parallelizing the fallback test among different variables. Our \texttt{R} package provides a parallelized implementation.



\section{Numerical demonstration of Model-X cKnockoffs}
\label{sec:simu-modelX}

As discussed in Section~\ref{sec:implement}, efficient implementation of model-X cKnockoff depends on specific model assumptions. For demonstration purposes, we implement model-X cKnockoffs using straightforward Monte-Carlo sampling. This section illustrates the model-X cKnockoff procedure on two smaller-scale, simple problems.

In our simulation, the covariate distribution, denoted as $F_X$, is zero-mean multivariate Gaussian. We consider three types of covariate distributions, mirroring Section \ref{sec:simu_main}: IID normal, MCC, and MCC-Block, where the covariance matrices are the same as its fixed-X counterpart. For example, the ``IID normal'' covariates distribution has $F_X =_d \cN(\vct 0, \mat I_p)$ and the ``MCC'' covariates has $F_X =_d \cN(\vct 0, \mat \Sigma)$, with
\[ 
\mat \Sigma = 
\begin{bmatrix}
1 & - \frac 1m & \ldots & - \frac 1m \\
- \frac 1m & 1 & \ldots & - \frac 1m \\
\vdots & \vdots & \ddots & \vdots \\
- \frac 1m & - \frac 1m & \ldots & 1 \\
\end{bmatrix} \in \RR^{m \times m}.
\]
Unless otherwise specified, other simulation settings remain consistent with Section \ref{sec:simu_main}.

\subsection{high-dimensional linear regression}

We generate the response variable $y$ using the Gaussian linear model
\[ y = \sum_{j = 1}^m X_j \beta_j + \ep, \quad \ep \sim \cN(0, 1).\]
The model includes $m_1 = 5$ non-null variables, with a total of $m=100$ variables exceeding the number of observations $n = 80$. We employ LCD statistics as feature importance statistics, with $\lambda$ determined through cross-validation.
Figure \ref{fig:simu-ModelX_gaussian} demonstrates a significant power improvement of our method over the baseline knockoffs.
\begin{figure}[htb]
    \centering
    \includegraphics[width = 0.95\linewidth]{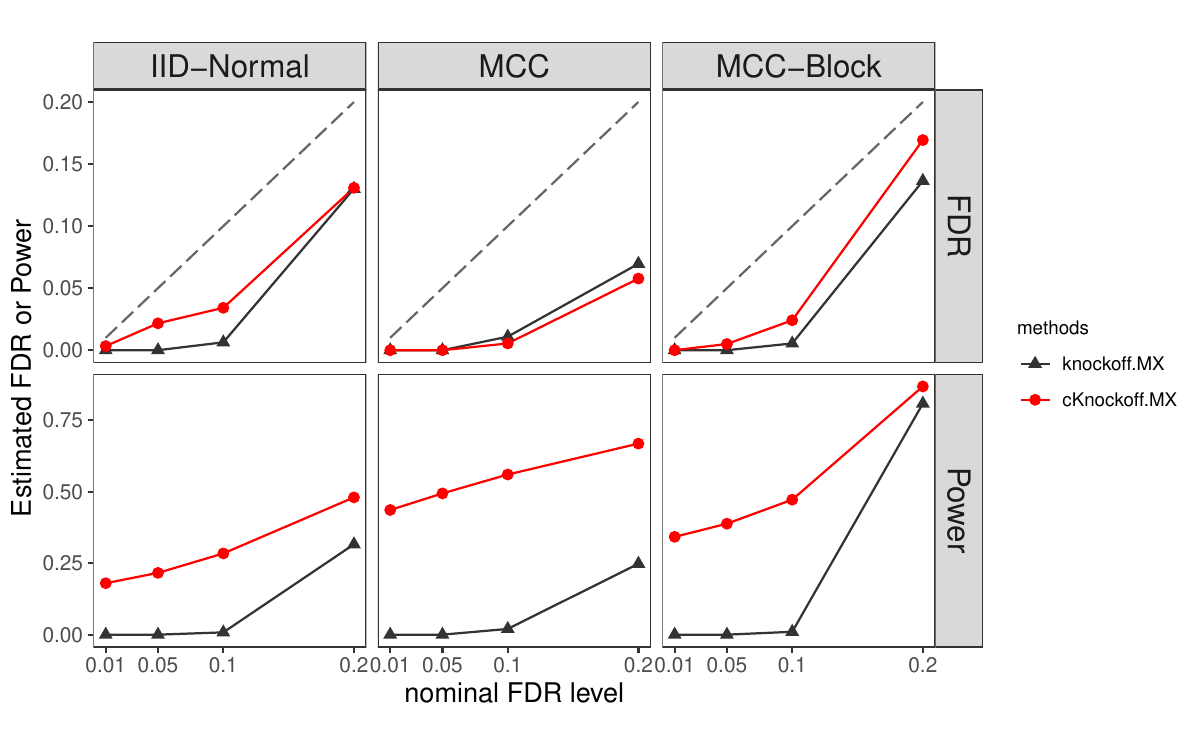}
    \caption{Estimated FDR and TPR under different design matrix settings for model-X knockoffs and cKnockoff in high dimensional linear regression.}
    \label{fig:simu-ModelX_gaussian}
\end{figure}

\subsection{logistic regression}

We consider a logistic regression setting where $y \in \set{0, 1}$ is generated according to:
\[ \ln \pth{\frac{\PP(y = 1)}{1 - \PP(y = 1)}} = \sum_{j = 1}^m X_j \beta_j. \]
The model includes $m_1 = 5$ non-null variables, with $m = 100$ total variables and $n = 300$ observations. We employ a logistic regression version of the LCD statistics as feature importance statistics:
\[ W_j = |\hat\beta_j| - |\hat\beta_{j+m}|, \]
where
\[
\hat{\vct \beta}^\lambda = \underset{\vct \beta \in \RR^{2m}}{\text{argmin}} \;\; \sum_{i = 1}^n y_i \cdot \theta_i - \ln \pth{e^{\theta_i} + 1} + \lambda \cdot \norm{\vct \beta}_1,
\text{ with }
\theta_i = \sum_{j = 1}^{m} \mat X_{ij} \beta_j + \tX_{ij} \beta_{j+m}
\]
is the $L_1$ penalized logistic regression coefficients on the augmented covariates $[\vct X, \tX]$. The regularity parameter $\lambda$ is determined through cross-validation.
Figure \ref{fig:simu-ModelX_logistic} demonstrates a significant power improvement of our method over the baseline knockoffs.
\begin{figure}[htb]
    \centering
    \includegraphics[width = 0.95\linewidth]{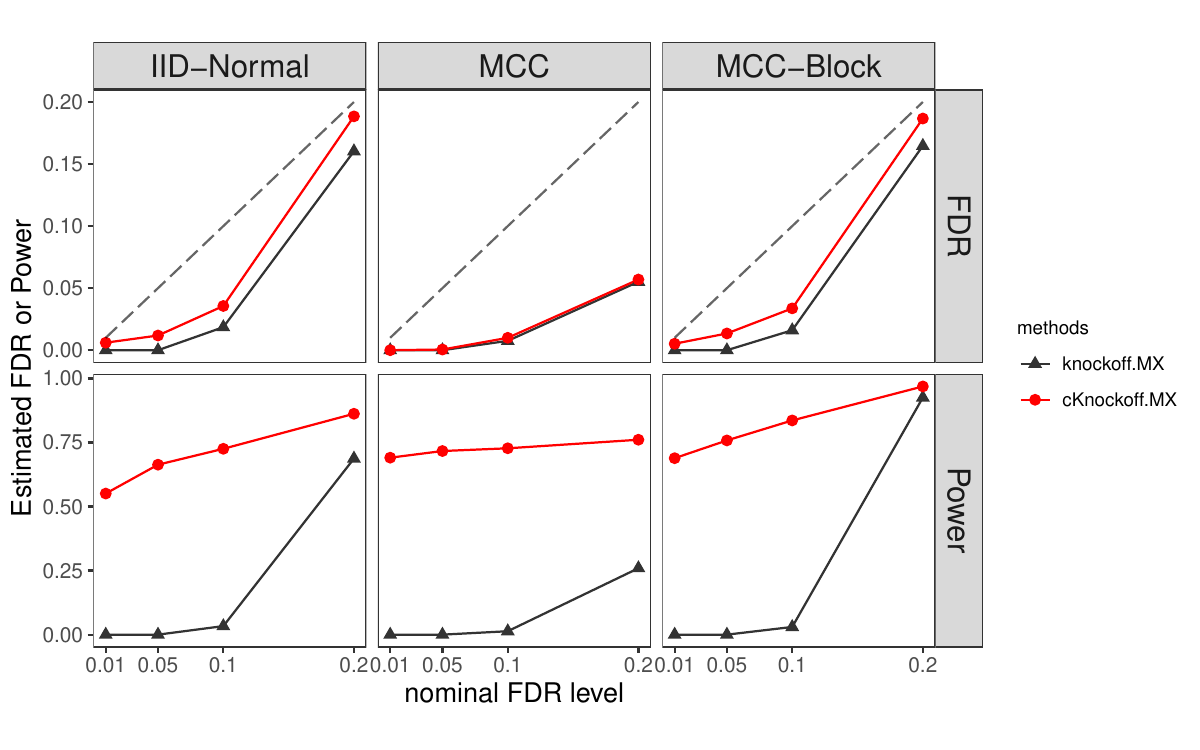}
    \caption{Estimated FDR and TPR under different design matrix settings for model-X knockoffs and cKnockoff in logistic regression.}
    \label{fig:simu-ModelX_logistic}
\end{figure}

\section{HIV drug resistance data}\label{sec:hiv}

This section applies fixed-X cKnockoff and cKnockoff* to identify mutations in the Human Immunodeficiency Virus (HIV) associated with drug resistance \citep{rhee2006genotypic}, following the analysis in \cite{barber15} and \cite{fithian2022conditional}.

The dataset encompasses experimental results for 16 drugs across three categories: protease inhibitors (PIs), nucleoside reverse transcriptase inhibitors (NRTIs), and nonnucleoside reverse transcriptase inhibitors (NNRTIs). Each experiment provides data on genetic mutations and drug resistance measures for a sample of HIV patients.
Following \cite{barber15}, we construct the design matrix using one-hot encoding for mutations, excluding an intercept term. We preprocess the data by discarding mutations occurring fewer than three times and removing duplicated columns. The 16 drug experiments vary in size, with most having sample sizes $n$ between 600 and 800, and genetic mutation counts $m$ between 200 and 300.
Given the absence of ground truth, we assess replicability by comparing selected mutations to those identified in an independent treatment-selected mutation (TSM) panel from \cite{rhee2006genotypic}, as detailed in Section 4 of \cite{barber15}.
For each dataset, we compare the performance of BH, knockoffs, cKnockoff, and cKnockoff* as described in Section \ref{sec:simu_main}.

Figure \ref{fig:HIV-0.05} presents results with $\alpha = 0.05$. The horizontal dashed line indicates the total number of mutations in the TSM list, encompassing scenarios where the number of non-null variables is above, approximately equal to, or below $1 / \alpha$. Knockoffs consistently underperforms, making almost no rejections across all cases. In contrast, our cKnockoff and cKnockoff$^*$ procedures save knockoffs and make a decent number of discoveries. 

\begin{figure}[!tb]
    \centering
    \includegraphics[width = 0.75\linewidth]{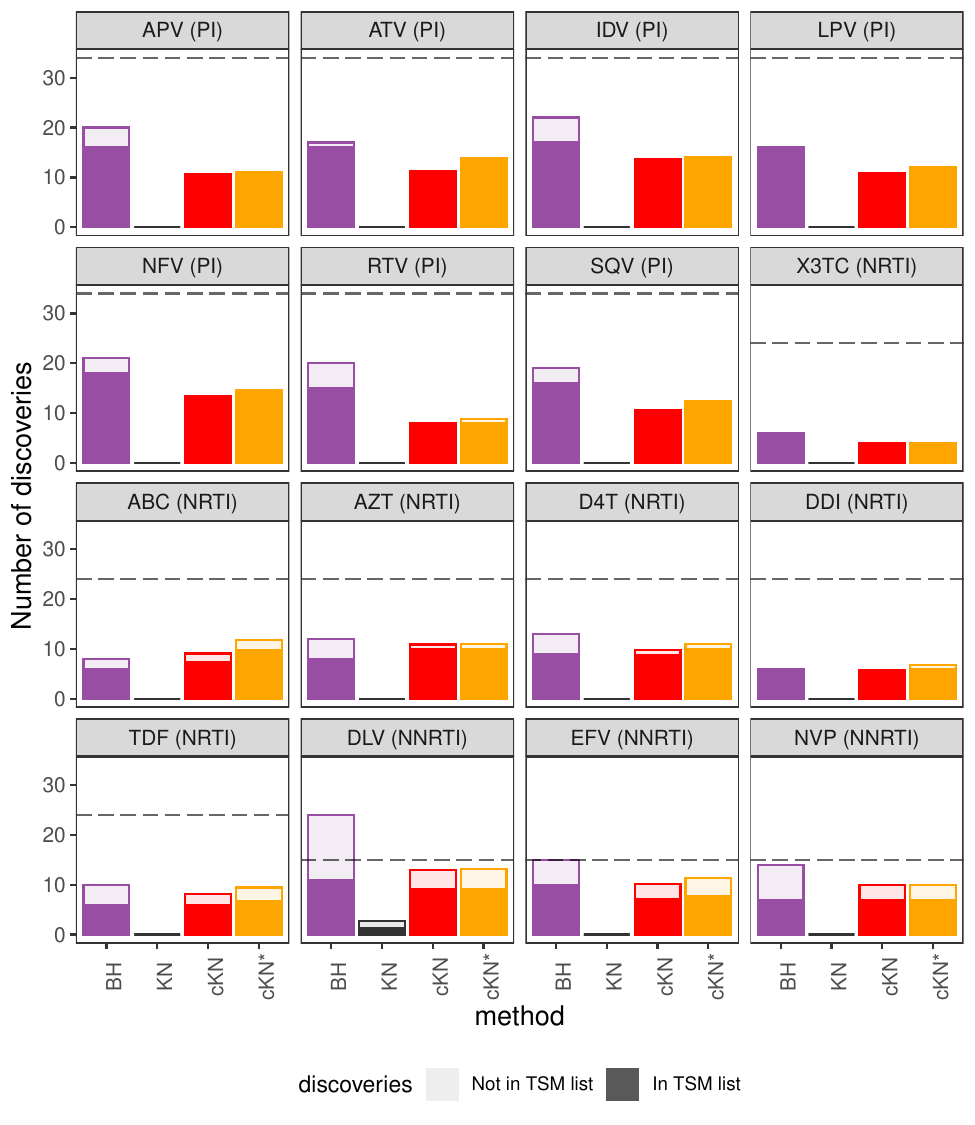}
    \caption{Results on the HIV drug resistance data with $\alpha = 0.05$. The darker segments represent the number of discoveries that were replicated in the TSM panel, while the lighter segments represent the number that was not. The horizontal dashed line indicates the total number of mutations appearing in the TSM list. Results are shown for the BH, fixed-X knockoffs, cKnockoff and cKnockoff$^*$.}
    \label{fig:HIV-0.05}
\end{figure}

Additional results for $\alpha = 0.2$ is included in \App~\ref{app:hiv}, where knockoffs is not powerless in most problems and our methods are able to make more discoveries.

\section{Summary}\label{sec:discussion}


We have presented a new approach, the calibrated knockoff procedure, for simultaneously testing if the explanatory variables are relevant to the outcome. Our cKnockoff procedure controls FDR and is uniformly more powerful than the knockoff procedure. And the power gain is especially large when the unknown $\vct \beta$ vector is very sparse, in particular when the number of nonzeros in $\vct \beta$ is not much larger than $1/\alpha$. While our new approach is more computationally intensive in principle, we introduce computational tricks that accelerate the method substantially without sacrificing FDR control in theory. Our implementation of cKnockoff turns out to be quite efficient in our numerical experiments in the sense that the computation time is only a small multiple of that of knockoffs, and it can be further accelerated by parallelization.

\section*{Reproducibility}
Calibrated knockoffs are implemented in an \texttt{R} package available online at the Github repository \texttt{https://github.com/yixiangLuo/cknockoff}. And the \texttt{R} code to reproduce all simulations and figures in this paper can be found at \texttt{https://github.com/yixiangLuo/cknockoff\_expr}.

\section*{Acknowledgments}

William Fithian and Yixiang Luo were partially supported by National Science Foundation grant DMS-1916220, and a Hellman Fellowship from UC Berkeley.


\bibliographystyle{imsart-nameyear} 
\bibliography{biblio.bib}       

\begin{thebibliography}{22}

\bibitem[\protect\citeauthoryear{Barber and Cand{\`e}s}{2015}]{barber15}
\begin{barticle}[author]
\bauthor{\bsnm{Barber},~\bfnm{Rina~Foygel}\binits{R.~F.}} \AND
  \bauthor{\bsnm{Cand{\`e}s},~\bfnm{Emmanuel~J}\binits{E.~J.}}
(\byear{2015}).
\btitle{Controlling the false discovery rate via knockoffs}.
\bjournal{The Annals of Statistics}
\bvolume{43}
\bpages{2055--2085}.
\end{barticle}
\endbibitem

\bibitem[\protect\citeauthoryear{Benjamini and
  Hochberg}{1995}]{benjamini1995controlling}
\begin{barticle}[author]
\bauthor{\bsnm{Benjamini},~\bfnm{Yoav}\binits{Y.}} \AND
  \bauthor{\bsnm{Hochberg},~\bfnm{Yosef}\binits{Y.}}
(\byear{1995}).
\btitle{Controlling the false discovery rate: a practical and powerful approach
  to multiple testing}.
\bjournal{Journal of the Royal statistical society: series B (Methodological)}
\bvolume{57}
\bpages{289--300}.
\end{barticle}
\endbibitem

\bibitem[\protect\citeauthoryear{Cand\`{e}s et~al.}{2018}]{candes2018panning}
\begin{barticle}[author]
\bauthor{\bsnm{Cand\`{e}s},~\bfnm{Emmanuel}\binits{E.}},
  \bauthor{\bsnm{Fan},~\bfnm{Yingying}\binits{Y.}},
  \bauthor{\bsnm{Janson},~\bfnm{Lucas}\binits{L.}} \AND
  \bauthor{\bsnm{Lv},~\bfnm{Jinchi}\binits{J.}}
(\byear{2018}).
\btitle{Panning for gold: ‘{model-X}’ knockoffs for high dimensional
  controlled variable selection}.
\bjournal{Journal of the Royal Statistical Society: Series B (Statistical
  Methodology)}
\bvolume{80}
\bpages{551--577}.
\end{barticle}
\endbibitem

\bibitem[\protect\citeauthoryear{Dunnett}{1955}]{dunnett1955multiple}
\begin{barticle}[author]
\bauthor{\bsnm{Dunnett},~\bfnm{Charles~W}\binits{C.~W.}}
(\byear{1955}).
\btitle{A multiple comparison procedure for comparing several treatments with a
  control}.
\bjournal{Journal of the American Statistical Association}
\bvolume{50}
\bpages{1096--1121}.
\end{barticle}
\endbibitem

\bibitem[\protect\citeauthoryear{Emery and Keich}{2019}]{emery2019controlling}
\begin{barticle}[author]
\bauthor{\bsnm{Emery},~\bfnm{Kristen}\binits{K.}} \AND
  \bauthor{\bsnm{Keich},~\bfnm{Uri}\binits{U.}}
(\byear{2019}).
\btitle{Controlling the FDR in variable selection via multiple knockoffs}.
\bjournal{arXiv preprint arXiv:1911.09442}.
\end{barticle}
\endbibitem

\bibitem[\protect\citeauthoryear{Fithian and
  Lei}{2022}]{fithian2022conditional}
\begin{barticle}[author]
\bauthor{\bsnm{Fithian},~\bfnm{William}\binits{W.}} \AND
  \bauthor{\bsnm{Lei},~\bfnm{Lihua}\binits{L.}}
(\byear{2022}).
\btitle{Conditional calibration for false discovery rate control under
  dependence}.
\bjournal{The Annals of Statistics}
\bvolume{50}
\bpages{3091--3118}.
\end{barticle}
\endbibitem

\bibitem[\protect\citeauthoryear{Gimenez and Zou}{2019}]{gimenez2019improving}
\begin{binproceedings}[author]
\bauthor{\bsnm{Gimenez},~\bfnm{Jaime~Roquero}\binits{J.~R.}} \AND
  \bauthor{\bsnm{Zou},~\bfnm{James}\binits{J.}}
(\byear{2019}).
\btitle{Improving the stability of the knockoff procedure: Multiple
  simultaneous knockoffs and entropy maximization}.
In \bbooktitle{The 22nd International Conference on Artificial Intelligence and
  Statistics}
\bpages{2184--2192}.
\bpublisher{PMLR}.
\end{binproceedings}
\endbibitem

\bibitem[\protect\citeauthoryear{Javanmard and
  Montanari}{2014}]{javanmard2014confidence}
\begin{barticle}[author]
\bauthor{\bsnm{Javanmard},~\bfnm{Adel}\binits{A.}} \AND
  \bauthor{\bsnm{Montanari},~\bfnm{Andrea}\binits{A.}}
(\byear{2014}).
\btitle{Confidence intervals and hypothesis testing for high-dimensional
  regression}.
\bjournal{The Journal of Machine Learning Research}
\bvolume{15}
\bpages{2869--2909}.
\end{barticle}
\endbibitem

\bibitem[\protect\citeauthoryear{Katsevich, Sabatti and
  Bogomolov}{2021}]{katsevich2021filtering}
\begin{barticle}[author]
\bauthor{\bsnm{Katsevich},~\bfnm{Eugene}\binits{E.}},
  \bauthor{\bsnm{Sabatti},~\bfnm{Chiara}\binits{C.}} \AND
  \bauthor{\bsnm{Bogomolov},~\bfnm{Marina}\binits{M.}}
(\byear{2021}).
\btitle{Filtering the rejection set while preserving false discovery rate
  control}.
\bjournal{Journal of the American Statistical Association}
\bpages{1--12}.
\end{barticle}
\endbibitem

\bibitem[\protect\citeauthoryear{Li and Fithian}{2021}]{li2021whiteout}
\begin{barticle}[author]
\bauthor{\bsnm{Li},~\bfnm{Xiao}\binits{X.}} \AND
  \bauthor{\bsnm{Fithian},~\bfnm{William}\binits{W.}}
(\byear{2021}).
\btitle{Whiteout: when do fixed-X knockoffs fail?}
\bjournal{arXiv preprint arXiv:2107.06388}.
\end{barticle}
\endbibitem

\bibitem[\protect\citeauthoryear{Nguyen et~al.}{2020}]{nguyen2020aggregation}
\begin{binproceedings}[author]
\bauthor{\bsnm{Nguyen},~\bfnm{Tuan-Binh}\binits{T.-B.}},
  \bauthor{\bsnm{Chevalier},~\bfnm{J{\'e}r{\^o}me-Alexis}\binits{J.-A.}},
  \bauthor{\bsnm{Thirion},~\bfnm{Bertrand}\binits{B.}} \AND
  \bauthor{\bsnm{Arlot},~\bfnm{Sylvain}\binits{S.}}
(\byear{2020}).
\btitle{Aggregation of multiple knockoffs}.
In \bbooktitle{International Conference on Machine Learning}
\bpages{7283--7293}.
\bpublisher{PMLR}.
\end{binproceedings}
\endbibitem

\bibitem[\protect\citeauthoryear{Ning and Liu}{2017}]{ning2017general}
\begin{barticle}[author]
\bauthor{\bsnm{Ning},~\bfnm{Yang}\binits{Y.}} \AND
  \bauthor{\bsnm{Liu},~\bfnm{Han}\binits{H.}}
(\byear{2017}).
\btitle{A general theory of hypothesis tests and confidence regions for sparse
  high dimensional models}.
\end{barticle}
\endbibitem

\bibitem[\protect\citeauthoryear{Rhee et~al.}{2006}]{rhee2006genotypic}
\begin{barticle}[author]
\bauthor{\bsnm{Rhee},~\bfnm{Soo-Yon}\binits{S.-Y.}},
  \bauthor{\bsnm{Taylor},~\bfnm{Jonathan}\binits{J.}},
  \bauthor{\bsnm{Wadhera},~\bfnm{Gauhar}\binits{G.}},
  \bauthor{\bsnm{Ben-Hur},~\bfnm{Asa}\binits{A.}},
  \bauthor{\bsnm{Brutlag},~\bfnm{Douglas~L}\binits{D.~L.}} \AND
  \bauthor{\bsnm{Shafer},~\bfnm{Robert~W}\binits{R.~W.}}
(\byear{2006}).
\btitle{Genotypic predictors of human immunodeficiency virus type 1 drug
  resistance}.
\bjournal{Proceedings of the National Academy of Sciences}
\bvolume{103}
\bpages{17355--17360}.
\end{barticle}
\endbibitem

\bibitem[\protect\citeauthoryear{Sarkar and Tang}{2021}]{sarkar2021adjusting}
\begin{barticle}[author]
\bauthor{\bsnm{Sarkar},~\bfnm{Sanat~K}\binits{S.~K.}} \AND
  \bauthor{\bsnm{Tang},~\bfnm{Cheng~Yong}\binits{C.~Y.}}
(\byear{2021}).
\btitle{Adjusting the Benjamini-Hochberg method for controlling the false
  discovery rate in knockoff assisted variable selection}.
\bjournal{arXiv preprint arXiv:2102.09080}.
\end{barticle}
\endbibitem

\bibitem[\protect\citeauthoryear{Shah and B{\"u}hlmann}{2023}]{shah2023double}
\begin{barticle}[author]
\bauthor{\bsnm{Shah},~\bfnm{Rajen~D}\binits{R.~D.}} \AND
  \bauthor{\bsnm{B{\"u}hlmann},~\bfnm{Peter}\binits{P.}}
(\byear{2023}).
\btitle{Double-estimation-friendly inference for high-dimensional misspecified
  models}.
\bjournal{Statistical Science}
\bvolume{38}
\bpages{68--91}.
\end{barticle}
\endbibitem

\bibitem[\protect\citeauthoryear{Shah and Peters}{2020}]{shah2020hardness}
\begin{barticle}[author]
\bauthor{\bsnm{Shah},~\bfnm{Rajen~D}\binits{R.~D.}} \AND
  \bauthor{\bsnm{Peters},~\bfnm{Jonas}\binits{J.}}
(\byear{2020}).
\btitle{The hardness of conditional independence testing and the generalised
  covariance measure}.
\end{barticle}
\endbibitem

\bibitem[\protect\citeauthoryear{Spector and
  Janson}{2022}]{spector2022powerful}
\begin{barticle}[author]
\bauthor{\bsnm{Spector},~\bfnm{Asher}\binits{A.}} \AND
  \bauthor{\bsnm{Janson},~\bfnm{Lucas}\binits{L.}}
(\byear{2022}).
\btitle{Powerful knockoffs via minimizing reconstructability}.
\bjournal{The Annals of Statistics}
\bvolume{50}
\bpages{252--276}.
\end{barticle}
\endbibitem

\bibitem[\protect\citeauthoryear{Tibshirani}{1996}]{tibshirani1996regression}
\begin{barticle}[author]
\bauthor{\bsnm{Tibshirani},~\bfnm{Robert}\binits{R.}}
(\byear{1996}).
\btitle{Regression shrinkage and selection via the lasso}.
\bjournal{Journal of the Royal Statistical Society: Series B (Methodological)}
\bvolume{58}
\bpages{267--288}.
\end{barticle}
\endbibitem

\bibitem[\protect\citeauthoryear{Van~de Geer
  et~al.}{2014}]{van2014asymptotically}
\begin{barticle}[author]
\bauthor{\bparticle{Van~de} \bsnm{Geer},~\bfnm{Sara}\binits{S.}},
  \bauthor{\bsnm{B{\"u}hlmann},~\bfnm{Peter}\binits{P.}},
  \bauthor{\bsnm{Ritov},~\bfnm{Ya’acov}\binits{Y.}} \AND
  \bauthor{\bsnm{Dezeure},~\bfnm{Ruben}\binits{R.}}
(\byear{2014}).
\btitle{On asymptotically optimal confidence regions and tests for
  high-dimensional models}.
\end{barticle}
\endbibitem

\bibitem[\protect\citeauthoryear{Waudby-Smith and
  Ramdas}{2020}]{waudby2020estimating}
\begin{barticle}[author]
\bauthor{\bsnm{Waudby-Smith},~\bfnm{Ian}\binits{I.}} \AND
  \bauthor{\bsnm{Ramdas},~\bfnm{Aaditya}\binits{A.}}
(\byear{2020}).
\btitle{Estimating means of bounded random variables by betting}.
\bjournal{arXiv preprint arXiv:2010.09686}.
\end{barticle}
\endbibitem

\bibitem[\protect\citeauthoryear{Weinstein et~al.}{2020}]{weinstein2020power}
\begin{barticle}[author]
\bauthor{\bsnm{Weinstein},~\bfnm{Asaf}\binits{A.}},
  \bauthor{\bsnm{Su},~\bfnm{Weijie~J}\binits{W.~J.}},
  \bauthor{\bsnm{Bogdan},~\bfnm{Ma{\l}gorzata}\binits{M.}},
  \bauthor{\bsnm{Barber},~\bfnm{Rina~F}\binits{R.~F.}} \AND
  \bauthor{\bsnm{Candes},~\bfnm{Emmanuel~J}\binits{E.~J.}}
(\byear{2020}).
\btitle{A power analysis for knockoffs with the lasso coefficient-difference
  statistic}.
\bjournal{arXiv preprint arXiv:2007.15346}
\bvolume{2}.
\end{barticle}
\endbibitem

\bibitem[\protect\citeauthoryear{Zhang and Zhang}{2014}]{zhang2014confidence}
\begin{barticle}[author]
\bauthor{\bsnm{Zhang},~\bfnm{Cun-Hui}\binits{C.-H.}} \AND
  \bauthor{\bsnm{Zhang},~\bfnm{Stephanie~S}\binits{S.~S.}}
(\byear{2014}).
\btitle{Confidence intervals for low dimensional parameters in high dimensional
  linear models}.
\bjournal{Journal of the Royal Statistical Society Series B: Statistical
  Methodology}
\bvolume{76}
\bpages{217--242}.
\end{barticle}
\endbibitem

\end{thebibliography}


\appendix

\section{Further details on knockoffs}\label{app:knockoff}

\subsection{Constructing knockoff variables and feature importance statistics}
\label{app:kn-variables}

This section briefly reviews the construction of knockoff variables and feature importance statistics.

In the fixed-X setting, with $n \geq 2m$ \citep{barber15}, the knockoff matrix $\tX = (\tX_1, \ldots, \tX_m) \in \RR^{n \times m}$ must satisfy:
\begin{equation}\label{eq:knockoffs-cond}
\tX^\tran \tX = \mat X^\tran \mat X, \quad \text{ and }\;\; \mat X^\tran \tX = \mat X^\tran \mat X - \mat D,
\end{equation}
for some diagonal matrix $\mat D$. We find $\mat D$ by optimizing a certain objective function under the constraint $\mat D \succeq 0$ and $\mat D \preceq 2 \mat X^\tran \mat X$ \citep{barber15, spector2022powerful}.
The model-$X$ knockoff variables $\tX=(\tX_1, \ldots, \tX_p)$ must be a family of random variables satisfying 
 \[ (\mat X, \tX)_{\operatorname{swap}(\cS)} \stackrel{d}{=}(\mat X, \tX) \]
for any subset $\cS \subseteq [m]$, where the vector $(\mat X, \tX)_{\text {swap }(\cS)}$ is obtained from $(\mat X, \tX)$ by swapping the entries $\vct X_j$ and $\tX_j$ for each $j \in \cS$ \citep{candes2018panning}.

The feature importance statistic $W_j$ must fulfill two conditions:
\begin{enumerate}
\item Sufficiency: In the fixed-X settings, $W_j$ must be a function of $\mat X_+^\tran \mat X_+$ and $\mat X_+^\tran \vct y$. In model-X, $W_j$ must be a function of $\mat X_+$ and $\vct y$.
\item Anti-symmetry: $W_j$ must have equal magnitude but opposite sign when $\vct X_j$ and $\tX_j$ are swapped:
\[ 
W_j([\mat X, \mat \tX]_{\operatorname{swap}(\cS)}, y) =
\begin{cases}
     W_j([\mat X, \mat \tX], \vct y), & j \notin \cS \\ -W_j([\mat X, \mat \tX], \vct y), & j \in \cS
\end{cases}.
\]
\end{enumerate}

\subsection{Deferred proofs}
\label{app:kn-proofs}

Here, for the sake of completeness, we prove several results that appeared first in \citet{barber15} and other works.

\begin{thm} \label{thm:signW}
    Suppose the feature statistics $\vct W = (W_1, W_2, \ldots, W_m)$ satisfy the sufficiency condition and the anti-symmetry condition. Then
    \[ \sgn(W_j) \,\mid\, |W_j|,\, \vct W_{-j} \;\stackrel{H_j}{\sim}\; \textnormal{Unif}\{-1,+1\} \]
    for any $j \in \cH_0$.
\end{thm}

\begin{proof}
It suffices to show for a given arbitrary $j \in [m]$,
\[ \vct W \eqd (W_1, \ldots, W_{j-1}, -W_j, W_{j+1}, \ldots, W_m)\triangleq \vct W_{j}^{-}. \]
The sufficiency condition allows us to write
\[ \vct W = g(\mat X_+^\tran \mat X_+, \mat X_+^\tran \vct y) \]
for some function $g$, where $\vct X_+$ is the augmented matrix defined in Section \ref{subsec:knockoff}. The anti-symmetry condition implies
\[ \vct W^- = g((\mat X^{\text{swap}}_+)^\tran \mat X^{\text{swap}}_+, (\mat X^{\text{swap}}_+)^\tran \vct y), \]
where $\mat X_+^{\text{swap}}$ is obtained by swapping $\vct X_j$ and $\tX_j$ in $\vct X_+$, i.e.
\[ 
(\mat X^{\text{swap}}_+)_j = (\mat X_+)_{j+m}, \quad
(\mat X^{\text{swap}}_+)_{j+m} = (\mat X_+)_{j}, \quad
(\mat X^{\text{swap}}_+)_{i} = (\mat X_+)_{i} \; \text{ for } i \neq j,\, j+m.
\]
$\vct X_j^\tran \vct X_i = \tX_j^\tran \vct X_i$ for all $i \neq j$ and $\beta_j = 0$,
\[
(\mat X^{\text{swap}}_+)^\tran \mat X^{\text{swap}}_+ = \mat X_+^\tran \mat X_+, \quad
(\mat X^{\text{swap}}_+)^\tran \mat X \beta = \mat X_+^\tran \mat X \vct \beta
\]
Furthermore,
\[
(\mat X^{\text{swap}}_+)^\tran \vct y
\;\eqd\; \cN((\mat X^{\text{swap}}_+)^\tran \mat X \vct \beta, (\mat X^{\text{swap}}_+)^\tran \mat X^{\text{swap}}_+)
\;\eqd\; \cN(\mat X_+^\tran \mat X \vct \beta, \mat X_+^\tran \mat X_+)
\;\eqd\; \mat X_+^\tran \vct y.
\]
Therefore,
\[ \vct W_j^- = g((\mat X^{\text{swap}}_+)^\tran \mat X^{\text{swap}}_+, (\mat X^{\text{swap}}_+)^\tran \vct y) \;\eqd\; g(\mat X_+^\tran \mat X_+, \mat X_+^\tran \vct y) = \vct W. \]
\end{proof}

\begin{thm} \label{thm:supMtg}
    Suppose the feature statistics satisfy the sufficiency condition and the anti-symmetry condition. Then
    \[ M_t \;=\; \frac{|\cC(w_t) \cap \cH_0|}{1+|\cA(w_t) \cap \cH_0|} \]
    is a supermartingale with respect to the filtration 
    \[ \cF_t = \sigma\Big(\,|W|, \;\left(W_j: j \in \cH_0^\setcomp \text{ or } |W_j| < w_t \right), \; |\cC(w_t)|\,\Big). \]
    Moreover, we have $\EE [M_1] \leq 1$.
\end{thm}

\begin{proof}
Without loss of generality, assume $|W_1| < \cdots < |W_m|$. Then $w_t = W_t$ and $|W_j| < w_t$ is equivalent to $j < t$. Let $V^+_t \coloneqq |\cC(w_t) \cap \cH_0|$ and $V^-_t \coloneqq |\cA(w_t) \cap \cH_0|$ for short. Since the non-null feature statistics are known given $\cF_t$, it's easy to see $V^+_t$ and $V^-_t$ are measurable with respect to $\cF_t$. Hence $M_t \in \cF_t$.

To see $\cF_t$ is indeed a filtration, note $\set{W_j: j \in \cH_0^\setcomp \text{ or } |W_j| < w_t}$ is monotone in $t$. Given this set and $|\cC(w_t)|$, we are able to compute $|\cC(w_s)|$ for all $s \leq t$. Therefore, $\cF_t$ is also monotone in $t$, namely, $\cF_t \subseteq \cF_{t+1}$ for all $t$.

It remains to show $\EE[M_{t+1} \mid \cF_t] \leq M_t$. By construction, $M_{t+1} = M_t$ if $t \in \cH_0^\setcomp$. Otherwise
\[
M_{t+1}
= \frac{V^+_t - \one\set{W_t > 0}}{1+V^-_t - (1 - \one\set{W_t > 0})}
= \frac{V^+_t - \one\set{W_t > 0}}{(V^-_t + \one\set{W_t > 0}) \mam 1}.
\]
Theorem \ref{thm:signW} implies that
\[ \sgn(W_j) \,\mid\, |\vct W|,\, \left(W_i: i \in \cH_0^\setcomp \text{ or } i < t \right) \simiid \textnormal{Unif}\{-1,+1\}, \quad \text{for } j \in \cH_0 \text{ and } j \geq t.\]
Hence
\[ \PP(W_t > 0 \mid \cF_t) = \frac{V^+_t}{V^+_t + V^-_t}. \]
Therefore
\begin{equation*}
    \EE[M_{t+1} \mid \cF_t]
    \;=\; \frac{1}{V^{+}_t+V^{-}_t}\left[V^{+}_t \frac{V^{+}_t-1}{V^{-}_t+1}+V^{-}_t \frac{V^{+}_t}{V^{-}_t \vee 1}\right]
    \;=\; \begin{cases}
      \frac{V^{+}_t}{1+V^{-}_t} = M_t, & V^{-}_t>0 \\
      V^{+}_t-1 = M_t - 1, & V^{-}_t=0
    \end{cases}.
\end{equation*}
So $M_t$ is s supermartingale with respect to the filtration $\mathcal{F}_t$.

To show $\EE [M_1] \leq 1$, note that $V^+_1 \mid |\vct W| \sim \text{Binomial}(m_0, 1/2)$. We have
\begin{align*}
    \EE \br{M_1 \mid |\vct W|}
    &\;=\; \EE \br{\frac{V^+_1}{1 + m_0 - V^+_1} \mid |\vct W|}
    \;=\; \sum_{i=1}^{m_0} \PP(V^+_1 = i) \cdot \frac{i}{1 + m_0 - i} \\
    &\;=\; \sum_{i=1}^{m_0} \frac{1}{2^{m_0}} \frac{m_0!}{i!(m_0-i)!} \cdot \frac{i}{1 + m_0 - i}
    \;=\; \sum_{i=1}^{m_0} \frac{1}{2^{m_0}} \frac{m_0!}{(i-1)!(m_0-i+1)!} \\
    &\;=\; \sum_{i=1}^{m_0} \PP(V^+_1 = i-1)
    \leq 1.
\end{align*}
Then $\EE [M_1] = \EE \br{\EE \br{M_1 \mid |\vct W|}} \leq 1$.
\end{proof}

\subsection{LCD with tiebreaker}
\label{app:LCD-T}

We define the {\em LCD with tiebreaker} (LCD-T) statistics as
\begin{equation}\label{eq:LCD-TB}
W_j^{\text{LCD-T}} \,=\, 
\begin{cases}
W_j^{\text{LCD}} + 2 \lambda \,\text{sgn}(W_j^{\text{LCD}}) & \text{if } W_j^{\text{LCD}} \neq 0\\[5pt]
|\vct X_j^\tran \vct r^{\lambda}| - |\tX_j^\tran \vct r^{\lambda}| & \text{otherwise.}
\end{cases}
\end{equation}
Applying the Karush–Kuhn–Tucker (KKT) condition, it is easy to verify that $|W_j^{\text{LCD-T}}| > 2 \lambda$ if and only if $W_j^{\text{LCD}} \neq 0$.

If $n \geq 2m + 1$, then $W$ can also take as input the unbiased variance estimator $\tilde\sigma^2 = \|\vct r^0\|_2^2/(n-2m)$, which is independent of $\mat X_+^\tran \vct y$ \citep{li2021whiteout}. This can help us to select $\lambda$ in \eqref{eq:LCD}; we find that $\lambda = 2\tilde\sigma$ is a practical choice \footnote{
Here $\lambda$ is chosen not for variable selection consistency but for feature importance score. Roughly speaking, setting $\lambda = 2\tilde\sigma$ blocks most (95\%) null variables for orthogonal designs. In our numerical examples, we found that knockoffs is not very sensitive to $\lambda$ when the tiebreaker is employed.
}
\footnote{
\citet{weinstein2020power} studied the effect of $\lambda$ on the power of model-X knockoffs with LCD statistics and propose a cross-validation-based method to find the optimal $\lambda$. However, the method cannot be applied to the fixed-X knockoffs because cross-validation violates the sufficiency principle. There could exist an equally efficient analog in the fixed-X setting. We leave it for future research.
},
where the explanatory variables are standardized to have a unit norm.

\subsection{A faster version of the LSM statistic}
\label{app:fast-LSM}

The LSM statistic defined in \eqref{eq:LSM} is computationally burdensome because it requires calculating the entire lasso path. Even if the great majority of variables are null, they will eventually enter and leave the model in a chaotic process once $\lambda$ becomes small enough. If we stop too early, most variables will never enter, so their feature statistics will be zero and there will be no chance to discover them, but if we stop too late, we will consume most of our computational resources fitting null variables at the end of the path. In practice, the path is also calculated only for a fine grid of $\lambda$ values, which has the added undesirable effect of introducing artificial ties between variables. 

This section introduces a more computationally efficient alternative that uses a coarser grid of $\lambda$ values and also stops the path early, but breaks ties by using variables' correlations using the residuals at each stage. Assume we calculate the lasso fit $\hat{\vct \beta}^{\lambda(\ell)}$ for $\ell=0,\ldots,L$ on a coarse grid defined by:
\[
\max_{1\leq j\leq 2m} \lambda_j^* = \lambda_{\text{max}} = \lambda(0) > \lambda(1) > \cdots > \lambda(L) = \lambda_{\text{min}} = 2\tilde\sigma \mim \frac{\lambda_{\text{max}}}{2}.
\]
We stop the path at $2\tilde\sigma$ because we find it tends to set most null variables' coefficients to zero, we take $\lambda(0),\ldots,\lambda(L)$ to be a decreasing geometric sequence:
\[
\lambda(\ell) = \lambda(0) \cdot \zeta^\ell, \quad \text{ for } \zeta = \left(\frac{\lambda(L)}{\lambda(0)}\right)^{1/L}.
\]
Then for variable $j = 1,\ldots,2m$, define its (discrete) time of entry as
\[
\ell_j^* = \min\left\{\ell \in \{1,\ldots,L\}:\; \hat\beta_j^{\lambda(\ell)} \neq 0\right\},
\]
with $\ell_j^* = L+1$ if the set is empty. To break ties between these discrete values, we use each variable's correlation with the lasso residuals at $\lambda_j^+ = \lambda(\ell_j^*-1)$, the last fit in the discrete path before variable $j$ enters:
\[
\rho_j = \left|(\mat X_+)_j^\tran \vct r^{\lambda_j^+}\right| \leq \lambda_j^+, \quad \text{ where } \vct r^\lambda = \vct y - \mat X_+\hat{\vct \beta}^\lambda.
\]

$\rho_j$ can be considered an estimate of $\lambda_j^*$, since we have
\[
\lambda_j^* = \left|(\mat X_+)_j^\tran r^{\lambda_j^*}\right| \approx \left|(\mat X_+)_j^\tran r^{\lambda_j^+}\right|.
\]
To combine these, we can use any transform that orders variables first by $\lambda_j^+$ and then by $\rho_j$. We use a transform that also aids the accuracy of the local linear regression approximation of $W_j(\vct z)$ on $\vct X_j^\tran \vct z$. Let $\iota$ denote a small positive quantity and $\lambda(L+1) = 0$, and define
\[
\hat\lambda_j \;=\; \max\{\rho_j - \iota,\; \lambda(\ell_j^*)\} + \iota \rho_j / \lambda_j^+ \;\approx\; \max\{\rho_j,\, \lambda(\ell_j^*)\}. 
\]
Finally, we define the {\em coarse LSM} (C-LSM) feature statistics by substituting $\hat\lambda_j$ for $\lambda_j^*$ in \eqref{eq:LSM}:
\[
W_j^{\text{C-LSM}} = (\hat\lambda_j \mam \hat\lambda_{j+m}) \cdot \sgn(\hat\lambda_j - \hat\lambda_{j+m}).
\]


\subsection{numerical comparisons}

Here we compare the performance of the vanilla LSM feature statistics with our LCD-T and C-LSM. Figure \ref{fig:simu-knockoff_stats} has settings the same as in Section \ref{sec:simu_main} and Figure \ref{fig:simu-knockoff_stats_lambda1_50m1} is less sparse with $m_1 = 50$ while using a different $\lambda = \tilde \sigma$. We see the performance of C-LSM is roughly the same as the original LSM in all cases we consider and LCD-T is sometimes even better, as suggested in \citet{weinstein2020power}.

\begin{figure}[!tb]
    \centering
    \includegraphics[width = 0.95\linewidth]{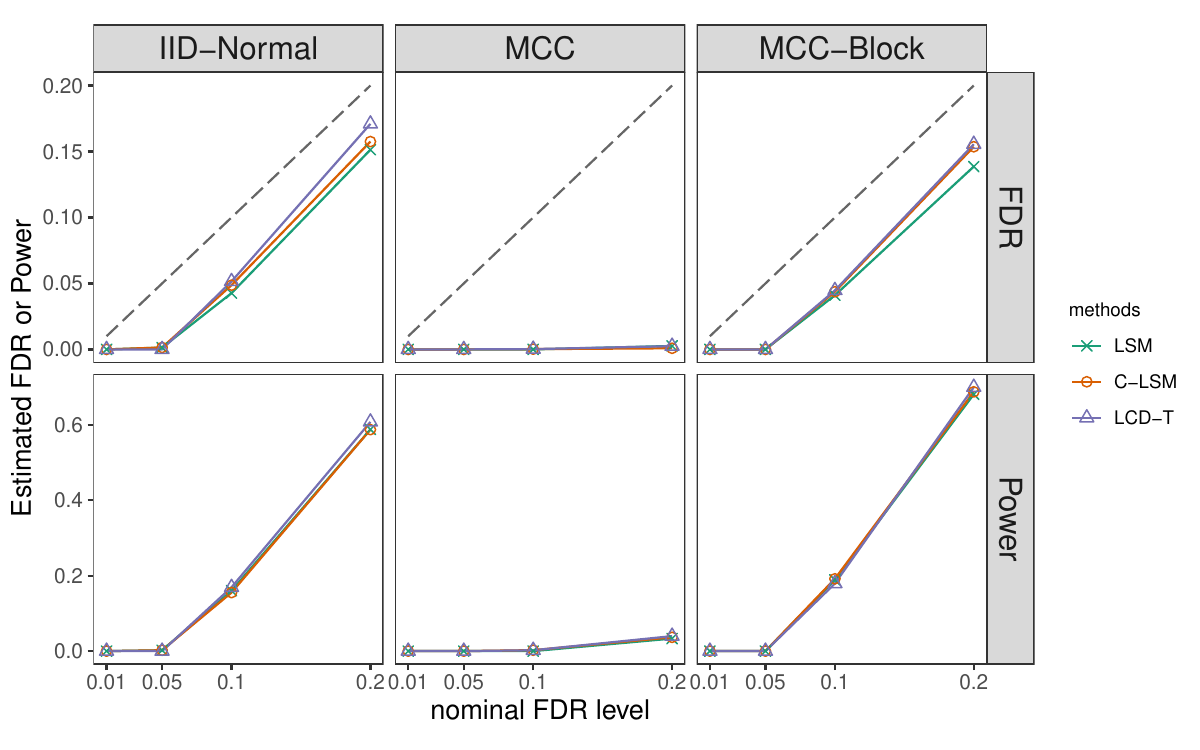}
    \caption{Estimated FDR and TPR of knockoffs employing LSM, LCD-T or C-LSM as feature statistics, where $m_1 = 10$ and $\lambda = 2 \tilde \sigma$ in LCD-T.}
    \label{fig:simu-knockoff_stats}
\end{figure}

\begin{figure}[!tb]
    \centering
    \includegraphics[width = 0.95\linewidth]{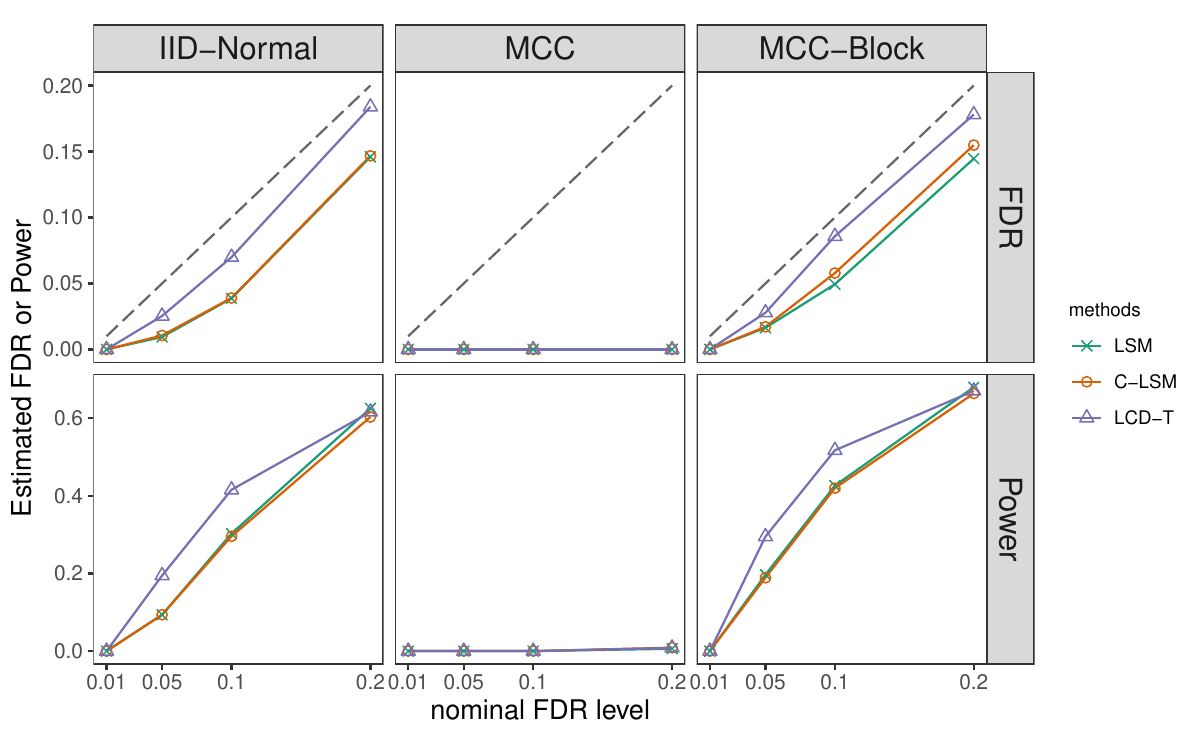}
    \caption{Estimated FDR and TPR of knockoffs employing LSM, LCD-T or C-LSM as feature statistics, where $m_1 = 50$ and $\lambda = \tilde \sigma$ in LCD-T.}
    \label{fig:simu-knockoff_stats_lambda1_50m1}
\end{figure}

\section{Implementation: conservative fallback test in fixed-X cKnockoff}
\label{app:implement}


\subsection{Construction of $\tE_j$}
\label{app:conservative-fallback}


Let $Q_j(\cdot \mid S_j)$ denote the conditional distribution of the response vector $\vct y$ under $H_j$ given $S_j(\vct y) = (\mat X_{-j}^{\tran} \vct y, \,\|\vct y\|_2^2)$. Throughout the section, we are only concerned with $E_j(c; S_j(\vct y))$ once the response $\vct y$ has already been observed. As such, we regard $S_j(\vct y)$, $c = T_j(\vct y)$, and $Q_j(\cdot \mid S_j(\vct y))$ as fixed inputs to the integral $E_j(c\,; S_j(\vct y))$, and suppress their dependence on $\vct y$. To avoid confusion, we use $\vct z$ to denote a generic response vector drawn from the conditional null distribution $Q_j$.

The support of $Q_j$ is the preimage of $S_j(\vct y)$, a sphere of dimension $n-m$ embedded in $\RR^n$, on which $\vct y$ is conditionally uniform under $H_j$; see Section~\ref{app:sampling}. We can write the conditional expectation as
\begin{equation}\label{eq:integral}
E_j(c\,; S_j) \;=\; \EE_{H_j}\left[ f_j(\vct y; c) \mid S_j\right] \;=\; \int f_j(\vct z; c)\,dQ_j(\vct z \mid S_j),
\end{equation}
where the integrand is given by
\begin{equation} \label{eq:integrand}
    f_j(\vct z; c) \;=\; \frac{\one \set{j \in \cR^\kn(\vct z)} \,\mam\, \one \set{T_j(\vct z) \geq c}}{\left| \cR^\kn(\vct z) \cup \{j\}\right|} \,-\, b_j(\vct z),
\end{equation}
with 
$ b_j(\vct z) \;=\; \alpha\, {\DP_j(\cC(w_{\tau_1}))}/{\hFDP(w_{\tau_1})}  $
as defined in \eqref{eq:bj_def}.

We express $f_j = f_j^{(T)} - f_j^{(b)}$, where
\begin{gather} \label{eq:int_decompose}
    f_j^{(T)} = \frac{\one \set{T_j \geq c}}{\left| \cR^\kn \cup \{j\}\right|}, \quad
    f_j^{(b)} = \DP_j(\cC(w_{\tau_1})) \cdot \frac{\alpha}{\hFDP(w_{\tau_1})} - \DP_j(\cR^\kn) \cdot \one \set{T_j < c}.
\end{gather}
$f_j^{(T)}$ is non-negative by definition. And we have
\begin{lemma}
    $f_j^{(b)}(\vct z) \geq 0$ for all $\vct z$.
\end{lemma}
\begin{proof}
    It suffices to show the statement when $\DP_j(\cR^\kn) \neq 0$. In this case, $j \in \cR^\kn$ and thus $\cR^{\kn} \neq \emptyset$. By definition \eqref{eq:bj_def}, $\tau_1 = \tau$ if $\cR^{\kn}$ is non-empty. Therefore, we have
    \[ 
    \cC(w_{\tau_1}) = \cC(w_{\tau}) = \cR^\kn, \quad
    \hFDP(w_{\tau_1}) = \hFDP(w_{\tau}) \leq \alpha.
    \]
    This implies
    $f_j^{(b)} \;\geq\; \DP_j(\cR^\kn) - \DP_j(\cR^\kn) \cdot \one \set{T_j < c} \;\geq\; 0$.
\end{proof}


Ideally, we would like to integrate $f_j^{(T)} - f_j^{(b)}$ on its support
\begin{equation}\label{eq:Omegaj}
\Omega_j(\vct y) \;=\; \{\vct z \in \text{Supp}(Q_j):\; f_j(\vct z; c) \neq 0\}.
\end{equation}
Again, we will suppress the dependence on $\vct y$ when no confusion can arise. We can write $\Omega_j \,= \,\Omega_j^{(T)} \,\cup\, \Omega_j^{(b)}$ with
\begin{equation}\label{eq:Omegaj-union}
\Omega_j^{(T)} \;=\; \{\vct z \in \Supp(Q_j):\; T_j(\vct z) \geq c\}, \quad \text{ and } \;\;
\Omega_j^{(b)} \;=\; \{\vct z \in \Supp(Q_j):\; j \in \cC(w_{\tau_1})\}
\end{equation}
are supports of $f_j^{(T)}$ and $f_j^{(b)}$ respectively.

For our fallback test statistic defined in \eqref{eq:Tj}, $\Omega_j^{(T)}$ amounts to a simple constraint on $\vct X_j^\tran \vct z$:
\begin{equation}\label{eq:Omega1_condition}
T_j(\vct z) \geq c \quad\Longleftrightarrow\quad \left|\vct X_j^\tran \left(\vct z - \hat{\vct y}^{(j)}(S_j)\right)\right| \geq c 
\quad\Longleftrightarrow\quad \vct X_j^\tran \vct z \;\notin\; \left(a_j^{(1)},\,a_j^{(2)}\right).
\end{equation}
where the bounds of the interval depend only on $\vct y$. Unfortunately, however, $\Omega_j^{(b)}$ admits no such simple description since it is defined implicitly in terms of the feature statistics.
Instead, we construct another set $\tOmega_j^{(b)}$ and define
\begin{equation}\label{eq:Etilde}
\tE_j(c\,; S_j) \;=\; \int_{\tOmega_j} f_j(\vct z; c) \,dQ_j(\vct z \mid S_j), \quad \text{ for } \tOmega_j \;=\; \Omega_j^{(T)} \cup \tOmega_j^{(b)}.
\end{equation}
Proposition~\ref{prop:approx} shows that if we implement the fallback test by checking if $\tE_j \leq 0$, the FDR control is intact.

\begin{prop}[Conservative fallback testing]\label{prop:approx}
Assume that, for each $j = 1,\ldots,m$, we calculate $\tE_j$ defined in \eqref{eq:Etilde} with any $\tOmega_j^{(b)}$. Define the conservative rejection set
\[
\tcR \;=\; \cR^{\kn} \,\cup\, \left\{j:\; \tE_j(T_j; S_j) \leq 0\right\}.
\]
Then $\FDR(\tcR) \leq \alpha$, and $\cR^{\kn} \subseteq \tcR \subseteq \cR^{\ckn}$.
\end{prop}

\begin{proof}
Fix some $j$ and let $c = T_j(\vct y)$. We have
\begin{align*}
    E_j(c) = \int_{\Omega_j} f_j^{(T)} - f_j^{(b)} \,dQ_j
    &= \int_{\Omega_j^{(T)}} f_j^{(T)} \,dQ_j - \int_{\Omega_j^{(b)}} f_j^{(b)} \,dQ_j \\
    &\leq \int_{\Omega_j^{(T)}} f_j^{(T)} \,dQ_j - \int_{\tOmega_j^{(b)}} f_j^{(b)} \,dQ_j
    = \tE_j,
\end{align*}
where the inequality is because that $f_j^{(b)} \geq 0$, hence
\[ \int_{\Omega_j^{(b)}} f_j^{(b)} \,dQ_j - \int_{\tOmega_j^{(b)}} f_j^{(b)} \,dQ_j = \int_{\Omega_j^{(b)} \setminus \tOmega_j^{(b)}} f_j^{(b)} \,dQ_j \geq 0. \]
Since $j$ is arbitrary and we reject $H_j$ when $E_j \leq 0$, this establishes that $\cR^{\kn} \subseteq \tcR \subseteq \cR^{\ckn}$, hence $\FDR(\tcR) \leq \alpha$ by Theorem~\ref{thm:sandwich}.
\end{proof}

Any choice of $\tOmega_j^{(b)}$ would not inflate the FDR while ideally, we like $\tOmega_j^{(b)}$ to approximate $\Omega_j^{(b)}$. In our implementation, we construct $\tOmega_j^{(b)}$ as another constraint on $\vct X_j^{\tran} \vct z$:
\[
\tOmega_j^{(b)} \;=\; \{\vct z\in \Supp(Q_j):\; \vct X_j^\tran \vct z \in A_j^{(b)}\}
\]
for some $A_j^{(b)}$ specified in Section~\ref{app:tOmega}. At a high level, our idea is to over estimate how likely $ j\in \cC_{\tau_1}$ based on the value of $\vct X_j^\tran \vct z$ and then set the region of $\vct X_j^\tran \vct z$ that implies high likelihood as $A_j^{(b)}$.

\subsection{Sampling $\vct z$ from $Q_j(\cdot \mid S_j)$} \label{app:sampling}

We first fix a basis to make our calculations convenient. Let $\mat V_{-j} \in \RR^{n \times (m-1)}$ denote an orthonormal basis for the column span of $\mat X_{-j}$, so that $\mat \Pi_{-j} = \mat V_{-j} \mat V_{-j}^\tran$. Next, for the projection of $\vct X_j$ orthogonal to the span of $\mat X_{-j}$, define the unit vector in that direction:
\[
\vct v_j = \frac{\mat \Pi_{-j}^\perp \vct X_j}{\|\mat \Pi_{-j}^\perp \vct X_j\|}, \quad \text{ where } \mat \Pi_{-j}^\perp = \mat I - \mat \Pi_{-j}
\]
Finally, let $\mat V_{\text{res}} \in \RR^{n \times (n-m)}$ denote an orthonormal basis for the subspace orthogonal to the span of $\mat X$. Then we can decompose $\vct z$ as
\[ \vct z = \mat V_{-j} \mat V_{-j}^\tran \vct z + \vct v_j \cdot \eta + \mat V_{\text{res}} \cdot \vct r, \]
where $\eta = \vct v_j^\tran \vct z \in \RR$ is the component of $\vct z$ in the direction $\vct v_j$, and $\vct r = \mat V_{\text{res}}^\tran \vct z \in \RR^{n-m}$ is the residual component.

Recall that $Q_j$ is uniform on its support $\{\vct z:\; S_j(\vct z) = S_j(\vct y)\}$ (see Section~\ref{app:cond_dist}). Fixing $S_j(\vct z) = S_j(\vct y)$ is equivalent to constraining
\[
\mat V_{-j}^\tran \vct z = \mat V_{-j}^\tran \vct y, \quad \text{ and } \; \eta^2 + \|\vct r\|^2 = \|\mat \Pi_{-j}^\perp \vct z\|^2 = \|\mat \Pi_{-j}^\perp \vct y\|^2.
\]
Let $\rho^2 = \|\mat \Pi_{-j}^\perp \vct y\|^2$, which depends only on $S_j(\vct y)$. We can sample $\vct z \sim Q_j$ by first sampling 
\[
(\eta, \vct r) \sim \text{Unif}\left(\rho \cdot \mathbb{S}^{n-m}\right),
\]
where $\mathbb{S}^{n-m} \subseteq \RR^{n-m+1}$ is the unit sphere of dimension $n-m$, and then reconstructing $\vct z$ using equation~\eqref{eq:z-decomp}.

To sample from $\tOmega_j$, we add a further constraint on $\eta$, that it lies in some union of intervals $U_j \subseteq [-\rho, \rho]$. See Section~\ref{app:tOmega} for details. In order to sample $\vct z$ satisfying this constraint efficiently, we first sample $\eta$ marginally obeying the constraint and then $\vct r$ conditional on $\eta$.
Standard calculations show the cumulative distribution function (CDF) of $\eta$ is
\begin{equation} \label{eq:cdf_eta}
    F_{\eta}(x; S_j) \coloneqq \PP(\eta \leq x) = F_{t_{n-m}} \pth{\frac{x \sqrt{n-m}}{\sqrt{(\rho^2 - x^2)}}} \;\; \text{for } |x|<\rho,
\end{equation}
where $F_{t_{n-m}}$ is the CDF of the t-distribution with degree of freedom $n-m$. Now write $\vct r = \|\vct r\| \cdot \vct u$ with $\vct u$ being the direction of $\vct r$. Given $\eta$, we have $\|\vct r\| = \sqrt{ \rho^2 - \eta^2}$ and $\vct u \sim \text{Unif}(\bS^{n-m-1})$ is independent of $\eta$.

To conclude, we sample $\eta \sim F_{\eta}$ in the desired union of intervals and $\vct u \sim \text{Unif}(\bS^{n-m-1})$ independently. The $\vct z$ is given by
\begin{equation}\label{eq:z-decomp}
\vct z(\eta, \vct u) = \mat V_{-j} \mat V_{-j}^\tran \vct y + \vct v_j \cdot \eta + \sqrt{ \rho^2 - \eta^2} \cdot \mat V_{\text{res}} \vct u.
\end{equation}

\subsection{The construction of $\tOmega_j$} \label{app:tOmega}

Section \ref{app:conservative-fallback} describes $\tOmega_j$ in terms of $\vct X_j^\tran \vct z$. To be consistent with the sampling scheme, we first rephrase $\tOmega_j$ in terms of $\eta$ as used in \eqref{eq:z-decomp}.

We decompose
\begin{equation} \label{eq:Xjz_decompose}
    \vct X_j^T \vct z = \vct X_j^T \mat \Pi_{-j} \vct z + (\vct v_j^\tran \vct X_j) \cdot \eta,
\end{equation}
which establishes a linear relationship between $\vct X_j^\tran \vct z$ and $\eta$ (recall $\mat \Pi_{-j} \vct z = \mat \Pi_{-j} \vct y$ is fixed). Thus we can write
\[
\tOmega_j \;=\; \{\vct z(\eta, \vct u) \in \Supp(Q_j):\; \eta \in A_j\}, \quad \text{ for } A_j = A_j^{(b)} \cup (a_j^{(1)}, a_j^{(2)})^{\setcomp},
\]
where $a_j^{(1)}$ and $a_j^{(2)}$ are solved from
\begin{equation} \label{eq:Omega1}
\{\vct z(\eta, \vct u) \in \Supp(Q_j):\; \eta \in (a_j^{(1)}, a_j^{(2)})^{\setcomp}\}
\;=\; \Omega_j^{(T)}(\vct y)
\;=\; \{\vct z \in \Supp(Q_j):\; T_j(\vct z) \geq c\}  
\end{equation}
and $A_j^{(b)}$ is a union of intervals in order to have
\[ \{\vct z(\eta, \vct u) \in \Supp(Q_j):\; \eta \in A_j^{(b)}\} \approx \Omega_j^{(b)}(\vct y). \]
Note we reuse the notation $A_j$, $A_j^{(b)}$, $a_j^{(1)}$, and $a_j^{(2)}$ as the constraint sets or boundaries for $\eta$ and they should not be confused with those in Section \ref{app:conservative-fallback}.

In the rest of this section, we give explicit ways to obtain $a_j^{(1)}$, $a_j^{(2)}$, and $A_j^{(b)}$.

For $a_j^{(1)}$ and $a_j^{(2)}$, note
\[
\Omega_j^{(T)}(\vct y) \;=\; \set{ \vct z \in \Supp(Q_j):\; \left|\vct X_j^\tran \left(\vct z - \hat{\vct y}^{(j)}(S_j)\right)\right| \geq c }.
\]
Using \eqref{eq:Xjz_decompose}, direct calculation shows
\[
a_j^{(1)} = \frac{\hat{\vct y}^{(j)}(S_j) - \vct X_j^\tran \mat \Pi_{-j} \vct y - c}{\vct v_j^\tran \vct X_j}, \quad
a_j^{(2)} = \frac{\hat{\vct y}^{(j)}(S_j) - \vct X_j^\tran \mat \Pi_{-j} \vct y + c}{\vct v_j^\tran \vct X_j}.
\]

For $A_j^{(b)}$, note
\[
\Omega_j^{(b)} = \set{\vct z \in \Supp(Q_j): \; j \in \cC(w_{\tau_1})} \;\subseteq\; \set{\vct z \in \Supp(Q_j): \; |W_j| \geq w_{\tau_1}}.
\]
Hence we want to set $\tOmega_j$ as $\set{\vct z \in \Supp(Q_j): \; |W_j| \geq w_{\tau_1}}$, which is approximately identified by local linear regression.

Specifically, we treat $\vct z = \vct z(\eta, \vct u)$ as a one-dimensional random function of $\eta$ with $\vct u$ being an independent random noise. Our local linear regression scheme then regresses $|W_j|$ and $w_{\tau_1}$ on $\eta$, and solve for the region where $|W_j| \geq w_{\tau_1}$ numerically.

Our method is motivated by the following observations in simulation studies.
\begin{enumerate}
    \item Given a fixed $\vct u$, for the lasso-based \textit{LSM} (\textit{C-LSM}) and \textit{LCD} (\textit{LCD-T}) feature statistics we consider, $|W_j(\vct z)|$ is roughly a piecewise linear function of $\eta$;
    \item $|W_j| \approx \EE \br{|W_j| \mid \eta}$, or say, the noise caused by $\vct u$ is small, if $\EE \br{|W_j| \mid \eta}$ is large. Specifically, we observe that the standard deviation of $|W_j|$ conditional on $\eta$ is typically a small fraction of the conditional mean when $\EE \br{|W_j| \mid \eta} \geq \EE \br{w_{\tau_1} \mid \eta}$.
\end{enumerate}
Heuristically, the first observation is because that the KKT conditions of Lasso is a piecewise linear system. And the second is that the randomness in $\vct u$ contributes to the correlation between $\vct y$ and the knockoff variables, which is considered noise. When $|W_j|$ is large, such a noise is expected to be dominated by the signal from the original variable. $w_{\tau_1}$ shares a similar behavior, though its value is determined by a more complicated mechanism. 

We then approximate
\[ 
\set{\vct z \in \Supp(Q_j): \; |W_j| \geq w_{\tau_1}} \approx
\set{\vct z \in \Supp(Q_j): \, \EE \br{|W_j| \mid \eta}\geq \EE \br{w_{\tau_1} \mid \eta}}, \]
with the conditional expectation estimated by the local linear regression.

To be specific, the construction of $\tOmega_j^{(b)}$ is done as follows.
\begin{enumerate}
    \item sample $\vct z_1(\eta_1, u_1), \ldots, \vct z_k(\eta_k, u_k)$ such that $\eta_1, \ldots, \eta_k$ are equi-spaced nodes in $[-\rho, \rho]$ and $u_1, \ldots, u_k$ are independent drawn from $\text{Unif}(\bS^{n-m-1})$.
    \item compute $|W_j(\vct z_i)|$ and $w_{\tau_1}(\vct z_i)$ for each $i$.
    \item estimate $\EE[|W_j(\vct z)| \mid \eta]$, denoted as $\widehat W_j(\eta)$, by a local linear regression on $|W_j(\vct z_i)|$ for $i = 1, \ldots, k$. We use the Gaussian kernel and set the bandwidth as the distance between consecutive nodes $\eta_i - \eta_{i-1}$. Similarly, estimate $\EE[w_{\tau_1}(\vct z) \mid \eta]$ as $\widehat w(\eta)$.
    \item $\tOmega_j^{(b)}$ is then
    \[ \tOmega_j^{(b)} \;=\; \{\vct z\in \Supp(Q_j):\; \eta \in A_j^{(b)}\}, \quad
    A_j^{(b)} \coloneqq \set{\eta: \widehat W_j(\eta) \geq \widehat w(\eta)}, \]
    where the inequality $\widehat W_j(\eta) \geq \widehat w(\eta)$ is solved numerically. 
\end{enumerate}

\subsection{Numerical error control} \label{app:error_control}

With enough computational budget, we can compute $\tE_j$ at arbitrary precision. While in practice, we should tolerate some level of numerical error introduced by Monte-Carlo. The key idea for such error control is to upper bound the probability of mistakenly claiming the sign of $\tE_j$ at some $\alpha_c$.

Our process of deciding $\sgn(\tE_j)$ can be formalized as constructing a confidence interval for $\tE_j$ from a sequence of i.i.d. samples $f_j(\vct z_1),\, f_j(\vct z_2),\, \ldots, f_j(\vct z_k)$, with a common mean $\tE_j / Q_j(\tOmega_j)$. For each $k = 1,\, 2,\, \ldots$, such a confidence interval at level $\alpha_c$, $C_k(f_j(\vct z_1),\, \ldots, f_j(\vct z_k); \, \alpha_c)$, is computed and once we see it excludes $0$, we stop the calculation and decide if $\tE_j \leq 0$ accordingly.

To control the sign error for  $\tE_j$, it suffices to have the sequence of confidence intervals $C_k$ hold for all (infinite many) $k$ simultaneously,
\[ \PP(\exists k: \, \tE_j \not\in C_k(f_j(\vct z_1),\, \ldots, f_j(\vct z_k)) \leq \alpha_c. \]




We can further control the inflation of the FDR due to the Monte-Carlo error, as shown next. 
\begin{prop} \label{prop:mc_conf}
    Let $\alpha_c(\vct y) = \left| \cR^\kn(\vct y) \cup \{j\}\right| \cdot \alpha_0 / m$ for some constant $\alpha_0$ and reject
    \[ \cR = \cR^{\kn(\alpha)} \cup \set{j: \;  \exists k, \; \text{s.t. } x \leq 0 \text{ for all } x \in C_k(f_j(\vct z_1),\, \ldots, f_j(\vct z_k); \, \alpha_c)}. \]
    Note $\cR$ is the cKnockoff rejection set if we compute $E_j$ and claim its sign with confidence as described above.
    Then
    \[ \FDR(\cR) \leq \alpha + \alpha_0. \]
\end{prop}
\begin{proof}
Denote event $D_j \coloneqq \set{\exists k, \; \text{s.t. } x \leq 0 \text{ for all } x \in C_k(f_j(\vct z_1),\, \ldots, f_j(\vct z_k); \, \alpha_c)}$ and
\[ D_j^\text{in} \coloneqq D \cap \set{\forall k,\; \tE_j \in C_k}, \quad
D_j^\text{out} \coloneqq D \cap \set{\exists k,\; \tE_j \not\in C_k}. \]
Recall $D_j^\text{in}$ implies $T_j \geq \hc_j$. And we have $\PP(D_j^\text{out} \mid \vct y) \leq \alpha_c$ by the construction of $C_k$. Therefore, we have
\begin{align*}
  \EE_{H_j}[\DP_j(\cR) \mid S_j]
  &\;=\; \EE_{H_j} \br{ \frac{\one \set{j \in \cR^\kn} \mam \one \set{D_j}}{\left| \cR \cup \{j\}\right|} \;\Big|\; S_j } \\[7pt]
  &\;\leq\; \EE_{H_j} \br{ \frac{\one \set{j \in \cR^\kn} \mam \one \set{T_j \geq \hc_j} \mam \one \set{D_j^\text{out}}}{\left| \cR^\kn \cup \{j\}\right|} \;\Big|\; S_j } \\[7pt]
  &\;\leq\; \EE_{H_j}[b_j \mid S_j] + \EE_{H_j} \br{ \frac{\one \set{D_j^\text{out}}}{\left| \cR^\kn \cup \{j\}\right|} \;\Big|\; S_j }.
\end{align*}
Marginalizing over $S_j$, we have
\[
  \FDR(\cR) = \sum_{j\in\cH_0} \EE [\,\DP_j(\cR)\,] \leq \sum_{j\in\cH_0} \EE[b_j] + \sum_{j\in\cH_0} \EE\left[ \EE\left[\frac{\one \set{D_j^\text{out}}}{\left| \cR^\kn(\vct y) \cup \{j\}\right|} \;\Big|\; \vct y \right] \right] \leq \alpha + \alpha_0.
\]
\end{proof}

\begin{remark}
The denominator $m$ in our choice of $\alpha_c$ can be replaced by the number of hypotheses that survived after filtering (see \App~\ref{app:filtering}), which would increase $\alpha_c$ and save some computation time while keeping the same FDR control.
\end{remark}
 
\begin{remark}
In practice, the above bound is rather loose and we found the FDR inflation is ignorable in most cases. 
In our implementation, we set $\alpha_0 = 0.1 \alpha$. To avoid having a infinite sequence in the sequential testing, we truncate the Monte-Carlo sample sequence after 500 samples if $0 \in C_k$ for all $k = 1, \ldots, 500$. We then use the sample mean to determine if $\tE_j \leq 0$.
\end{remark}



In particular, since the values of $f_j(z_k)$ are bounded, we apply Theorem 3 in \cite{waudby2020estimating} to construct such a confidence sequence $C_k$, which is adaptive to the sample variance and achieves state-of-the-art performance for bounded random variables.



\section{Implementation: filtering the rejection set in fixed-X cKnockoffs}
\label{app:filtering}

As suggested by Theorem \ref{thm:sandwich}, we reject the {\em filtered cKnockoff rejection set}
\[ \cR = \cR^{\kn} \cup \{j \in \cS:\; E_j(T_j; S_j) \leq 0\} \;\subseteq\; \cR^\ckn \]
for some set $\cS$ that only contains the hypotheses likely to be rejected by cKnockoff.
In practice, we find it is a good choice to set
\begin{equation} \label{eq:filter}
    \cS(s; \vct y) = (\cS^\BH(s; \vct y) \cap \cS^{\text{p}}(s; \vct y)) \cup \cS^{\kn}(s; \vct y) \setminus \cR^\kn(\vct y),
\end{equation}
where
\[
\cS^\BH = \cR^{\BH(s \cdot 4 \alpha)}, \quad
\cS^{\text{p}} = \set{j: \; p_j \leq s \cdot \alpha / 2}, \quad
\cS^{\kn} = \set{j: \; |W_j| \geq w_{m - |\cS^\BH \cap \cS^{\text{p}}|} \mim w_{\tau}}
\]
with $s = 1$. In words, $\cS^\BH \cap \cS^{\text{p}}$ requires that any $j \in \cS$ to have a relatively small $p$-value and $\cS^{\kn}$ requires $|W_j|$ to be relatively large. By the construction of the fallback test, $H_j$ is very unlikely to be rejected by cKnockoff if it has a large $p$-value and a small $|W_j|$. Among our simulations, this choice almost does not exclude any rejections in the vanilla cKnockoff and achieves $|S| \ll m$. That is to say, it preserves the power of cKnockoff when accelerating it a lot by filtering out many non-promising hypotheses.

Moreover, such filtering can be done in an online manner. Let
\begin{equation} \label{eq:score}
    s^+_j(\vct y) = \min \set{s:\; j \in \cS(s; \vct y)}.
\end{equation}
Then $s^+_j$, which we call {\em promising score}, roughly measures how likely $H_j$ will be rejected by cKnockoff. Then we can check $E_j(T_j; S_j) \leq 0$ sequentially on the hypotheses ranked by their promising scores. Theorem \ref{thm:sandwich} allows us to stop at any time we like and report the rejection set with a valid FDR control. This feature is available in our \texttt{R} package but not employed in the numerical experiments shown in this paper.

To apply filtering to cKnockoff$^*$, we reject the {\em filtered cKnockoff$^*$ rejection set}
\[ \cR = \cR^{\kn} \cup \{j \in \cS:\; E^*_j(T_j; S_j) \leq 0 \} \;\subseteq\; \cR^{\ckn^*} \]
where $\cS$ is further required to satisfy $\cR^\kn \cup \cS \supseteq \cR^*$. Like in cKnockoff, Theorem \ref{thm:sandwich_star} shows such filtering does not hurt the FDR control. 
\begin{thm}\label{thm:sandwich_star}
    For any rejection rule $\cR$ with $\cR^{*} \subseteq \cR \subseteq \cR^{\ckn^*}$ almost surely, we have $\FDR(\cR) \leq \alpha$.
\end{thm}
\begin{proof}
Recall
\[
\cR^{\ckn^*} \,=\, \cR^{\kn} \cup \{j:\; T_j \geq \hc_j^*\}.
\]
Hence $\cR^{*} \subseteq \cR \subseteq \cR^{\ckn^*}$ implies 
\[
\EE_{H_j}[\DP_j(\cR) \mid S_j]
\;\leq\; \EE_{H_j} \br{ \frac{\one \set{j \in \cR^\kn} \mam \one \set{T_j \geq \hc_j^*}}{\left| \cR^* \cup \{j\}\right|} \;\Big|\; S_j }\\[5pt]
\;\leq\; \EE_{H_j}[b_j \mid S_j],
\]
so that $\FDR(\cR) \leq \sum_{j\in \cH_0} \EE[b_j] \leq \alpha$.
\end{proof}
For cKnockoff$^*$, $\cS$ is set again by \eqref{eq:filter}. The condition $\cR^\kn \cup \cS \supseteq \cR^*$ is made true by the construction of $\cR^*$ in \App~\ref{app:R_star}.

\section{Refined cKnockoff procedure}
\label{app:R_star}

\subsection{Refined cKnockoff procedure}

Recall in the {\em refined calibrated knockoff} (cKnockoff$^*$) procedure, we find certain $\cR^*$ satisfying $\cR^\kn \subseteq \cR^* \subseteq \cR^\ckn$. Then cKnockoff$^*$ rejects
\[
\cR^{\ckn^*} \;=\; \cR^\kn \cup \left\{j:\; T_j \geq \hc_j^*\right\} \;=\; 
\cR^\kn \cup \left\{j:\; E_j^*(T_j; S_j) \leq 0\right\},
\]
where
\[
E_j^*(c\,; S_j) \;=\;\EE_{H_j} \br{\frac{\one \set{j \in \cR^\kn} \mam \one \set{T_j \geq c}}{\left| \cR^* \cup \{j\}\right|} - b_j \;\Big|\; S_j} \;\leq\; E_j(c\,; S_j),
\]
and $\hc_j^* = \min \set{c:\; E_j^*(c\,; S_j) \leq 0} \leq \hc_j$.

Before explicitly defining $\cR^*$ in the next section, we first show cKnockoff$^*$ controls FDR and is uniformly more powerful than cKnockoff.

\begin{thm} \label{thm:fdr_star}
    Assume $\cR^\kn \subseteq \cR^* \subseteq \cR^\ckn$, and the budgets $b_1,\ldots,b_m$ satisfy the two conditions in \eqref{eq:budget-conditions}. Then $\cR^{\ckn^*} \supseteq \cR^\ckn$, and $\cR^{\ckn^*}$ controls FDR at level $\alpha$.
\end{thm}

\begin{proof}
  Because $E_j^*(c\,; S_j) \leq E_j(c\,; S_j)$, we have 
  \[
  \cR^{\ckn^*}\supseteq \cR^{\ckn} \supseteq \cR^*.
  \]
  Recall $\hc_j^* = \min \set{c:\; E_j^*(c\,; S_j) \leq 0}$, so that $E_j^* \leq 0$ if and only if $T_j \geq \hat{c}_j^*$. Then we have
  \begin{align*}
  \EE_{H_j}[\DP_j(\cR^{\ckn^*}) \mid S_j]
  &\;\leq\; \EE_{H_j} \br{ \frac{\one \set{j \in \cR^{\ckn^*}}}{\left| \cR^{*} \cup \{j\}\right|} \;\Big|\; S_j }\\[5pt]
  &\;=\; \EE_{H_j} \br{ \frac{\one \set{j \in \cR^\kn} \mam \one \set{T_j \geq \hc_j^*}}{\left| \cR^* \cup \{j\}\right|} \;\Big|\; S_j }\\[5pt]
  &\;\leq\; \EE_{H_j}[b_j \mid S_j],
  \end{align*}
  so that $\FDR(\cR^{\ckn^*}) \leq \sum_{j\in \cH_0} \EE[b_j] \leq \alpha$.

\end{proof}

\subsection{Constructing $\cR^*$ in fixed-X cKnockoff} To make $\cR^*$ easy to compute while bringing in additional power, we need a delicate construction. Let
\[ \cR^* = \cR^\kn \cup \left\{j \in \cS^*:\; E_j(T_j; S_j) \leq 0 \right\}, \]
with
\[
\cS^* = \set{j \in \cS:\; p_j \leq \frac{\alpha}{m} \mim \frac{0.01 \alpha}{\lceil 1/\alpha - 1 \rceil},\, s_j^+ \text{ has rank no larger than } K^\text{cand}},
\]
where $\cS$ is the filtering set in \eqref{eq:filter}, $s_j^+$ is the promising score in \eqref{eq:score}, and $K^\text{cand}$ is a pre-specified constant to restrict the size of $\cS^*$. In words, we construct $\cS^*$ by picking the $K^\text{cand}$-most promising hypotheses in $\cS$ who have $p$-values below a certain threshold. We will explain the rationale behind this $p$-value cutoff in the next two sections. By construction, we have $\cR^\kn \subseteq \cR^* \subseteq \cR^\ckn$ and $\cR^\kn \cup \cS \supseteq \cR^*$.

\subsection{Computing $\cR^*$ in fixed-X cKnockoff}

Computing $\cR^*$ shares the same goal as computing $\cR^\ckn$. That is, we want to calculate
\[ E_j(c; S_j) \;=\; \int f_j(\vct z; c)\,dQ_j(\vct z) \]
with $c = T_j(\vct y)$.
Let's adopt the same narrative as in Section \ref{sec:implement}, regarding $S_j(\vct y)$, $c = T_j(\vct y)$, and $Q_j(\cdot \mid S_j(\vct y))$ as fixed and use $\vct z$ to denote a generic response vector drawn from the conditional null distribution $Q_j$.
But this time, we will use numerical integration instead of Monte-Carlo to compute $E_j$ approximately.

\subsubsection{Approximating $E_j$ by an integral}

To formulate the calculation as a numerical integration, recall \eqref{eq:z-decomp} that $\vct z$ can be written as
\[ \vct z = \vct z(\eta, \vct u) = \mat V_{-j} \mat V_{-j}^\tran \vct y + \vct v_j \cdot \eta + \sqrt{ \rho^2 - \eta^2} \cdot \mat V_{\text{res}} \vct u. \]
given $S_j$, where the two random variables $\eta = \vct v_j^\tran \vct z$ and $\vct u = \mat V_\text{res}^\tran \vct z / \norm{\mat V_\text{res}^\tran \vct z}$ are independent.

Following the idea in \App~\ref{app:tOmega}, we treat $f_j(\vct z(\eta, \vct u))$ as a random function of $\eta$ with an independent random noise $\vct u$. Then 
\[ 
    E_j \;=\; \int \EE_{\vct u \sim \text{Unif}(\bS^{n-m-1})} \br{ f_j(\vct z(\eta, \vct u))} \,d\, F_{\eta}(\eta),
\]
where $F_{\eta}(\cdot; S_j)$ is the CDF of $\eta$ given in \eqref{eq:cdf_eta}.

This motivates computing $E_j$ approximately by
\[
    E_j \approx \sum_{i=1}^{1/h} f_j(\vct z(\eta_i, \vct u_i)) \cdot h,
\]
where 
\[ \eta_i = F_{\eta}^{-1}(ih), \quad \vct u_i \simiid \text{Unif}(\bS^{n-m-1}) \]
In words, we numerically integrate over the variable $\eta$ and use Monte-Carlo for the leftover noise variable $\vct u$. When $h$ is tiny, we have a fine grid for the numerical integration and many samples for Monte-Carlo, and hence the calculation would be accurate.

\subsubsection{Checking if $E_j \leq 0$}

Since $f_j \leq 1$, we have
\[ 
    \int_{\Omega_j^{(T)}} f_j \, dQ_j
    \;\leq\; Q_j(\Omega_j^{(T)}) \coloneqq B_j^+.
\]
The bound $B_j^+$ is easy to compute, referring to \App~\ref{app:tOmega}.
Therefore, to conclude $E_j \leq 0$, it suffices to show
\[ B_j^+ + \int_{\pth{\Omega_j^{(T)}}^\setcomp} f_j \, dQ_j \leq 0. \]
We use our numerical scheme derived in the previous section to compute the integral. Recall \eqref{eq:Omega1},
\[ \pth{\Omega_j^{(T)}}^\setcomp = \set{\vct z(\eta, \vct u): \eta \in \pth{a_j^{(1)}, a_j^{(2)}}}. \]
It's easy to see that $\vct v_j^\tran \vct y$ is equal to either $a_j^{(1)}$ or $a_j^{(2)}$ since we have $T(\vct z) = T(\vct y)$ on the boundary of $\Omega_j^{(T)}$. Assuming $\vct v_j^\tran \vct y = a_j^{(1)}$ without loss of generality, we have
\[ \int_{\pth{\Omega_j^{(T)}}^\setcomp} f_j \, dQ_j \approx \sum_{i=1}^{(1-B_j^+)/h} f_j(\vct z(\eta_i, \vct u_i)) \cdot h, \]
where
\begin{equation}\label{eq:etai_ui}
\eta_i = F_\eta^{-1}(F_\eta(\vct v_j^\tran \vct y) + ih), \quad \vct u_i \simiid \text{Unif}(\bS^{n-m-1}).
\end{equation}

The key of our fast algorithm is to recall that $f_j(\vct z(\eta, \vct u)) \leq 0$ on $\pth{\Omega_j^{(T)}}^\setcomp$. That is to say, $f_j(\vct z(\eta_i, \vct u_i)) \leq 0$ for all $\eta_i$ and $\vct u_i$ considered in \eqref{eq:etai_ui}. Therefore, once we see 
\[ \sum_{i=1}^{k} f_j(\vct z(\eta_i, \vct u_i)) \cdot h \leq - B_j^+ \]
for some $k \leq (1-B_j^+)/h$, we can claim $E_j \leq 0$ without further calculation. Setting the starting point of the numerical integration nodes as $\eta_0 = \vct v_j^\tran \vct y$ rather than the other endpoint $a_j^{(2)}$ makes $f_j$ at $\vct z(\eta_1, \vct u_1), \, \vct z(\eta_2, \vct u_2), \ldots$ more likely to be nonzero, hence earlier to conclude $E_j \leq 0$, when $H_j$ is nonnull.

In particular, we try $k = 1, 2, \ldots, K^\text{step}$ sequentially, once we see $\sum_{i=1}^{k} f_j(\vct z(\eta_i, u_i)) \cdot h \leq - B_j^+$, we add $j$ into $\cR^*$ and stop the calculation; otherwise we still stop but do not include $j$. 

The overall running time of $\cR^*$ is upper bounded by $K^\text{cand} \cdot K^\text{step}$ times the time of evaluating the knockoff feature statistics. In our implementation, we set $K^\text{cand} =  K^\text{step} = 3.$

\subsubsection{Choosing a proper $h$} \label{app:R_star_h}

Now the only question that remains is how to choose a proper $h$. The choice of $h$ is driven by a tradeoff --- a smaller $h$ yields a more accurate estimate of $E_j$ and a larger $h$ allows us to determine if $E_j \leq 0$ within a few steps. 

There are only two cases for $\vct z \in \pth{\Omega_j^{(T)}}^\setcomp$ where $f_j(\vct z) \neq 0$: when $j \in \cR^\kn(\vct z)$ or $j \in \cC(w_{\tau_1})$. As we will see, the computational tricks introduced in \App~\ref{app:R_star_trick} imply the first case is uncommon when we need to compute $\cR^*$. So considering $j \in \cC(w_{\tau_1})$ but $j \not\in \cR^\kn(\vct z)$, we have
\[ f_j = -b_j(\alpha; \vct z) \approx - \frac{\alpha}{\lceil 1/\alpha - 1 \rceil} \]
since $|\cC(w_{\tau_1})| = \lceil 1/\alpha - 1 \rceil$ and $\hFDP(w_{\tau_1}) \lessapprox 1$. Therefore, if we set
\[ h = B_j^+ / \frac{\alpha}{\lceil 1/\alpha - 1 \rceil}, \]
we expect to see $\sum_{i=1}^{k} f_j(\vct z(\eta_i, \vct u_i)) \cdot h \leq - B_j^+$ with $k = 1$ or $2$, if $j$ is to be rejected.

Moreover, by noticing a monotone bijection between $\eta = \vct v_j^\tran \vct z$ and the $t$-statistic $t_j$,
\[
    t_j = \frac{\vct v_j^\tran \vct z}{\sqrt{(\|\vct y\|^2 - \|\mat \Pi_{-j}\vct y\|^2 - (\vct v_j^\tran \vct z)^2) / (n-m)}}, \quad 
    \vct v_j^\tran \vct z = t_j \cdot \sqrt{\frac{\|\vct y\|^2 - \|\mat \Pi_{-j}\vct y\|^2}{t_j^2 + n-m}},
\]
simple calculation shows
\[ Q_j\pth{\set{\vct v_j^\tran \vct z \not\in \pth{a_j^{(1)}, a_j^{(2)}}}} = B_j^+ \;\leq\; p_j = 2 F_{t_{n-m}}(-|t_j|). \]
Recall we have
\[ p_j \leq \frac{\alpha}{m} \mim \frac{0.01 \alpha}{\lceil 1/\alpha - 1 \rceil}, \quad \forall j \in \cS^*.\]
This guarantees
\[ h = B_j^+ / \frac{\alpha}{\lceil 1/\alpha - 1 \rceil} \leq 0.01, \]
which is reasonably small.

It is worth pointing out here that with calculating $\cR^*$ in this way, we do not have an estimation of the error in computing $E_j$, unlike in cKnockoff where we have a confidence sequence for this purpose. Although our current implementation of cKnockoff$^*$ works very well in simulations, analysts can choose a smaller $h$ or just stick to cKnockoff if they do not want to rely on a numerical integration without an error estimation.

\subsection{Computational tricks in applying $\cR^*$} \label{app:R_star_trick}

Now we have an algorithm to compute $\cR^*$ efficiently. When computing $\cR^{\ckn^*}$, we need to calculate $E_j^*(T_j; S_j)$. This is done in the same way as calculating $E_j(T_j; S_j)$ in cKnockoff, but replacing 
\[ f_j(\vct z_i) \;\coloneqq\; \frac{\one \set{j \in \cR^\kn(\vct z_i)} \,\mam\, \one \set{T_j(\vct z_i) \geq T_j(\vct y)}}{\left| \cR^\kn(\vct z_i) \cup \{j\}\right|} \,-\, b_j(\vct z_i) \]
by a smaller value
\[ f_j^*(\vct z_i) \;\coloneqq\; \frac{\one \set{j \in \cR^\kn(\vct z_i)} \,\mam\, \one \set{T_j(\vct z_i) \geq T_j(\vct y)}}{\left| \cR^*(\vct z_i) \cup \{j\}\right|} \,-\, b_j(\vct z_i) \]
for all Monte-Carlo samples $\vct z_1, \vct z_2, \cdots$.

Although computing $\cR^*$ is now affordable, it's still heavier compared to $\cR^\kn$. Hence we want to spend our computational budget where such a replacement is most likely to give us additional rejections. In particular, suppose
\[ \vct z_1,\, \vct z_2, \ldots, \vct z_k \in \tOmega_j \]
is the sequence of Monte-Carlo samples we use to compute $\tE_j^*$, where $\tE_j^*$ is the conservative approximation of $E_j^*$, defined in the same way as in \eqref{eq:Etilde}.
Our computational trick is to replace $f_j(\vct z_i)$ by $f_j^*(\vct z_i)$ only for those $i \in I$, where $I \subseteq [k]$ is a subset of the index set of all the Monte-Carlo samples. So if $|I|$ is small, we only need to compute $\cR^*$ for a few Monte-Carlo samples. Formally, let
\[ \hE_j(I) \coloneqq \frac{1}{k} \pth{\sum_{i \in I} f_j^*(\vct z_i) + \sum_{i \not\in I} f_j(\vct z_i)}. \]
Then $\hE_j(I)$ for $I = \emptyset$ is our Monte-Carlo calculation of $\tE_j$ and $\hE_j(I)$ for $I = [k]$ is the Monte-Carlo calculation of $\tE_j^*$. The trick is to use $\hE_j(I)$ for some $I \subseteq [k]$, instead of $\hE_j([k])$, as the computed approximation of $\tE_j^*$.

In principle, we want a small set $I$ that makes $\hE_j(\emptyset) - \hE_j(I)$ large. Then $\cR^*$ needs to be evaluated only on a small set of Monte-Carlo samples but still, it is very likely $\hE_j(I) \leq 0$ even if $\hE_j(\emptyset) > 0$. Specifically, the choice of set $I$ follows the steps below.
\begin{enumerate}
    \item Note if we want to compute $\hE_j(I)$, the value of $\hE_j(\emptyset)$ is an inevitable by-product. Hence it does not hurt to compute $\hE_j(\emptyset)$ before we decide $I$. Therefore, we set $I = \emptyset$ if $\hE_j(\emptyset) \leq 0$, since it already implies $\hE_j(I) \leq 0$.
    
    \item If the previous step does not set $I = \emptyset$, we propose to set $I = \bar I$ and trim it in the next step, where
    \[ \bar I \coloneqq \set{i:\; \frac{\one \set{j \in \cR^\kn(\vct z_i)} \,\mam\, \one \set{T_j(\vct z_i) \geq c}}{\left| \cR^\kn(\vct z_i) \cup \{j\}\right|} = 1}. \]
    In words, we include $i$ in $I$ only if $\cR^\kn(\vct z_i) = \emptyset$ and $T_j(\vct z_i) \geq c$ for $c = T_j(\vct y)$. The idea behinds is that, in this case, including $i$ in $I$ may introduce the largest decline in the value of $\hE_j(I)$.
    
    Since $\cR^\kn(\vct z_i) = \emptyset$ implies that Knockoff is having a hard time making any rejections, it's uncommon to see $j \in \cR^\kn$ in computing $\cR^*$, as mentioned in \App~\ref{app:R_star_h}.
    
    \item Following step 2, if the denominators $|\cR^*(\vct z_i) \cup \set{j}|$ take the same value for all $i$, then we can easily solve for the desired smallest value of them, denoted as $R^*$, so as to make $\hE_j(\bar I)$ below zero given $\hE_j(\emptyset) > 0$. We set $R^* = \infty$ if we cannot have $\hE_j(\bar I) \leq 0$.
    
    If $R^* > K^\text{cand}+1$, it's impossible that the computed $\cR^*$ can be large enough to make $H_j$ rejected, since the computed $|\cR^*(\vct z_i) \cup \set{j}| \leq K^\text{cand}+1$ by construction. If this is the case, we trim all elements in the proposed $I = \bar I$. That is, we set $I = \emptyset$.
    
    If $R^* \leq K^\text{cand}+1$, we trim $I$ in an online manner. That is, we compute $\cR^*(\vct z_i)$ one by one for $i \in I$. Once we see most $\cR^*(\vct z_i)$ computed so far have $|\cR^*(\vct z_i) \cup \set{j}| < R^*$, we trim the rest elements in $I$. 
\end{enumerate}
Details of the implementation are available at \texttt{github.com/yixiangLuo/cknockoff}.

Note this computational trick does not hurt the FDR control of cKnockoff$^*$ no matter how we choose set $I$. Since $\hE_j(I)$ is in between the Monte-Carlo calculation of $\tE_j$ and $\tE_j^*$, the resulting rejection set is in between $\cR^\ckn$ and the vanilla $\cR^{\ckn^*}$ without this trick. Theorem \ref{thm:sandwich_star} then ensures its FDR control.

\section{Deferred proofs}

\subsection{Formulating MCC as a linear model}
\label{app:mcc-details}

In this section, we formulate Example \ref{ex:mcc-block} as a linear model. Since all blocks in the experiment are mutually independent, it suffices to show this formulation under $K=1$, the classical MCC problem.

Let $z_{g,i}$ be the observed outcome for patient $i$ in treatment group $g$ and $\ep_{g,i}$ be the corresponding $\cN(0,\sigma^2)$ error. We use $g = 0$ to denote the control group and $g = 1, \ldots, G$ for treatment groups.
Denote $\vct z_i = (z_{0,i}, z_{1,i}, \ldots, z_{G,i})^\tran \in \RR^{G+1}$ and $\vct \ep_i = (\ep_{0,i}, \ep_{1,i}, \ldots, \ep_{G,i})^\tran \in \RR^{G+1}$. Define
\[
\check{\mat I}_G = 
\begin{pmatrix}
\vct 0^\tran \\
\mat I_G
\end{pmatrix} =
\begin{pmatrix}
0 & \ldots & 0 \\
1 &  &  \\
 & \ddots & \\
 &  & 1 \\
\end{pmatrix} 
\in \RR^{(G+1) \times G}
\]
as the matrix obtained by padding the identity matrix of dimension $G$ on top with a row vector of all zeros. Then define
\[ 
\vct z = 
\begin{pmatrix}
\vct z_1 \\
\vct z_2 \\
\vdots \\
\vct z_r \\
\end{pmatrix} \in \RR^{r(G+1)}, \quad
\check{\mat X} = 
\begin{pmatrix}
\check{\mat I}_G \\
\check{\mat I}_G \\
\vdots \\
\check{\mat I}_G \\
\end{pmatrix} \in \RR^{r(G+1) \times G},
\]
\[
\vct \beta = 
\begin{pmatrix}
\beta_1 \\
\vdots \\
\beta_G \\
\end{pmatrix} \in \RR^{G}, \quad
\check{\vct \ep} = 
\begin{pmatrix}
\vct \ep_1 \\
\vct \ep_2 \\
\vdots \\
\vct \ep_r \\
\end{pmatrix} \in \RR^{r(G+1)}.
\]
We have
\[ \vct z = \check{\mat X} \vct \beta + \vct 1 \mu + \check{\vct \ep}, \]
where $\vct 1 = (1, 1, \ldots, 1)^\tran \in \RR^{r(G+1)}$ is a vector of all ones. This is a linear model with an intercept term. So we can project everything onto the subspace orthogonal to $\vct 1$ to get rid of the intercept. Specifically, let $\mat V_{1,\text{res}} \in \RR^{r(G+1) \times (r(G+1)-1)}$ be an orthonormal basis for the subspace orthogonal to $\vct 1$. Define
\[ 
\vct y \coloneqq \mat V_{1,\text{res}}^\tran \vct z, \quad
\mat X \coloneqq\mat V_{1,\text{res}}^\tran \check{\mat X}, \quad
\vct \ep \coloneqq \mat V_{1,\text{res}}^\tran \check{\vct \ep}.
\]
We have
\[ \vct y = \mat X \vct \beta + \vct \ep \]
follows the Gaussian linear model \eqref{eq:linear-model} with $m = G$ and $n = r(G+1)-1$.

More generally, if $K > 1$, we do the same thing for each experiment block and concatenate the linear models derived from them. Specifically, the resulting $\mat X \in \RR^{K(r(G+1)-1) \times KG}$ for $K > 1$ is block-diagonal with $K$ blocks of dimension $(r(G+1)-1)$-by-$G$. Each block of $\mat X$ comes from the same procedure as above applying to each block of the experiments.

In our simulations of Section~\ref{sec:simu_main}, to maintain consistent dimensions $n=3000, m=1000$ as we used for the i.i.d. Gaussian case, we slightly modify the linear model construction by removing a few extra residual degrees of freedom, while maintaining the same correlation structure $\mat X^\tran \mat X$ for the explanatory variables. For example, in the MCC problem with $m=G=1000$, $K=1$, and $r=3$, the canonical construction above would give $n = r(G+1)-1 = 3002$. In our simulation, we remove the extra $2$ residual degrees of freedom, while maintaining the correlation structure of $\mat X^\tran \mat X$, thus ensuring that $\hat{\vct \beta}$ has the same positively equicorrelated covariance structure, and $[\mat X \mat \tX]^\tran \vct y$ has the same distribution when we make an equivalent choice of $\mat \tX$. However, the distribution of the residual variance estimators $\hat\sigma^2$ and $\tilde\sigma^2$ both change slightly because they are based on $2002$ and $1002$ residual degrees of freedom, respectively, instead of $2000$ and $1000$ residual degrees of freedom. A direct comparison confirms the results have no observable difference. For consistency, we also use the construction with $n=3000$ for the simulation in Figure~\ref{fig:mcc-block}.


\subsection{Null distribution conditional on $S_j$}
\label{app:cond_dist}

In the model-X setting, the null conditional distribution of $(\mat X, \vct y)$ given $S_j = (\mat X_{-j}, \vct y)$ is trivial. We focus on the fixed-X setting next.

Recall we define
\[ \mat \Pi_{-j} = \mat X_{-j}(\mat X_{-j}^{\tran}\mat X_{-j})^{-1}\mat X_{-j}^\tran, \quad \mat \Pi_{-j}^\perp = \mat I - \mat \Pi_{-j} \]
as the projection onto the column span of $\mat X_{-j}$ and its orthogonal projection, respectively. The bijection between $S_j = (\mat X_{-j}^\tran \vct y,\, \norm{\vct y}^2)$ and $(\mat \Pi_{-j} \vct y,\, \norm{\mat \Pi_{-j}^\perp \vct y}^2)$
\[ \mat \Pi_{-j} \vct y \;=\; (\mat X_{-j}(\mat X_{-j}^{\tran}\mat X_{-j})^{-1}) \cdot (\mat X_{-j}^\tran \vct y), \quad
\norm{\mat \Pi_{-j}^\perp \vct y}^2 \;=\; \norm{\vct y}^2 - \norm{\mat \Pi_{-j} \vct y}^2 \]
\[ \mat X_{-j}^\tran \vct y \;=\; \mat X_{-j}^{\tran} (\mat \Pi_{-j} \vct y), \quad
\norm{\vct y}^2 \;=\; \norm{\mat \Pi_{-j}^\perp \vct y}^2 + \norm{\mat \Pi_{-j} \vct y}^2\]
shows that conditioning on $S_j$ is equivalent to conditioning on $(\mat \Pi_{-j} \vct y,\, \norm{\mat \Pi_{-j}^\perp \vct y}^2)$. Then we have the following null conditional distribution.

\begin{prop} \label{prop:cond_dist}
    Assume the linear model \eqref{eq:linear-model} and that $H_j$ is true. Then
    \[
         \vct y \mid \mat \Pi_{-j} \vct y,\, \norm{\mat \Pi_{-j}^\perp \vct y}^2
         \;\eqd\; \mat \Pi_{-j} \vct y + \norm{\mat \Pi_{-j}^\perp \vct y} \cdot \mat V_{-j, \text{res}} \vct u,
    \]
    where $\mat V_{-j, \text{res}} \in \RR^{n \times (n-m+1)}$ is an orthonormal basis for the subspace orthogonal to the span of $\mat X_{-j}$, and $\vct u \sim \text{Unif}(\bS^{n-m})$ is a vector uniformly distributed on the unit sphere of dimension $n-m$.
\end{prop}

\begin{proof}

Since $\vct y \sim \cN(\mat X \vct \beta, \sigma^2 \mat I_n)$ is isotropic Gaussian, $\mat \Pi_{-j} \vct y$ and $\mat \Pi_{-j}^\perp \vct y$ are independent. So it suffices to show
\[ \mat V_{-j, \text{res}}^\tran \cdot \mat \Pi_{-j}^\perp \vct y \mid \norm{\mat \Pi_{-j}^\perp \vct y}^2 \;\eqd\; \norm{\mat \Pi_{-j}^\perp \vct y} \cdot \text{Unif}(\bS^{n-m}). \]
This is true since
\[ \norm{\mat V_{-j, \text{res}}^\tran \cdot \mat \Pi_{-j}^\perp \vct y} = \norm{\mat \Pi_{-j}^\perp \vct y} \]
and
\[ \mat V_{-j, \text{res}}^\tran \cdot \mat \Pi_{-j}^\perp \vct y \;\sim\; \cN(\vct 0,\; \sigma^2 \mat I_{n-m+1}). \]
The second claim is by
\[ 
    \mat \Pi_{-j}^\perp \vct y
    \;\sim\; \cN(\mat \Pi_{-j}^\perp X \beta,\; \sigma^2 \mat \Pi_{-j}^\perp \mat I_n)
    \;\overset{H_j}{\sim}\; \cN(\vct 0,\; \sigma^2 \mat \Pi_{-j}^\perp)
\]
and noticing
\[ \mat \Pi_{-j}^\perp = \mat V_{-j, \text{res}} \mat V_{-j, \text{res}}^\tran. \]

\end{proof}

\subsection{Proof of Theorem \ref{thm:strict_better}}
\label{app:strict_better}

The proof of Theorem \ref{thm:strict_better} first needs a technical lemma.
\begin{lemma} \label{lemma:strict_better}
    Let the budget be defined as in \eqref{eq:bj_def}. Suppose $\PP_{H_j}(b_j > \DP_j(\cR^\kn) \mid S_j) > 0$ and $\PP_{H_j}(j \not\in \cR^\kn \mid S_j) > 0$, then
    \[ \EE_{H_j} \br{ \one \set{j \in \cR^\kn} \;\Big|\; S_j } \;<\; \EE_{H_j} \br{ \one \set{j \in \cR^\kn} \mam \one \set{T_j \geq \hc_j} \;\Big|\; S_j }. \]
\end{lemma}
\begin{proof}
Recall
\[ E_j(c\,; S_j) \;\coloneqq\; \EE_{H_j} \br{\frac{\one \set{j \in \cR^\kn} \mam \one \set{T_j \geq c}}{\left| \cR^\kn \cup \{j\}\right|} \,-\, b_j \;\Big|\; S_j}. \]
Since $b_j \geq \DP_j(\cR^\kn)$ almost surely and $\PP_{H_j}(b_j > \DP_j(\cR^\kn) \mid S_j) > 0$, we have
\[ E_j(\infty\,; S_j) < 0. \]
Moreover, $\PP_{H_j}(j \not\in \cR^\kn \mid S_j) > 0$ gives
\[ E_j(\infty\,; S_j) < E_j(0\,; S_j). \]
Recall $E_j(c\,; S_j)$ is a continuous, non-increasing function of $c$ and
\[ \hc_j = \min_{c \geq 0} \set{E_j(c\,; S_j) \leq 0}. \]
We have
\[ E_j(\infty\,; S_j) < E_j(\hc_j\,; S_j). \]
That is
\[
\EE_{H_j} \br{ \frac{\one \set{j \in \cR^\kn}}{\left| \cR^\kn \cup \{j\}\right|} \;\Big|\; S_j }
\;<\;
\EE_{H_j} \br{ \frac{\one \set{j \in \cR^\kn} \mam \one \set{T_j \geq \hc_j}}{\left| \cR^\kn \cup \{j\}\right|} \;\Big|\; S_j }.
\]

Now we prove our claim by contradiction. Note
\[ \one \set{j \in \cR^\kn} \leq \one \set{j \in \cR^\kn} \mam \one \set{T_j \geq \hc_j}, \quad \text{almost surely.} \]
So the opposite of our proposition implies
\[ \PP_{H_j} \br{ \one \set{j \in \cR^\kn} = \one \set{j \in \cR^\kn} \mam \one \set{T_j \geq \hc_j} \;\Big|\; S_j } = 1, \]
which further yields
\[
\EE_{H_j} \br{ \frac{\one \set{j \in \cR^\kn}}{\left| \cR^\kn \cup \{j\}\right|} \;\Big|\; S_j }
\;=\;
\EE_{H_j} \br{ \frac{\one \set{j \in \cR^\kn} \mam \one \set{T_j \geq \hc_j}}{\left| \cR^\kn \cup \{j\}\right|} \;\Big|\; S_j }
\]
since the denominators of both sides are the same. This contradicts what we have derived.
\end{proof}

Then we can prove Theorem \ref{thm:strict_better}.
\begin{proof}[Proof of Theorem \ref{thm:strict_better}]
For a general testing procedure that produces rejection set $\cR(\vct y)$, define its {\em $H_j$-rejection region} as
\[ \cG_j = \set{\vct y: \; j \in \cR(\vct y)}. \]
In words, $\cG_j$ is the set of observed data $\vct y$ such that $H_j$ is rejected. Note $G_j$ is determined by the testing procedure $\cR$ itself and is not affected by the unknown value of $\vct \beta$ or $\sigma$.

By construction, the cKnockoff rejection region is always no smaller than the knockoffs, i.e. $\cG^\ckn_j \supseteq \cG^\kn_j$ for all $j$. So 
\[ 
    \TPR(\cR^\ckn) 
    \;=\; \frac{1}{m_1} \sum_{j \in \cH_0^\setcomp} \PP(j \in \cG^\ckn_j)
    \;\geq\; \frac{1}{m_1} \sum_{j \in \cH_0^\setcomp} \PP(j \in \cG^\kn_j)
    \;=\; \TPR(\cR^\kn). 
\]
To show that inequality is strict, it suffices to prove
\[ \cG^\kn_j \subsetneq \cG^\ckn_j \]
for some $j \in \cH_0^\setcomp$, because the support of the density of $\vct y$ is the whole $\RR^n$ space no matter what $\beta$ is.
For the same reason, this is equivalent to show
\[ \PP_0(j \in \cR^\kn) = \PP_0(\cG^\kn_j) \;<\; \PP_0(\cG^\ckn_j) = \PP_0(j \in \cR^\ckn), \]
where $\PP_0$ is the probability measure under the global null model $\cH_0 = [m]$.

Recall that under the global null and conditional on $|\vct W|$, $\sgn(W_i)$ are independent Bernoulli for all $i$. Hence $\PP_0(A \mid |\vct W|) > 0$, where
\[ A = \set{W_j > 0 \text{ and } |\set{i:\; W_i > 0}| = (\ceil{1/\alpha}-1) \mim m }. \]
Note $A$ implies $\cR^\kn = \emptyset$ (hence $\DP_j(\cR^\kn) = 0$) and $b_j > 0$. We have 
\[ \PP_0 \pth{\set{b_j > \DP_j(\cR^\kn)} \cap \set{j \not\in\cR^\kn}} > 0 \]
by the tower property.

As a consequence, there exist a set $\cC \subseteq \RR^{m-1} \times \RR$ such that $\PP_0(S_j(\vct y) \in \cC) > 0$ and 
\[ \PP_0 \pth{\set{b_j > \DP_j(\cR^\kn)} \cap \set{j \not\in\cR^\kn} \mid S_j} > 0 \]
for all $S_j \in \cC$. Then by Lemma~\ref{lemma:strict_better}, for any $S_j \in \cC$,
\[ \EE_0 \br{ \one \set{j \in \cR^\kn} \;\Big|\; S_j } \;<\; \EE_0 \br{ \one \set{j \in \cR^\kn} \mam \one \set{T_j \geq \hc_j} \;\Big|\; S_j }, \]
where $\EE_0$ is taking an expectation over the global null distribution.
So
\[ \EE_0 \br{ \one \set{j \in \cR^\kn} \;\Big|\; S_j \in \cC } \;<\; \EE_0 \br{ \one \set{j \in \cR^\kn} \mam \one \set{T_j \geq \hc_j} \;\Big|\; S_j \in \cC }. \]
Note
\[ \one \set{j \in \cR^\kn} \leq \one \set{j \in \cR^\kn} \mam \one \set{T_j \geq \hc_j} \]
and $\PP_0(S_j \in \cC) > 0$. We have
\[ \EE_0 \br{ \one \set{j \in \cR^\kn} } \;<\; \EE_0 \br{ \one \set{j \in \cR^\kn} \mam \one \set{T_j \geq \hc_j} }. \]
That is
\[ \PP_0(j \in \cR^\kn) < \PP_0(j \in \cR^\ckn). \]
\end{proof}

\section{Numerical simulations}
\label{app:simu}

The simulation settings in this section are the same as in Section \ref{sec:simu_main} if not specified.

\subsection{Extensions of simulations in Section \ref{sec:simu_main}}

\subsubsection{Additional design matrix settings}

Consider the following two design matrix settings.
\begin{enumerate}
    \item \textbf{OLS $\hat \beta_j$ positively auto-regression (Coef-AR)}: Set $\mat X$ such that the OLS fitted $\hat \beta_j$ is AR(1) with $\text{cov}(\hat \beta_j, \hat \beta_{j+1}) = 0.5$.
    \item \textbf{$\vct X_j$ positively auto-regression (X-AR)}: Set $\mat X$ such that $\vct X_j$ ($j = 1, \ldots, m)$ is a AR(1) random process with $\text{cov}(\vct X_j, \vct X_{j+1}) = 0.5$.
\end{enumerate}

Figure \ref{fig:main_expr_10_app} shows the results. They are similar to the ones from the MCC-Block problem.


\begin{figure}[!tb]
    \centering
    \includegraphics[width = 0.8\linewidth]{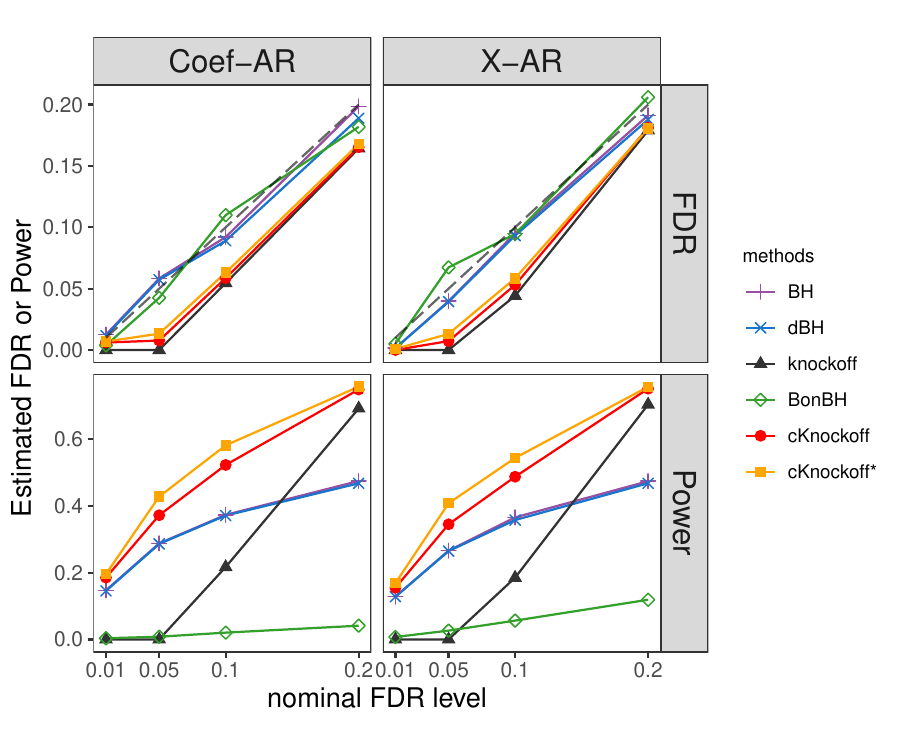}
    \caption{Estimated FDR and TPR for additional design matrix settings.}
    \label{fig:main_expr_10_app}
\end{figure}

\subsubsection{Cases where $m_1 \gg 1/\alpha$}
\label{app:less_sparse}

We show the performance of the procedures in the case where $\vct \beta$ is not too sparse. Figure \ref{fig:main_expr_10_app} shows the results with $m_1 = 30$ non-null hypotheses, instead of $m_1 = 10$ in Section~\ref{sec:simu_main}. The general behavior of the procedures remains the same, but
\begin{enumerate}
    \item the power improvement of cKnockoff over knockoffs is smaller;
    \item the knockoff-like methods are less powerful compared to the BH-like methods in the MCC problem.
\end{enumerate}
\begin{figure}[!tb]
    \centering
    \includegraphics[width = 0.95\linewidth]{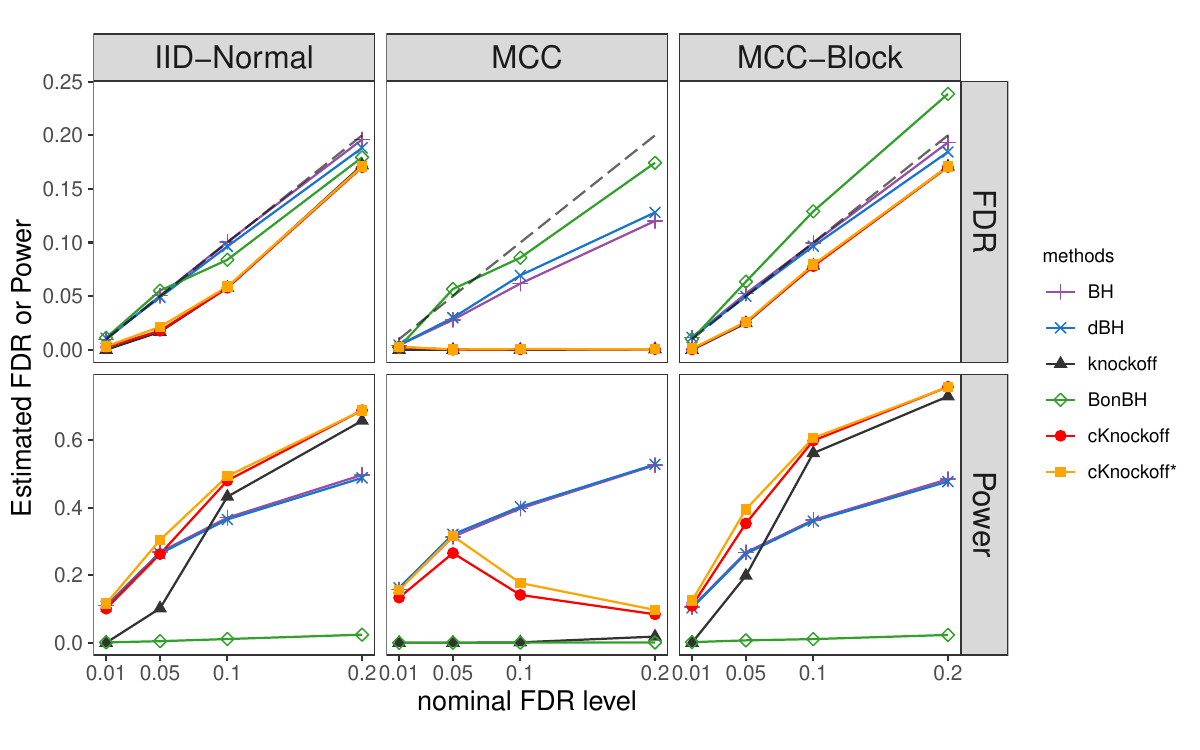}
    \caption{Estimated FDR and TPR with $30$ non-null hypotheses.}
    \label{fig:main_expr_30}
\end{figure}

\subsubsection{Comparison with multiple knockoffs}
\label{app:multi_kn}

Multiple knockoffs has more implementation choices than the vanilla knockoffs, e.g. the number of knockoff variables to employ and the threshold for deciding whether the original variable ``wins'' the competition with its knockoffs. There is no known best choice for them, but we implement multiple knockoffs as follows:
\begin{enumerate}
    \item We employ $5$-multiple knockoffs. That is, we generate $5$ knockoff matrices $\mat \tX_{(i)}$ for $i = 1, \ldots, 5$ such that $\mat X^\tran \mat X = \mat \tX_{(i)}^\tran \mat \tX_{(i)}$ for all $i$ and $\mat X^\tran \mat X - \mat \tX^\tran \mat \tX_{(j)} = \mat X^\tran \mat X - \mat \tX_{(i)}^\tran \mat \tX_{(j)}$ is certain diagonal matrix for all $i \neq j$.
    \item For the feature statistics, we run lasso on the augmented model
    \[ \vct y = [\mat X, \mat \tX_{(1)}, \ldots, \mat \tX_{(5)}] \vct \beta + \vct \ep \]
    with regularity parameter $\lambda$ determined in the same way as in LCD-T.
    And let
    \[ 
    \sgn(W_j) \coloneqq 
    \begin{cases}
    1 & \text{ if } |\hat \beta_{j,(0)}| > |\hat \beta_{j,(i)}| \text{ for all } i \in [5] \\
    -1 & \text{ otherwise }
    \end{cases}, \quad\quad
    |W_j| \coloneqq \max_{i = 0,\ldots, 5} \set{|\hat \beta_{j,(i)}|},
    \]
    where $\hat \beta_{j,(i)}$ is the fitted lasso coefficient of the $i$th knockoffs of the variable $\vct X_j$ and $\hat \beta_{j,(0)}$ is the fitted lasso coefficient of the original variable $\vct X_j$.
\end{enumerate}

To make multiple knockoffs applicable, we set $n = 7m$ with $m = 300$. The number of hypotheses is smaller than our usual setting, so as to save memory space in our laptop. The number of non-null hypotheses is set to be $m_1 = 30$ so as to show the performance of multiple knockoffs in the region where $m \gg 1/\alpha$.

Figure \ref{fig:multiple_kn} shows the results. When $m_1 < 1/\alpha$, multiple knockoffs relieves the threshold phenomenon but performs worse than cKnockoff/cKnockoff$^*$; when $m_1 \gg 1/\alpha$, multiple knockoffs is even worse than the vanilla knockoffs. Moreover, multiple knockoffs does not help the whiteout phenomenon.
\begin{figure}[!tb]
    \centering
    \includegraphics[width = 0.95\linewidth]{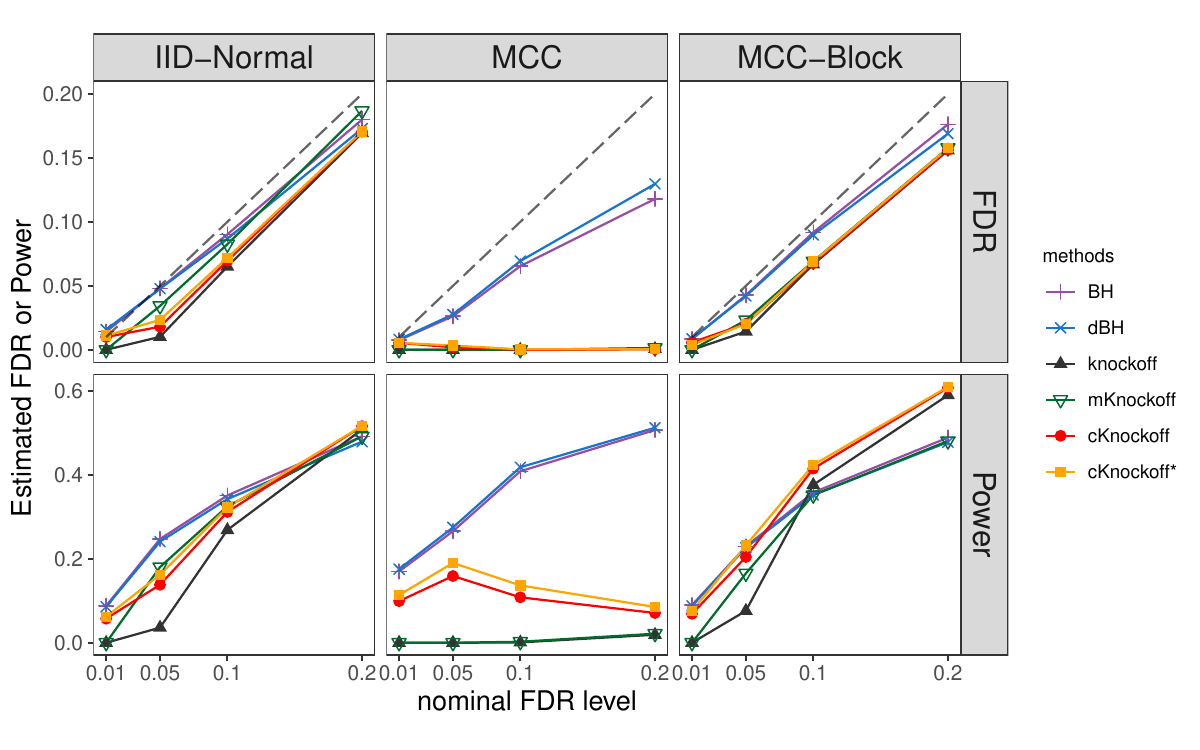}
    \caption{Estimated FDR and TPR with multiple knockoffs included.}
    \label{fig:multiple_kn}
\end{figure}

In addition, we see the knockoff-like methods are even less powerful, compared to the BH-like methods, in Figure~\ref{fig:multiple_kn} than in Figure~\ref{fig:main_expr_30}. This is because $\pi_1 = m_1 / m = 0.1$ in this setting. That is, $\vct \beta$ is less sparse.

\subsection{Variations of cKnockoff}



Figure \ref{fig:lambdasmax_stat} shows the estimated FDR and power when we use C-LSM rather than LCD-T as the feature statistics. The behavior of our methods is almost the same as in Figure \ref{fig:main_expr_10}.

\begin{figure}[!tb]
    \centering
    \includegraphics[width = 0.95\linewidth]{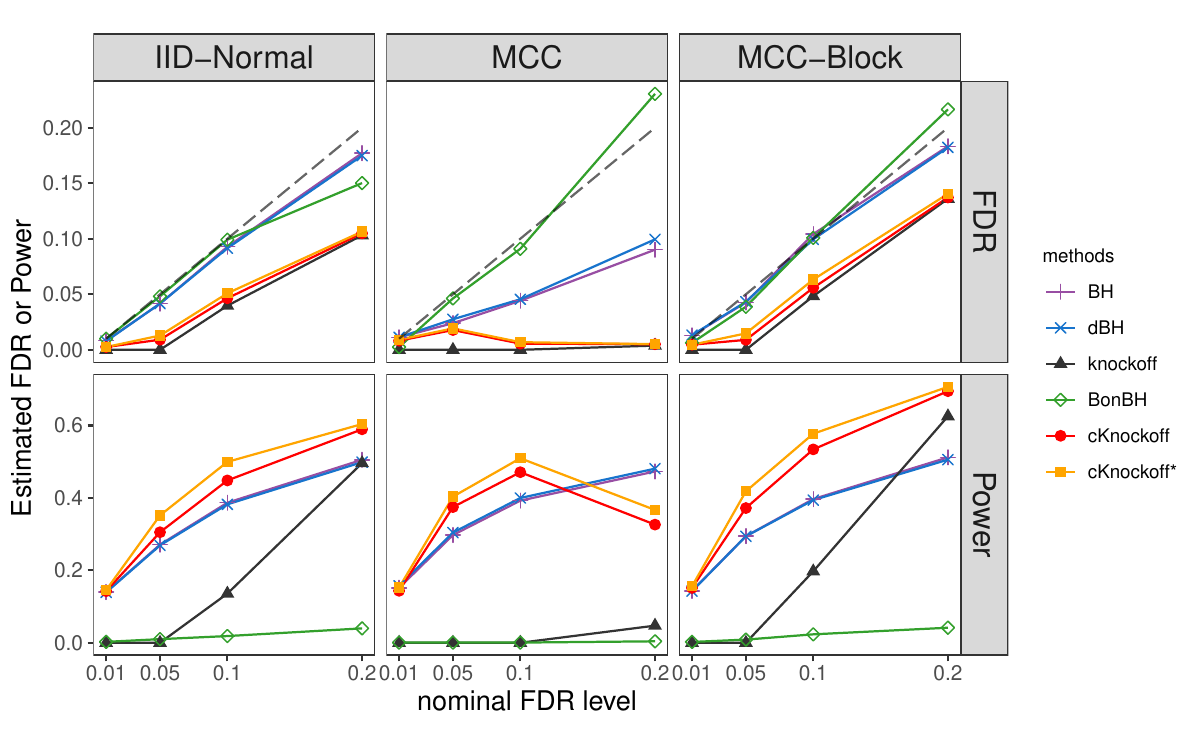}
    \caption{Estimated FDR and TPR under settings using C-LSM as the feature statistics.}
    \label{fig:lambdasmax_stat}
\end{figure}




\subsection{Distributions of FDP and TPP}

 Figure~\ref{fig:main_expr_10-FDP-ECDF} and \ref{fig:main_expr_10-TPP-ECDF} show the empirical cumulative distribution functions (ECDF) of the FDP and TPP, respectively, of selected procedures, under the same setting as the experiments in Section~\ref{sec:simu_main}.

For FDP, we see that knockoffs, cKnockoff, and cKnockoff$^*$ all work well, in the sense that the FDP is smaller than $\alpha$ with high probability in most cases. By contrast, BH and dBH both have stochastically larger FDP in all these cases even when their power is lower. In particular, in the MCC problem, the FDP distribution of BH/dBH puts large mass at both $0$ and $1$, rendering the rejection sets less reliable.

\begin{figure}[!tb]
    \centering
    \includegraphics[width = 0.99\linewidth]{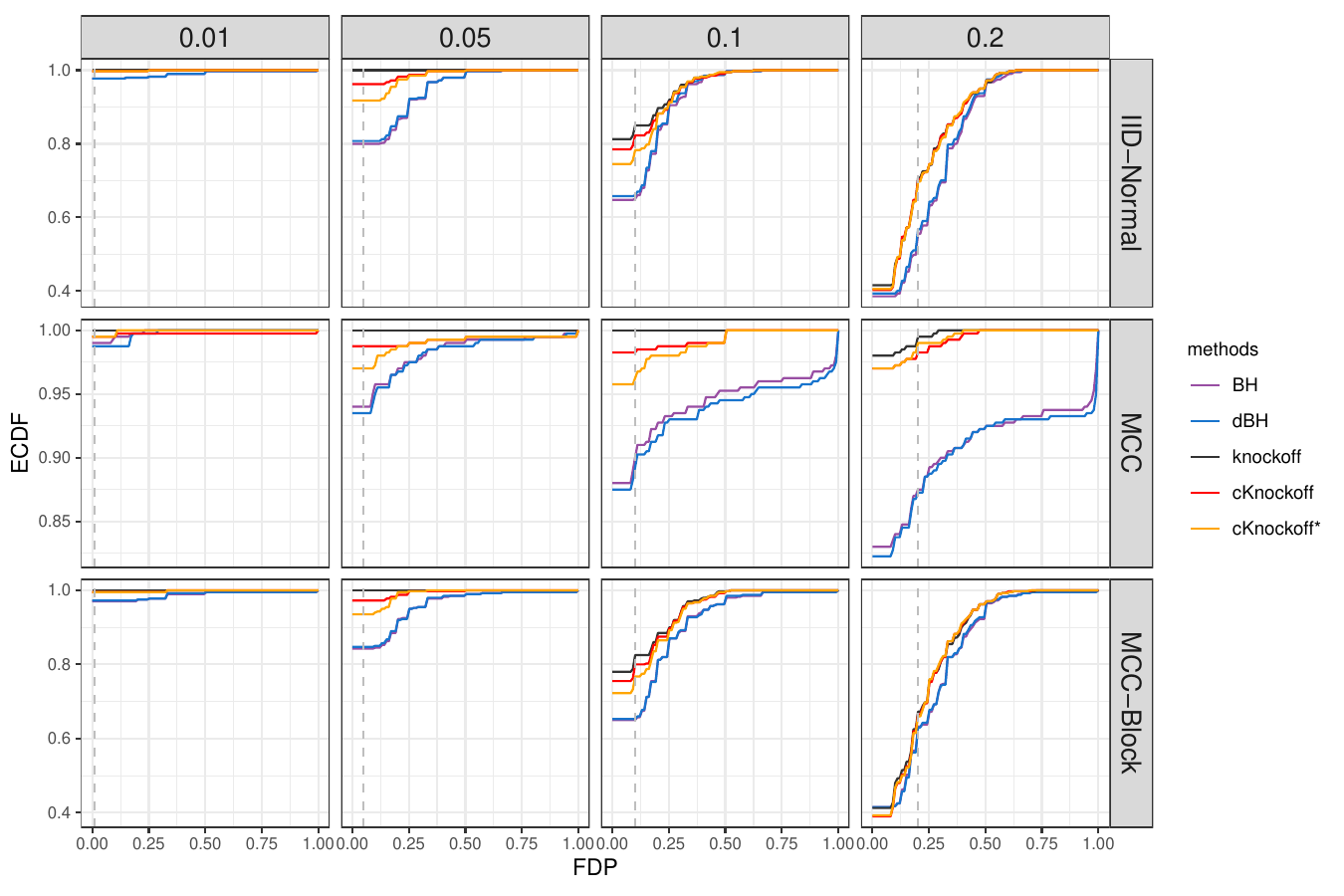}
    \caption{Empirical CDF of FDP from different procedures. Different columns represent different nominal FDR level $\alpha$, whose value is also indicated by a vertical dashed line. }
    \label{fig:main_expr_10-FDP-ECDF}
\end{figure}

For TPP, all procedures perform similarly and the TPP is not concentrated at the TPR. Moreover, when knockoffs make some rejections, cKnockoff/cKnockoff$^*$ makes a bit more; and when knockoffs fails to reject anything (a flat TPP CDF towards $\text{TPP}=0$), the CDF of cKnockoff/cKnockoff$^*$ keeps its trend. This indicates that our methods fully unleash the potential power of the knockoffs when it suffers. 

\begin{figure}[!tb]
    \centering
    \includegraphics[width = 0.99\linewidth]{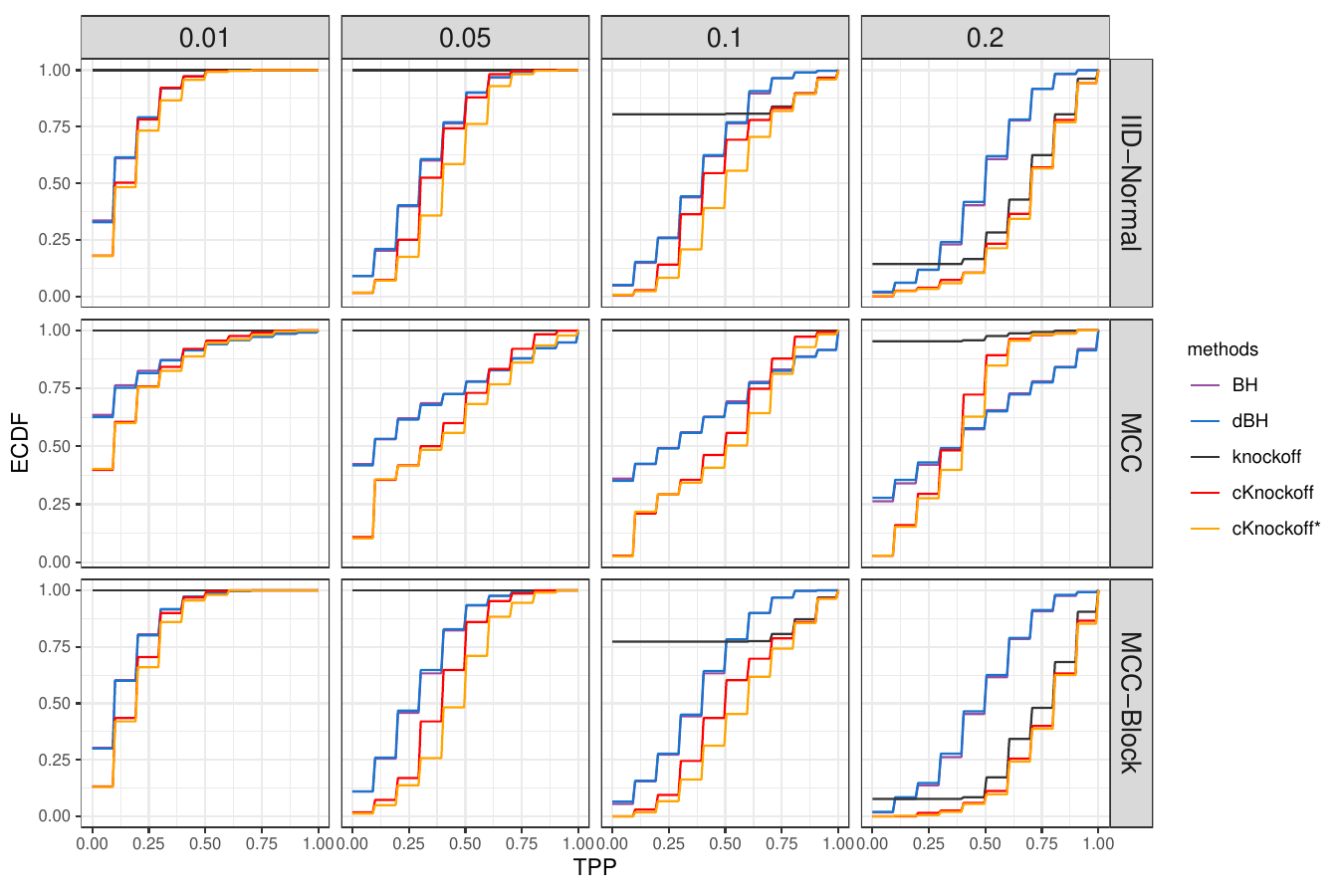}
    \caption{Empirical CDF of TPP from different procedures. Different columns represent different nominal FDR level $\alpha$. }
    \label{fig:main_expr_10-TPP-ECDF}
\end{figure}





\section{HIV drug resistance data}\label{app:hiv}

Continuing the discussion in Section \ref{sec:hiv}, Figure \ref{fig:HIV-0.2} presents the results where we set $\alpha = 0.2$.
\footnote{Readers may have noticed that the rejections made by knockoffs shown here are not exactly the same as the ones shown in \citet{fithian2022conditional} or \citet{barber15}. This is because that knockoff is implemented as a random method in their R package. In particular, they randomly swap $\vct X_j$ and $\tX_j$ to protect the FDR control from the bias that Lasso, implemented in the glmnet R package, prefers to select a feature with a smaller index. To avoid the interference of such random noise, the results we show are averaged over $20$ times applying each procedure.}
The knockoffs suffers in experiments like the drug AZT and D4T. Our cKnockoff and cKnockoff$^*$ procedures save knockoffs in these cases and make strictly more discoveries than knockoffs in almost all cases.

\begin{figure}[!tb]
    \centering
    \includegraphics[width = 0.75\linewidth]{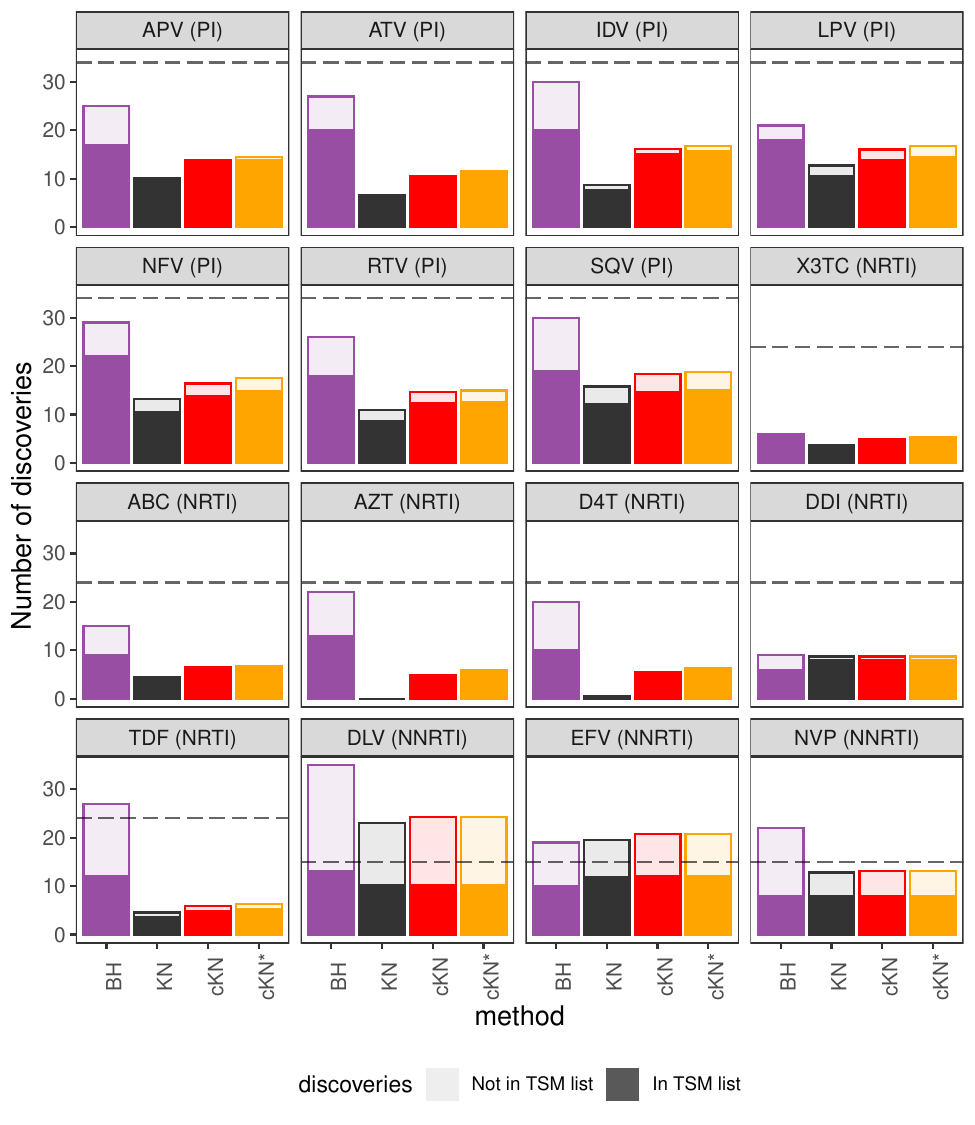}
    \caption{Results on the HIV drug resistance data with $\alpha = 0.2$. The darker segments represent the number of discoveries that were replicated in the TSM panel, while the lighter segments represent the number that was not. The horizontal dashed line indicates the total number of mutations appearing in the TSM list. Results are shown for the BH, fixed-$X$ knockoffs, cKnockoff and cKnockoff$^*$.}
    \label{fig:HIV-0.2}
\end{figure}

\end{document}